\documentclass[10pt,doublesided,doublecolumn]{IEEEtran} %Double column TSP
\usepackage{graphicx,epsfig,amssymb,bm}
\usepackage{amssymb,balance}
\usepackage{amsmath}
\usepackage{amsthm} 
\usepackage{tikz}
\usetikzlibrary{calc}
\usepackage{pgfplots}
\usetikzlibrary{plotmarks}
\usepackage{amssymb}
\usetikzlibrary{decorations.shapes} % LATEX and plain TEX when using Tik Z

\newtheorem{lemma}{Lemma}%[section]

\newtheorem{theorem}{Theorem}
\newtheorem{algorithm}{Algorithm}
\def\LSB{\left[}        %Left square bracketontrT)
\def\RSB{\right]}       %Right square bracket
\def\LB{\left(}         %Left bracket
\def\RB{\right)}        %Right bracket

\newfont{\bbb}{msbm10 scaled 500}

\newfont{\bb}{msbm10 scaled 1100}
\newcommand{\CC}{\mbox{\bb C}}
\newcommand{\RR}{\mbox{\bb R}}

\newcommand{\fv}{{\bf f}}
\newcommand{\gv}{{\bf g}}
\newcommand{\hv}{{\bf h}}

\newcommand{\mv}{{\bf m}}
\newcommand{\nv}{{\bf n}}

\newcommand{\qv}{{\bf q}}

\newcommand{\wv}{{\bf w}}
\newcommand{\vv}{{\bf v}}
\newcommand{\xv}{{\bf x}}
\newcommand{\yv}{{\bf y}}
\newcommand{\zv}{{\bf z}}

\newcommand{\Am}{{\bf A}}
\newcommand{\Bm}{{\bf B}}

\newcommand{\Fm}{{\bf F}}

\newcommand{\Hm}{{\bf H}}
\newcommand{\Id}{{\bf I}}

\newcommand{\Mm}{{\bf M}}
\newcommand{\Nm}{{\bf N}}

\newcommand{\Qm}{{\bf Q}}
\newcommand{\Rm}{{\bf R}}

\newcommand{\Tm}{{\bf T}}

\newcommand{\Xm}{{\bf X}}
\newcommand{\Ym}{{\bf Y}}

\newcommand{\alphav}{\hbox{\boldmath$\alpha$}}
\newcommand{\betav}{\hbox{\boldmath$\beta$}}
\newcommand{\gammav}{\hbox{\boldmath$\gamma$}}
\newcommand{\deltav}{\hbox{\boldmath$\delta$}}

\newcommand{\lambdav}{\hbox{\boldmath$\lambda$}}

\newcommand{\muv}{\hbox{\boldmath$\mu$}}

\newcommand{\rhov}{\hbox{\boldmath$\rho$}}

\newcommand{\Gammam}{\hbox{\boldmath$\Gamma$}}
\newcommand{\Lambdam}{\hbox{\boldmath$\Lambda$}}
\newcommand{\Deltam}{\hbox{\boldmath$\Delta$}}
\newcommand{\Sigmam}{\hbox{\boldmath$\Sigma$}}
\newcommand{\Phim}{\hbox{\boldmath$\Phi$}}

\newcommand{\diag}{{\hbox{diag}}}

\newcommand{\trace}{{\hbox{tr}}}

\renewcommand{\Im}{{\rm Im}}

\newcommand{\defines}{{\,\,\stackrel{\scriptscriptstyle \bigtriangleup}{=}\,\,}}

\newtheorem{proposition}[theorem]{Proposition}

\newcommand{\beqa}{\begin{eqnarray}}
\newcommand{\eeqa}{\end{eqnarray}}
\newcommand{\dsp}{\displaystyle}

\begin{document}

\title{Coordinated Multi-cell Beamforming for Massive MIMO: A Random Matrix Approach}
\author{Subhash~Lakshminarayana,~\IEEEmembership{Member, IEEE}
, Mohamad Assaad,~\IEEEmembership{Senior Member, IEEE,} and Merouane Debbah,~\IEEEmembership{Fellow, IEEE}
\thanks{
This paper was presented in part at the IEEE International
Symposium on Personal Indoor and Mobile Radio Communications
(PIMRC), September, 2010. S.Lakshminarayana is with the Singapore University of Technology and Design, Singapore. M.Assaad is with the ``Laboratoire des Signaux et Systemes (L2S, UMR CNRS 8506), CENTRALESUPELEC", France. M. Debbah is with CentraleSupelec, France and the Huawei Mathematical and Algorithmic Sciences Lab. e-mail: subhashl@princeton.edu,  mohamad.assaad@centralesupelec.fr and  merouane.debbah@centralesupelec.fr. The research of M.Debbah has been supported by the ERC Starting Grant 305123 MORE and by the French po\^le de compétitivité SYSTEM@TIC within the project 4G in Vitro. The research of M. Assaad was partially supported by the Celtic Project ``SHARING."
\newline The authors would like to thank Jakob Hoydis for his helpful discussions. 
\newline Copyright (c) 2014 IEEE. Personal use of this material is permitted.  However, permission to use this material for any other purposes must be obtained from the IEEE by sending a request to pubs-permissions@ieee.org.
} }

\maketitle

%%%%%%%%%%%%%%%%%%%%%%%%%%%%%%%%%%%%%%%%%%%%%%%%%%%%%%
%%%%%%%%%%%%%%%%%%%%%%%%%% abstract  %%%%%%%%%%%%%%%%%%%%%%%%
%%%%%%%%%%%%%%%%%%%%%%%%%%%%%%%%%%%%%%%%%%%%%%%%%%%%%%%%
%
\begin{abstract}
We consider the problem of coordinated multi-cell downlink beamforming in massive multiple input multiple output (MIMO) systems consisting of $N$ cells, $N_t$ antennas per base station (BS) and $K$ user terminals (UTs) per cell. 
%While cooperation among BSs via channel state information (CSI) and data sharing is known to increase the system throughput, it quickly becomes impractical as $N_t, N$ and $K$ grow.
Specifically, we formulate a multi-cell beamforming
algorithm for massive MIMO systems which requires limited amount of information exchange between the BSs. The design objective is to minimize the aggregate transmit power across all the BSs subject to satisfying the user signal to interference noise ratio (SINR) constraints. The algorithm requires the BSs to exchange parameters which
can be computed solely based on the channel statistics rather than the instantaneous CSI.
We make use of tools from random matrix theory to formulate the decentralized algorithm. 
We also characterize a lower bound on the set of target SINR values for which the decentralized multi-cell beamforming algorithm is feasible.
We further show that the performance of our algorithm asymptotically matches the performance of 
the centralized algorithm with full CSI sharing. 
While the original result focuses on minimizing the aggregate transmit power across all the BSs, we formulate a heuristic extension of this algorithm to incorporate a practical constraint in multi-cell systems, namely the individual BS transmit power constraints.
Finally, we investigate the impact of imperfect CSI and pilot contamination 
effect on the performance of the decentralized algorithm, and propose a heuristic extension of the algorithm
to accommodate these issues. 
Simulation results illustrate that our algorithm closely satisfies the target SINR constraints and achieves minimum power 
in the regime of massive MIMO systems. In addition, it also provides substantial power savings as compared to zero-forcing beamforming when the number of antennas per BS is of the same orders of magnitude as the number of UTs per cell. 
\end{abstract}


\section{Introduction}
Massive multiple input multiple output (MIMO) 
has been identified as an essential ingredient 
in the design of next generation cellular systems, as it provides
substantial improvement in both spectral and energy
efficiency \cite{LarssonMarzettaCommag2014}.
It refers to the idea of scaling up the number of antennas 
on the base station (BS) to a few hundreds, serving many
tens of user terminals (UTs) on the same resource block.
The basic idea is to exploit large number of antennas
to achieve greater spatial resolution and array gain,
resulting in a higher throughput and greater energy efficiency.

The idea of massive MIMO was first proposed in the seminal work of \cite{MarzettaMassiveMIMO2010}. The main finding of the paper was that as the number of antennas on the BS grows without bound, the effects of fast fading and interference vanish, and the system performance
is ultimately limited by only pilot contamination \cite{JoseAshikmin2011}.
The other attractive feature is that simple signal processing techniques at the BS, such as the use of eigen beamforming and matched filter were optimal under this setting.
Subsequently, massive MIMO systems was also studied 
from an energy efficiency point of view and shown to  achieve dramatic improvement in this regard \cite{NgoLarssonMarzetta2013}.
The impact of channel estimation, pilot contamination, and
antenna correlation in massive MIMO systems was
investigated in \cite{HoydisBrinkDebbah2013}. In
particular, it was concluded that 
sophisticated beamforming techniques such as regularized zero forcing (RZF) 
outperform  eigen beamforming in a massive MIMO setting when the number of antennas is of the same orders of magnitude as 
the number of users. Reference \cite{BjornsonSangDebbah2014} addressed the
question of how to select the system parameters
in a massive MIMO system (number of antennas per BS, number of users, transmit power etc.) to maximize the energy 
efficiency. The work in  \cite{LiuLau2014} proposes a hierarchical interference mitigation scheme in massive MIMO systems 
based on a two level precoding with the objective of maximizing the system utility. 
Although all these studies point to impressive gains in massive MIMO both in terms 
of spectral efficiency and energy efficiency, constructing such large dimensional arrays can result in significant 
additional hardware cost of the analog front ends. Moreover, extra physical dimensions are required in order reduce the mutual
coupling between the antenna elements.

Subsequently, it was proposed
that in a multi-cellular environment, the gains obtained by massive MIMO systems can be replicated
by using much lesser number of antennas if BSs are allowed to cooperate with each other.  In this context, \cite{HuhCaire2012} proposed a TDD architecture based network MIMO like system with BS cooperation and zero-forcing (ZF) beamforming, and showed that massive MIMO performance can be achieved with one order of magnitude fewer antennas per active user per cell.
These results motivate us to consider multi-cell cooperation in massive MIMO systems. However, since massive MIMO systems are inherently large, enabling BS cooperation in a multi-cellular environment requires tremendous amount of information exchange between them.

In order to address this issue, in this work, we propose an optimal decentralized multi-cell beamforming algorithm for massive MIMO systems that requires limited amount of information exchange between the BSs. We primarily focus on the so called \emph{coordinated beamforming}, in which BSs formulate their beamforming vectors taking into account the interference they cause to the neighboring cells \cite{Dahrouj2010}. This is accomplished by exchanging the CSI information between them.
However, unlike network MIMO systems, no user data exchange takes place between the BSs. The design objective considered is to minimize the aggregate transmit power across all the BSs subject to satisfying the user signal to noise ratio (SINR) constraints.
Reference \cite{Dahrouj2010} provides an optimal centralized algorithm
to solve this problem. However, the centralized solution demands high computational ability, and exchange of the fast fading CSI co-efficient between the  BSs.
Such an algorithm requires high capacity backhaul links, especially when implemented in a massive MIMO setting.

In order to overcome the heavy backhaul requirement,  we propose in this work a decentralized approach to compute the multi-cell beamforming vectors. In our algorithm, the BSs must exchange parameters at the time scale of slow fading coefficients rather than the instantaneous channel realizations (fast fading coefficients). We use tools from random matrix theory (RMT) to formulate our algorithm.

The multi-cell beamforming strategy involving exchange of  parameters based on channel statistics was first proposed in 
\cite{Lakshminarayana2010} using tools from RMT. It was shown with the help of simulations
that such an algorithm performs well for large system dimensions. %The theoretical analysis regarding the asymptotic optimality of such an algorithm was an open question. 
Similar ideas for multi-cell beamforming were subsequently proposed in \cite{Randa2010}, \cite{ZhakHanlyJSAC2013} and theoretical arguments for the asymptotic optimality of this algorithm were provided in the special case of a two cell Wyner model, with symmetric SINR constraints for all the UTs. 
However, the analyses in \cite{Randa2010}, \cite{ZhakHanlyJSAC2013} rely on obtaining closed form expressions for the system
parameters (such as the uplink and downlink power), and are not extendable to more practical channel models.
In this work, we provide a comprehensive design of the RMT based
decentralized beamforming under a massive MIMO multi-cell setting with a distance based pathloss model. Further, we provide arguments for asymptotic optimality of this algorithm.
Specifically, our main contributions are as follows:

For the original problem of minimizing the aggregate transmit power across all the BSs subject to UT SINR constraints, we present the following results: 
\begin{itemize}
\item We propose a \emph{reduced overhead} beamforming algorithm (ROBF) in a massive MIMO multi-cell setting. In this algorithm,
the BSs require the knowledge of local CSI (of the UTs they are serving and also the UTs present in the other cells).
In addition, the BSs have to compute parameters that depend only on the channel statistics, which they exchange between them to
compute the beamforming vectors. 
%The BSs have to exchange these parameters only
%at the time scale of channel statistics rather than instantaneous channel realization. 
\item Using a large system analysis, we provide closed form expression for the lower
bound on the set of target SINR, for which the decentralized multi-cell beamforming algorithm is feasible.
\item We prove that when the dimensions of the system become large, the achieved SINR in the uplink and downlink by the ROBF algorithm exactly match the target SINR.
Moreover, we also prove that when the dimensions of the system become large, the performance of our algorithm in terms of uplink and downlink transmit powers perfectly match that of the optimal algorithm proposed in \cite{Dahrouj2010}.
\end{itemize}
With this algorithm as reference, we present heuristic extensions to incorporate two practical constraints in MIMO multi-cell systems (1) individual BS transmit power constraints (2) impact of imperfect CSI and pilot contamination. The contributions are as follows:
\begin{itemize}
\item We formulate a heuristic extension of the decentralized multi-cell beamforming algorithm to incorporate the individual BS transmit power constraints, and present numerical
results to show the convergence as well as the performance of this algorithm.
\item Finally, we investigate the impact of CSI estimation and pilot contamination on the performance of the ROBF algorithm and propose a heuristic adaptation of the ROBF algorithm that can provide better performance in the presence of pilot contamination. 
\end{itemize}
In addition, our work contains several novel ideas of
combining RMT results with optimization theory
and uplink-downlink duality in MIMO systems,
which are of independent interest.

It is worth noting that RMT results have been used extensively to assess the performance of linear beamforming strategies in multi-user MIMO systems 
\cite{NguyenEvans2008}, \cite{sebastian2010}, \cite{HoydisBrinkDebbah2013}, \cite{HuhCaire2012}.
However, all these works consider ``pre-defined" transmit strategies such as eigen beamforming, zero forcing (ZF), regularized zero forcing (RZF) etc., and 
use RMT as a tool for \emph{performance analysis} of these schemes in the large dimensional regime. In contrast, in this work, we use RMT results for the purpose of \emph{system design} (as opposed to performance analysis), i.e., design of optimal beamforming vectors in MIMO multi-cell systems\footnote{In an unrelated context, RMT results have also been used in the context of system design in \cite{TulinoVerdu2001} and \cite{HuhTulinoCaire2012}}.
Such an approach is novel and is facilitated by combining RMT results with optimization theory. 
Moreover, optimizing the system performance imposes additional technical difficulties in applying RMT results. For e.g., implementing a power control algorithm (both in the uplink and downlink) implies that the transmit powers explicitly depend on the channels (and hence the randomness associated with the channel realization). Such dependency of transmit powers on the channel realization makes it unsuitable to apply RMT results. Our approach in this work is to first propose an algorithm that depends only on the second order statistics of the channel vectors (the path-loss in our case). Then, we apply such an algorithm to the original system set-up, and prove that this algorithm is optimal in the large system domain.

Although the theoretical results prove the optimality of ROBF algorithm in the asymptotic regime, we provide numerical results to 
show that the performance of ROBF algorithm closely matches that of the centralized algorithm for moderate system dimensions (when the number of
antennas are comparable to the number of UTs per cell), both in terms of satisfying the user SINR targets and minimizing the downlink transmission
power. Moreover, these results indicate that ROBF algorithm provides substantial trasnmit power savings
as compared to other beamforming strategies such as zero-forcing n the regime where the number of antennas is comparable to
the number of UTs.

Finally, we remark that all the analysis in this work is performed assuming independent and identical (i.i.d.) channel vectors and time division duplex (TDD) mode of operation. 
The performance of massive MIMO systems with correlated channel models have been studied in prior works such as \cite{HoydisBrinkDebbah2013}.
Recently, there has also been an interest in exploring 
massive MIMO systems with frequency division duplex (FDD) mode of operation \cite{LiuLau2014, AdhikaryCaireMassiveMIMOFDD2014, MasoodMassiveMIMOFDD2015}. As this work is a first step towards exploring power control and the design of optimal beamforming in massive MIMO systems, in order to keep the analysis simple, we restrict our attention to i.i.d. channels and TDD mode of operation. The extension to the case of correlation channels and FDD mode of operation will be a topic of future research.

The rest of the paper is organized as follows. We provide the system model and describe our {reduced overhead} multicell beamforming in section \ref{sec:SysmodAlg}. In section \ref{sec:AlgAnal}, we provide the asymptotic analysis
of the {reduced overhead} algorithm formulated in the previous section. In Section \ref{sec:Pilot}, we investigate the impact of 
imperfect CSI and pilot contamination on the performance of the ROBF algorithm.
We summarize the simulation results in section \ref{sec:SimRes}. Finally, we provide concluding remarks in section \ref{sec:Conc}.
Appendix A provides some relevant results from RMT
which will be used in formulating our algorithm. Appendices B, C, D, E, F, G and H provide the proofs of some of the 
results stated in the paper.

\emph{Notations:} Throughout this work, we use boldface lowercase and
uppercase letters to designate column vectors and matrices,
respectively. For a matrix $\Xm$, $\Xm(p,q)$ denotes the $(p,q)$
entry of $\Xm$. $\Xm^T,$  $\Xm^H,$ $\trace(\Xm),$ $||\Xm||,$ and $\rho(\Xm)$ denote the transpose, the complex conjugate
transpose, the trace, the spectral norm and the spectral radius of the matrix $\Xm,$ respectively.
For two matrices $\Xm$ and $\Ym,$ the notation $\Xm \leq \Ym$ denotes
element wise inequality ($\Xm(p,q) \leq \Ym(p,q) \ \forall p,q$) and 
similarly for vectors.
We denote an identity matrix of
size $M$ as $\Id_M$ and diag$(x_1, . . . , x_M)$ is a diagonal matrix of
size M with the elements $x_i$ on its main diagonal. We use
$\xv \sim \mathcal{CN}(\mv,\Rm)$ to state that the vector $\xv$ has a complex
Gaussian distribution with mean $\mv$ and covariance matrix $\Rm$.
We use the notation $\xrightarrow{\text{a.s.}}$ to denote almost sure convergence.
Let $a_N$ and $b_N$ denote a pair of infinite sequences. We write $a_N \asymp b_N$, iff $a_N - b_N \xrightarrow{\text{a.s.}} 0.$ We denote the expectation of a random variable by the notation
$\mathbb{E} \LSB.\RSB$ Let $N_t,K \in \mathbb{N}^+,$ we use the notation $N_t,K \to \infty$ to denote the following condition on $N_t$ and $K,$
$0 < \liminf_{K \to \infty} \frac{N_t}{ K} \leq \limsup_{K \to \infty} \frac{N_t}{ K} < \infty.$ Finally, the notation $(x)^+$ is used 
to denote $\max(x,0).$

\section{System Model and Algorithm Description}
\label{sec:SysmodAlg}
\subsection{System Model}
We consider the problem of  multi-cell beamforming across $N$ cells and $K$ UTs per cell where each BS is equipped with $N_t$ antennas
and each UT has a single antenna. Each BS serves only the UTs in its cell.
Let $\hv_{i,j,k} \in \mathbb{C}^{N_t}$ denote the channel from the BS $i$ to the $k$-th UT in cell $j.$ 
We consider reciprocity between the uplink and downlink channels, and hence the TDD mode of operation, as it is the preferred mode of
operation in massive MIMO systems \cite{LarssonMarzettaCommag2014}.
We assume that the elements of the channel vector are independent and identically distributed (i.i.d.) with Gaussian distribution, i.e., $\hv_{i,j,k} \sim \mathcal{CN}(0,\sigma_{i,j,k}\Id_{N_t}),$ the variance $\sigma_{i,j,k}$ of the channel depends upon the path loss model  between BS $i$ and UT$(j,k).$
Recent works on channel measurements indicate that i.i.d. assumption is a reasonable model for massive MIMO arrays \cite{MaMIMORealProp2014}.
We assume that the BSs have perfect CSI of the downlink channels
to all the users in the system (i.e., $\hv_{i,n,k}, \forall n,k$).
Let $\wv_{i,j} \in \mathbb{C}^{N_t}$ denote the transmit downlink beamforming vector for the $j$-th UT in cell $i.$ Likewise, let $\Lambda^{\text{DL}}_{i,j} $ denote the received SINR for the
$j$th UT in cell $i$ and $\gamma_{i,j}$ the corresponding target SINR. The received signal $y_{i,j} \in \mathbb{C}$ for the $j$th UT in cell $i,$  is given by 
\begin{align*}
y_{i,j} =  \hv^H_{i,i,j} \wv_{i,j}x_{i,j}+\sum_{(n,k) \neq (i,j)} \hv^H_{n,i,j} \wv_{n,k}x_{n,k}+z_{i,j}
\end{align*}
where $x_{i,j}\in\CC$ represents the information signal for the $j$-th user in cell $i$ and
$z_{i,j} \sim \mathcal{CN}(0,N_0)$ is the corresponding additive white Gaussian complex noise.
Under this model, the achieved SINR in downlink for
the UT$_{i,j}$ is given by\footnote{Note that the downlink SINR expression in \eqref{eqn:jaja693} assumes that the channels are known perfectly at the UTs. However, the main focus of this work 
is in the massive MIMO regime, i.e., $N_t$ and $K$ being very large. Fortunately, as shown in \cite{HoydisBrinkDebbah2013}, under this regime, the UTs only need to have
the knowledge of the average effective channels.}
\begin{align}
\Lambda^{\text{DL}}_{i,j} = \frac{|\wv^H_{i,j}\hv_{i,i,j}|^2}{ \sum_{(n,k) \neq (i,j)}|\wv^H_{n,k}\hv_{n,i,j}|^2+N_0}. \label{eqn:jaja693}
\end{align}
The denominator terms of \eqref{eqn:jaja693} represent the
intra-cell interference, inter-cell interference and the noise (in order as
they appear).
The downlink sum power minimization problem can be formulated as the following optimization problem given by
\beqa
 &\dsp  \min_{\wv_{i,j} \ \forall i,j} & \sum_{i,j}  \wv^H_{i,j}\wv_{i,j} \label{eqn:OPT_basic}\\
& s.t. & \Lambda^{\text{DL}}_{i,j} \geq \gamma_{i,j} \qquad \forall i,j. \nonumber
 \eeqa
 
 \subsection{Algorithm Design}
 As show in  \cite{WeiselShamai2006},
 the optimization problem \eqref{eqn:OPT_basic} can be reformulated  as a  second order conic programming (SOCP) problem and, strong duality holds for this problem.
Following the approach of \cite{Dahrouj2010},
we solve this problem using duality theory.
Accordingly, we introduce the Lagrange multiplier $\frac{\lambda_{i,j}}{N_t}$
associated with the downlink SINR constraints. The Lagrangian
is given by
\begin{align}
L(\wv,\lambdav) & = \sum_{i,j}  \wv^H_{i,j} \wv_{i,j}  - \sum_{i,j} \frac{\lambda_{i,j}}{N_t} \Big{[}  \frac{|\wv^H_{i,j} \hv_{i,i,j}|^2}{\gamma_{i,j}} \nonumber \\ & -\sum_{(n,k) \neq (i,j)} |\wv_{n,k} \hv_{n,i,j}|^2 -N_0 \Big{]} \label{eqn:LagDL}.
\end{align}
Note that the Lagrange multiplier is scaled by the factor $N_t$
in order to ensure that the sum power in the system is finite,
when the dimensions of the system grow large (in terms 
of the number of antennas on the BS and number of UTs).
In fact, the Lagrange multipliers $\frac{\lambda_{i,j}}{N_t}$ can be interpreted as the dual uplink powers in the formulation of the dual uplink problem obtained in the following manner.

Rearranging \eqref{eqn:LagDL}, we obtain
\begin{align}
& L(\wv,\lambdav)  = \sum_{i,j}  \frac{\lambda_{i,j} N_0}{N_t}   + \sum_{i,j}  \wv^H_{i,j} \Bm_{i,j} \wv_{i,j} \label{eqn:LagUL},  
\end{align}
where the matrix $\Bm_{i,j}$ is given by
\begin{align}
& \Bm_{i,j}  
    = \Id \ -\LB 1+\frac{1}{\gamma_{i,j}}\RB \frac{\lambda_{i,j}}{N_t} \hv_{i,i,j} \hv^H_{i,i,j} \nonumber \\ & \qquad \qquad \qquad + \sum_{n,k}  \frac{\lambda_{n,k}}{N_t} \hv_{i,n,k} \hv^H_{i,n,k}  \nonumber \\
  & \ =  \Id \ - \frac{\lambda_{i,j}}{\gamma_{i,j} N_t} \hv_{i,i,j} \hv^H_{i,i,j}+ \sum_{(n,k) \neq (i,j)}  \frac{\lambda_{n,k}}{N_t} \hv_{i,n,k} \hv^H_{i,n,k} \label{eqn:refhere0011}.
\end{align}
%The downlink sum power minimization problem in \eqref{eqn:OPT_basic} has been solved in  using the framework first developed in \cite{WeiselShamai2006}. 
%In \cite{WeiselShamai2006}, the authors formulate the SINR constraints as second order conic constraints. Further, the authors formulate the dual of the optimization problem in \eqref{eqn:OPT_basic} and prove that KKT conditions are both necessary and sufficient to solve the problem. The authors of \cite{Dahrouj2010} provide the interpretation of uplink-downlink duality to the SOCP based framework in the context of multicell scenario. 
The dual uplink problem  corresponding to the optimization in \eqref{eqn:OPT_basic}
is formulated as
\beqa
 &\dsp \min_{\lambda_{i,j}, \ \forall i,j} & \sum_{i,j} \frac{\lambda_{i,j}}{N_t}N_0 \label{eqn:OPT_basic2}\\
& s.t. & \Lambda^{\text{UL}}_{i,j} \geq \gamma_{i,j}, \qquad \forall i,j \nonumber
 \eeqa
 where the left hand side of the constraint equation represents the uplink SINR given by
\begin{align*}
\Lambda^{\text{UL}}_{i,j} = \frac{\frac{\lambda_{i,j}}{N_t}|\hat{\wv}^H_{i,j}\hv_{i,i,j}|^2}{\sum_{(n,k) \neq (i,j)}\frac{\lambda_{n,k}}{N_t}|\hat{\wv}^H_{i,j}\hv_{i,n,k}|^2+ ||{\hat{\wv}}_{i,j}||_2^2}
\end{align*}
where $\hat{\wv}_{i,j}$ denotes the corresponding uplink receive filter. 

We now provide a brief description of the beamforming algorithm presented in \cite{Dahrouj2010}.
Before introducing the algorithm,
we define the following matrices. 
\begin{align*}
& \Hm_{i,n}  = [\hv_{i,n,1}, \dots ,\hv_{i,n,K}] \in \CC^{N_t \times K} \\
& \Hm_i = \LSB \Hm_{i,1},\dots,\Hm_{i,N}\RSB \in \CC^{N_t \times NK} \\
& \lambdav_i = [\lambda_{i,1}, \dots, \lambda_{i,K}] \in \CC^{K \times 1}, \\
&\Lambdam = \diag \LSB \lambdav_1,\dots,\lambdav_N\RSB \in \CC^{NK \times NK}.
\end{align*} 
We also define the matrix $\Sigmam^\lambda_i = \frac{1}{N_t}\Hm_i\Lambdam \Hm^H_i \in \CC^{N_t \times N_t}.$ 
%We first briefly summarize the multi-cell beamforming algorithm provided in \cite{Dahrouj2010}.
\vspace{0.05in} \hrule
\vspace{0.01in}\hrule\vspace{0.05in}
\begin{algorithm}[Centralized Algorithm - {\bf CBF}]
\label{alg:centralized}
Perform the following steps.
\begin{itemize}
\item Starting from any initial $\lambda^0_{i,j} > 0 \ \forall i,j$ the uplink power allocation is given by 
${\lambda}_{i,j} \defines \lim_{t \to \infty} {\lambda}^t_{i,j},$ where 
\begin{align} \label{eqn:lam_org}
\lambda^{t+1}_{i,j} = \frac{1}{\frac{1}{N_t}(1+\frac{1}{\gamma_{i,j}})\hv^H_{i,i,j}(\Sigmam^{\lambda^t}_i+\Id_{N_t})^{-1}\hv_{i,i,j}} \ \forall i,j 
\end{align}
where $\Sigmam^{\lambda^t}_i = \frac{1}{N_t}\Hm_i\Lambdam^t \Hm^H_i$ and $\Lambdam^t = \diag \LSB \lambdav^t_1,\dots,\lambdav^t_N\RSB.$
\item The optimal receive uplink receive filter is given by
\begin{align} \label{eqn:ULBF}
\hat{\wv}_{i,j} = \frac{1}{\sqrt{N_t}}\Big( \sum_{n,k} \frac{\lambda_{n,k} N_0}{N_t} \hv_{i,n,k} \hv^H_{i,n,k}+  N_0 \Id \Big)^{-1} \hv_{i,i,j}.
\end{align}
\item The optimal transmit downlink beamforming vectors are given by
$\wv_{i,j} = \sqrt{\frac{\delta_{i,j}}{N_t}}\hat{\wv}_{i,j},$ 
where $\delta_{i,j}$ is given as
\begin{align*}
\deltav = \Fm^{-1} {\bf 1} N_0.
\end{align*}
Here, 
\begin{align*}
& \deltav_i = [\deltav_{i,1}, \dots, \deltav_{i,K}] \in \CC^{K \times 1} 
\\ & \deltav =  \LSB \deltav_1,\dots,\deltav_N\RSB \in \CC^{NK \times 1} \\
& {\bf 1} \in [1,\dots,1]^T \in \RR^{NK \times 1}
\end{align*}
and the elements of the matrix $\Fm \in \CC^{NK \times NK}$ and 
the submatrix $\Fm^{i,j} \in \CC^{K \times K}$ are given by,
\begin{equation}
\Fm = \left(
\begin{array}{ccc}
\Fm^{1,1} &  \ldots & \Fm^{1,N} \\
\vdots &   \ddots & \vdots \\
\Fm^{N,1} &  \ldots & \Fm^{N,N}\\
\end{array} \right)
\end{equation}
\begin{equation} \label{eqn:FMtx}
\Fm^{i,n}_{j,k} \defines \begin{cases} \frac{1}{\gamma_{i,j}N_t} |\hat{\wv}^H_{i,j}\hv_{i,i,j}|^2 , & n=i ,\ k = j \\
                 \frac{-1}{N_t}|\hat{\wv}^H_{n,k}\hv_{n,i,j}|^2  ,& ( n,k) \neq (i,j). 
                   \end{cases}
\end{equation}
\end{itemize}
\end{algorithm}
\hrule 
\vspace{0.01in} \hrule \vspace{0.05in}
We remark that the scaling of $\sqrt{N_t}$ in 
the expression for the uplink receive filter \eqref{eqn:ULBF},
and in the definition of the scaling factor $\delta_{i,j}$ ensure that the total power in the system is finite, when the dimensions of the system grow large.

As mentioned before, the solution provided in \cite{Dahrouj2010} cannot be implemented in a distributed manner.
The computation of dual uplink power ($\lambda_{i,j}$) and the scaling factors ($\delta_{i,j}$) requires a central station which has the global CSI knowledge.
In what follows, we overcome this problem.
%The intuition behind the in the large dimensional regime, i.e., when the number of antennas per BS and the number of UTs per cell become large, 
%RMT results can be applied for the computation 

We now formulate our \emph{reduced overhead} beamforming algorithm.
The main idea behind this algorithm is that under the massive MIMO regime (i.e. when $N_t$ and $K$
become large), the parameters in \eqref{eqn:lam_org} and \eqref{eqn:FMtx} can be
approximated by their asymptotic equivalents using results from RMT. Moreover, the computation of these parameters
will only depend on the second order statistics (path-loss), and not on the fast fading component
of the channel vectors. 
However, note that RMT results are not directly applicable to this scenario. This is due to
the fact that the computation of $\lambda_{i,j}$ in \eqref{eqn:lam_org} explicitly depend 
on the channel vectors (RMT results require that the matrix $\Lambdam$ in \eqref{eqn:lam_org}
are independent of the channel matrices). This imposes additional technical difficulties 
in the application of RMT results in our scenario. Our approach in this work is to first propose an algorithm that depends only on the second order statistics of the channel vectors (the path-loss in our case). Mathematically/theoretically, it is not ensured that it achieves the optimal solution.
However, we apply such an algorithm to the original system set-up, and prove that this algorithm is optimal in the large system domain.

We hereby represent the dual uplink power, the uplink and downlink beamforming vectors
of the decentralized algorithm by the notation $\mu_{i,j},\hat{\vv}_{i,j}$ and $\vv_{i,j},$ respectively, which are 
the counterparts of $\lambda_{i,j},\hat{\wv}_{i,j}$ and $\wv_{i,j}$ of the CBF algorithm.
\vspace{0.05in} \hrule
\vspace{0.01in}\hrule\vspace{0.05in}
\begin{algorithm}[Reduced Overhead Beamforming algorithm - {\bf ROBF}] Perform the following steps.
\label{alg:decentralized}
\begin{itemize}
\item Starting from any initial $\mu^0_{i,j} > 0 \ \forall i,j$ the uplink power allocation is given by ${\mu}_{i,j} \defines \lim_{t \to \infty} \mu^t_{i,j},$ where 
\begin{align}
\mu^{t+1}_{i,j} =   \frac{\gamma_{i,j}}{\sigma_{i,i,j}\bar{m}^t_{i}} \qquad \forall i,j \label{eqn:lam_stil}
\end{align}
and $\bar{m}^t_i$ is evaluated as  $\bar{m}^{t}_i \defines \lim_{p \to \infty} \bar{m}^{t,p}_i$  (initializing with any $\bar{m}^{t,0}_i > 0, \forall i$)
\begin{align}
& \bar{m}^{t,p}_i  = \LB \frac{1}{N_t}\sum^N_{n = 1} \sum^K_{k = 1} \frac{\sigma_{i,n,k} \mu^t_{n,k}}{1+\sigma_{i,n,k}\mu^{t}_{n,k} \bar{m}^{t,{p-1}}_i}+1 \RB^{-1}. \label{eqn:FPeqn}
\end{align}
\item The optimal receive uplink receive filter is given by
\begin{align} \label{eqn:ULBF_asymp}
\hat{\vv}_{i,j} = \sqrt{\frac{1}{N_t}}\Big( \sum_{n,k} \frac{\mu_{n,k}N_0}{N_t} \hv_{i,n,k} \hv^H_{i,n,k}+ N_0 \Id\Big)^{-1} \hv_{i,i,j}.
\end{align}
\item The optimal transmit downlink beamforming vectors are given by
$\vv_{i,j} = \sqrt{\frac{\bar{\delta}_{i,j}}{N_t}}\hat{\vv}_{i,j}.$
The scaling factor $\bar{\delta}_{i,j}$ is given as
\begin{align}
\bar{\deltav} = (\Id-\Gammam \Deltam)^{-1} \rhov, \label{eqn:lin_eq}
\end{align}
where
\begin{align*}
& \bar{\deltav}_i = [\bar{\delta}_{i,1}, \bar{\delta}_{i,2},\dots, \bar{\delta}_{i,K}]^T \in  \RR^{K \times 1}\\
& \bar{\deltav} = [\bar{\deltav}_1 , \bar{\deltav}_2,\dots, \bar{\deltav}_N]^T \in \RR^{NK \times 1} \\
 & \gammav_i = \LSB  \frac{\gamma_{i,1}}{\sigma_{i,i,1}\bar{G}_{i,i,1}\bar{m}^2_{i}},\dots,\frac{\gamma_{i,K}}{\sigma_{i,i,K}\bar{G}_{i,i,K}\bar{m}^2_{i}}\RSB^T \\
& \gammav = [\gammav_1,\dots,\gammav_N]^T \\ & \Gammam = \text{diag}(\gammav)
\end{align*}
and the matrix $\Deltam \in \CC^{NK \times NK} $ is defined as
\begin{equation}
\Deltam = \left(
\begin{array}{ccc}
\Deltam^{1,1} & \ldots & \Deltam^{1,N} \\
\vdots  & \ddots & \vdots \\
\Deltam^{N,1} & \ldots & \Deltam^{N,N}\\
\end{array} \right)
\end{equation}
where each submatrix $\Deltam^{i,j} \in \CC^{K \times K}$ is given by 
\begin{equation} 
\Deltam^{i,n}_{j,k} \defines \begin{cases} 0 , & n=i ,\ k = j \\
                     \frac{1}{N_t}\bar{G}_{n,i,j} \bar{G}_{n,n,k}  \bar{m}^{\prime}_{n}  ,& (n,k) \neq (i,j).
                   \end{cases} \label{eqn:coeff_mtx}
\end{equation}
$\bar{m}^{\prime}_{i}$ can be evaluated from $\bar{m}_i$ as
\begin{align}
\bar{m}^\prime_i = \frac{\bar{m}^2_i}{1-\frac{1}{N_t} \sum^N_{n = 1} \sum^K_{k=1} \frac{(\sigma_{i,n,k} \mu_{n,k}\bar{m}_i)^2}{(1+\sigma_{i,n,k} \mu_{n,k}\bar{m}_i)^2}}
\end{align}
and the term
\begin{align}
\bar{G}_{i,n,k} & = \frac{\sigma_{i,n,k}}{\LB 1+\mu_{n,k}\sigma_{i,n,k}\bar{m}_{i}\RB^2} \label{eqn:haha14}.
\end{align}
The vector $\rhov_i = \LSB \frac{N_0}{\sigma_{i,i,1}\bar{G}_{i,i,1}\bar{m}^2_{i}}, \dots,\frac{N_0}{\sigma_{i,i,K}\bar{G}_{i,i,K}\bar{m}^2_{i}}\RSB^T$
and $\rhov = [\rhov _1 , \rhov _2,\dots, \rhov_N]^T \in \RR^{NK \times 1}.$
\end{itemize}
\end{algorithm}
\hrule 
\vspace{0.01in} \hrule \vspace{0.1in}

We now provide some remarks on this algorithm.
\begin{lemma}
The iterative algorithm \eqref{eqn:lam_stil} converges to a fixed
point. 
\end{lemma}
\begin{proof}
The proof is provided in Appendix B.
\end{proof}
In section III , we characterize the solution provided by this fixed point, and show that it is asymptotically optimal in the sense that when the dimensions of the system grow large, the achieved uplink and downlink SINR by this algorithm exactly match the target SINR, and the allocated uplink and downlink powers
match the optimal solution.

In the rest of the paper, we address the decentralized beamforming algorithm
by the acronym ROBF. We now discuss the practical advantages of the ROBF algorithm over the CBF algorithm. 
\subsection{ROBF Algorithm - Signaling Overhead and Complexity Reduction} 
%Notice that the computation of the uplink power allocation in the ROBF algorithm
%depends only on the second order statistics of the channel matrix and not on the instantaneous CSI.
We now provide a brief discussion on the signaling overhead and complexity reduction associated with the ROBF algorithm.
{\bf Signaling Overhead:} \\
Recall that in the ROBF algorithm only the statistical information (path loss) must be exchanged between the BSs where as in the CBF algorithm, 
the fast fading co-efficients need to be exchanged. Let us
denote the channel coherence time by $\tau_{\text{coh}}$ units, and the long term time constant by $\tau_{\text{LT}}$ units (over which the statistical properties of the channel change).
In has been demonstrated in some prior works related to channel measurements, 
that in an urban setting the statistical information of the channel can be viewed as constant roughly for over $100$ channel coherence intervals \cite{SpatialLongTerm2002}. 

We now characterize the signaling overhead for the two algorithms.  
For the CBF algorithm, exchanging full CSI would imply $N N_t K $ complex channel coefficients of $2 N N_t K $ real numbers every $\tau_{\text{coh}}$ units of time. This would amount to exchanging 
$N N_t K $ complex channel coefficients or $2 N N_t K $ real co-efficients. 
Therefore, the rate of information exchange for CBF algorithm would be
$$R_{\textbf{CBF}} = \frac{2 N N_t K }{\tau_{\text{coh}}} \ \text{real coefficients/sec} .$$ 
For the ROBF algorithm, $N K$ real numbers (path loss information assuming i.i.d. channel model) must be exchanged 
every $\tau_{\text{LT}}$ units of time, and the resulting rate of information exchange is given by $$R_{\textbf{ROBF}} = \frac{N K }{\tau_{\text{LT}}} \ \text{real coefficients/sec}.$$
In order to get real feel of these numbers, assume $N = 3, N_t = 100, K = 40, \tau_{\text{LT}} = 22.6 s, \tau_{\text{coh}} = 180 ms$ \cite{SpatialLongTerm2002}. Therefore, 
\begin{align*}
R_{\textbf{CBF}} & = 1.3 \times 10^5 \ \text{real coefficients/sec}; \\ 
R_{\textbf{ROBF}} & = 5.3  \ \text{real coefficients/sec},
\end{align*}
Therefore, the ratio of two quantities can be given by
\begin{align*}
\frac{R_{\textbf{CBF}}}{R_{\textbf{ROBF}}} =2 N_t \LB \frac{ \tau_{\text{LT}}}{\tau_{\text{coh}}} \RB \approx 10^4.
\end{align*}
However, for correlated channel model, $R_{\textbf{ROBF}} = \frac{N N_t K }{\tau_{\text{LT}}} \ \text{real coefficients/sec}.$ Therefore
\begin{align*}
\frac{R_{\textbf{CBF}}}{R_{\textbf{ROBF}}} =2  \LB \frac{ \tau_{\text{LT}}}{\tau_{\text{coh}}} \RB \approx 100.
\end{align*}
Finally, note that the exact number of bits to be exchanged depends on how these channel co-efficients are quantized (which is beyond the scope of this paper).

{\bf Implementation Complexity:} \\
An exact characterization of the complexity associated with the ROBF algorithm is beyond the scope of this work. Nevertheless, we provide 
a brief analysis of the same. \\
{\bf Computing the uplink power allocation} \\
 Recall that implementing the iterations for the computation of the uplink power in the CBF algorithm requires matrix inversion operations, where as the ROBF algorithm only requires performing scalar operations.
This tremendously reduces the computational complexity with respect to the ROBF algorithm (the complexity of inverting a matrix is provided next).
Moreover, in the CBF algorithm the uplink power must be evaluated for every channel realization (fast fading CSI), where as the ROBF algorithm
requires parameters to be computed only once (at the time scale of changing of slow fading CSI). Therefore, there is a huge reduction 
in the computational complexity.

{\bf Formulation of the downlink beamforming vectors} \\
Note that the downlink beamforming vector in the ROBF algorithm (and also the CBF algorithm) is in the form of a regularized zero forcing (RZF) beamforming, which requires computation of a matrix inverse 
of dimension $N_t \times NK.$
Therefore, its computational complexity scales as $N_t(NK)^2.$ This can be computationally demanding especially in a massive MIMO setting. Fortunately, alternate  
schemes are being developed to implement RZF based on
truncated polynomial expansion incur much less computational complexity as compared to matrix inversion, which are shown to have performance very close to RZF beamforming \cite{KammounMuller2013}.

The price to pay for the reduction in information exchange between the BSs is that in the ROBF algorithm, the target SINR values are not met perfectly 
for every channel realization. In fact, the achieved SINR in the downlink fluctuates around the target SINR. However, we show through simulations that even for practical values of the system dimensions, the fluctuations
of the achieved SINR around the target SINR value are small.
In order to make sure that the target SINR requirements are satisfied, one could solve the optimization
problem with ROBF algorithm by considering a higher value target SINR (than the actual desired one) in order to
compensate for the fluctuations. 

Next, we show that the performance of ROBF algorithm perfectly matches the CBF algorithm when
the number of antennas per BS and the number of UTs become large, i.e., in the regime of massive MIMO systems.
We also provide simulation results to examine the performance of the ROBF algorithm.

\section{Algorithm Analysis}
\label{sec:AlgAnal}
In this section, we provide extensive analysis
of the ROBF algorithm. 
In the rest of the paper, we
use the phrase ``large system" to refer to 
the regime when the number of antennas per BS and 
the number of UTs per-cell become large, i.e.,
$N_t, K \to \infty$ while their ratio
$\frac{N_t}{K}$ tends to a finite constant 
$0 <\beta < \infty,$ as considered in 
some past works in this field \cite{HoydisBrinkDebbah2013}.
However, we mention that our results provide tight approximations for practical system dimensions of massive MIMO systems.
Specifically, we focus on the 
following aspects:
\begin{itemize}
\item We first characterize a lower bound on
the feasible SINR targets for the ROBF algorithm.
\item We prove that the uplink and the downlink SINR achieved 
by the ROBF algorithm asymptotically converge to their target
value in the large system regime.
\item We prove that the uplink and the downlink power allocations yielded by the ROBF algorithm asymptotically converge to the respective values of the CBF algorithm, hence making the ROBF algorithm optimal in the large system regime.
\end{itemize}
We first start with the characterization of feasible SINR 
targets for the ROBF algorithm.
\subsection{Feasible SINR targets for the ROBF algorithm}
Throughout this subsection, we use the notations
$\gammav$ to denote the vector of target SINRs as follows:
\begin{align*}
\gammav_i = [\gamma_{i,1}, \dots,\gamma_{i,K}]; \ \
\gammav = [\gammav_{1}, \dots,\gammav_{N}].
\end{align*}
Similarly, we define the vectors
\begin{align*}
\lambdav_i = [\lambda_{i,1}, \dots,\lambda_{i,K}]; \ \
\lambdav = [\lambdav_{1}, \dots,\lambdav_{N}], \\
\muv_i = [\mu_{i,1}, \dots,\mu_{i,K}]; \ \
\muv = [\muv_{1}, \dots,\muv_{N}]. 
\end{align*}
We now define the notion of feasible SINR target for the
ROBF algorithm. Feasible SINR target implies that the following three conditions must be satisfied :
\begin{itemize}
\item[{\bf [C1]}] The iterations of the 
fixed point equation \eqref{eqn:lam_stil} must converge
to a finite value $\muv$.
This implies that given a target SINR 
vector $\gammav,$ there exists
$\muv < \infty$ that satisfies
the fixed point equation
\begin{align}
\gamma_{i,j} = \sigma_{i,i,j} \mu_{i,j} \bar{m}_i \qquad \forall i,j \label{eqn:FPconv}.
\end{align}

Further from the property of the fixed point equation,
for a given $\gammav$ there exists a unique
$\muv$ satisfying \eqref{eqn:FPconv} (if the target SINR
vector is feasible). Let us define
such a pair of vector by $\{ \gammav,\muv\}.$ 
\item[{\bf [C2]}] 
For every pair of vectors $\{ \gammav,\muv\}$ that satisfy \eqref{eqn:FPconv},
the matrix $\Id-\Gammam \Delta$ must be invertible.
\item[{\bf [C3]}] The elements of the vector
$(\Id-\Gammam \Delta)^{-1} \rhov$ must be positive
(in order for $\sqrt{\frac{\delta_{i,j}}{N_t}}$ to be real).
\end{itemize}
We focus on the conditions {\bf [C1]} and {\bf [C2]}
and defer {\bf [C3]} to the end of the section.
We first state the main result of this subsection and 
later on provide the proof.
\begin{theorem}
Every target SINR vector $\gammav$ 
whose elements satisfy the condition 
\begin{align}
 \frac{1}{N_t}  \sum^K_{k = 1} \frac{\gamma_{i,k}}{1+\gamma_{i,k}}
 +\frac{1}{N_t} \sum^N_{\substack{
   n =1 \\
   n \neq i
  }} \sum^K_{k = 1} \frac{\frac{\sigma_{max}(n)}{\sigma_{n,n,k}} \gamma_{n,k}}{1+\frac{\sigma_{max}(n)}{\sigma_{n,n,k}} \gamma_{n,k}}  < 1. \label{eqn:feascondROBF}
\end{align}
are feasible for the ROBF algorithm, where $\sigma_{max}(n) = $.
\end{theorem}
We remark that in our result, the condition \eqref{eqn:feascondROBF} is only a sufficient condition.
The proof of this theorem involves several steps which will
be illustrated in the rest of this subsection.
We proceed as follows.

First, we establish a relationship 
between the two conditions {\bf [C1]} and {\bf [C2]}.
\begin{lemma}
\label{lem:updownfeas}
For every pair of vectors $\{ \gammav,\muv\}$ that satisfy \eqref{eqn:FPconv}, the matrix $\Id-\Gammam \Deltam$ is invertible.
\end{lemma}
\begin{proof}
The proof can be found in Appendix C, part I. 
\end{proof}
Lemma \ref{lem:updownfeas} implies that the condition
{\bf [C2]} is automatically satisfied for every 
pair of vectors $\{ \gammav,\muv\}$ that satisfy
the condition {\bf [C1]}.
Following the result of Lemma \ref{lem:updownfeas},
we consider characterizing only the set of target SINR vectors $\gammav$ that satisfy the condition {\bf [C1]}.

The exact set of feasible target SINR vectors satisfying the condition {\bf [C1]}  is difficult to be characterized in closed form. Therefore, we only establish a lower bound on this set.
In order to do so, we consider a modified system in which the inter-cell interference path loss coefficients are replaced by 
\begin{align}
\sigma_{\max}(n) \defines \sup_k \sigma_{i,n,k}, \ \forall n \neq i.
\end{align}
where $i = 1,\dots,N.$
Therefore, in the modified system $\sigma_{i,n,k} = \sigma_{\max}(n).$

Let us consider the ROBF algorithm applied to both the original
and the modified systems.
The fixed point equation for the computation of the uplink power allocation for the original system must satisfy the following equations (we use the superscript "org" to represent the original system):
\begin{align}
\gamma_{i,j} & = \sigma_{i,i,j} \mu^{\text{org}}_{i,j} \bar{m}^{\text{org}}_i \qquad \forall i,j,
\\  \frac{1}{\bar{m}^{\text{org}}_i}  & = \frac{1}{N_t}  \sum^K_{k = 1}\frac{\sigma_{i,i,k} \mu^{\text{org}}_{i,k}}{1+\sigma_{i,i,k}\mu^{\text{org}}_{i,k} \bar{m}^{\text{org}}_i} \nonumber \\ & +\frac{1}{N_t} \sum^N_{\stackrel{n = 1} { n \neq i}} \sum^K_{k = 1} \frac{\sigma_{i,n,k} \mu^{\text{org}}_{n,k}}{1+\sigma_{i,n,k}\mu^{\text{org}}_{n,k} \bar{m}^{\text{org}}_i}+1.
\end{align}
Similarly, the fixed point equation for the \emph{modified system} must satisfy the following equations (we use the superscript "mod" to represent the modified system):
\begin{align}
\gamma_{i,j} & = \sigma_{i,i,j} \mu^{\text{mod}}_{i,k} \bar{m}^{\text{mod}}_i \label{eqn:jaja666}\\
\frac{1}{\bar{m}^{\text{mod}}_i} & = \frac{1}{N_t}  \frac{\sigma_{i,i,k}  \mu^{\text{mod}}_{i,k}}{1+\sigma_{i,i,k}\mu^{\text{mod}}_{i,k} \bar{m}^{\text{mod}}_i} \nonumber \\ & +\frac{1}{N_t} \sum^N_{\stackrel{n = 1} { n \neq i}} \sum^K_{k = 1} \frac{\sigma_{max} (n) \mu^{\text{mod}}_{n,k}}{1+\sigma_{max} (n) \mu^{\text{mod}}_{n,k} \bar{m}^{\text{mod}}_i}+1. \label{eqn:ref668}
\end{align}

In what follows, we consider the ROBF algorithm applied to the \emph{modified system} and characterize the feasible SINR targets corresponding to this system. The set of feasible target SINR 
of the ROBF algorithm applied to the \emph{modified system} will act as a lower bound on the set of feasible target SINR 
of the ROBF algorithm applied to the original system. Intuitively, this is not hard to see.
In the \emph{modified system}, the path losses corresponding to the inter-cell interference 
links are scaled up to $\sigma_{max}(n).$ Therefore, the \emph{modified system} represents a more interference
limited regime as compared to the original system. Hence, any SINR feasible for ROBF applied to the \emph{modified system} should be feasible for the ROBF algorithm applied to the original system as well. We will later on make rigorous arguments to prove that the above statement in Lemma \ref{lem:ModOrgFeas}.

\begin{proposition}
\label{prop:FeasCond}
The set of feasible SINR targets for the ROBF algorithm applied to the modified system must satisfy for $i = 1,\dots,N,$
%The feasibility conditions for the \emph{modified system} can be given in closed form as
%the set of equations which satisfy the following conditions,
\begin{align}
 \frac{1}{N_t}  \sum^K_{k = 1} \frac{\gamma_{i,k}}{1+\gamma_{i,k}}
 +\frac{1}{N_t} \sum^N_{\substack{
   n =1 \\
   n \neq i
  }} \sum^K_{k = 1} \frac{\frac{\sigma_{max}(n)}{\sigma_{n,n,k}} \gamma_{n,k}}{1+\frac{\sigma_{max}(n)}{\sigma_{n,n,k}} \gamma_{n,k}} < 1. \label{eqn:feas_cond_final_finite}
\end{align}
\end{proposition}
\begin{proof}
The proof is provided in Appendix C, part II.
\end{proof}
We now present the following Corollary. \\
{\bf Corollary 1:} \\
In the large system regime, the set of feasible SINR targets for the ROBF algorithm applied to the modified system must satisfy for $i = 1,\dots,N,$
%The feasibility conditions for the \emph{modified system} can be given in closed form as
%the set of equations which satisfy the following conditions,
\begin{align}
\limsup_{N_t,K \to \infty} \LSB \frac{1}{N_t}  \sum^K_{k = 1} \frac{\gamma_{i,k}}{1+\gamma_{i,k}}
 +\frac{1}{N_t} \sum^N_{\substack{
   n =1 \\
   n \neq i
  }} \sum^K_{k = 1} \frac{\frac{\sigma_{max}(n)}{\sigma_{n,n,k}} \gamma_{n,k}}{1+\frac{\sigma_{max}(n)}{\sigma_{n,n,k}} \gamma_{n,k}}\RSB \nonumber \\ < 1. \label{eqn:feas_cond_final}
\end{align}
The arguments for the Corollary is also provided in Appendix C, part II.

Finally, we show that the feasibility condition of \eqref{eqn:feas_cond_final_finite} will
act as a lower bound on the set of feasibile SINR targets for the ROBF algorithm applied to the original system.
\begin{lemma}
\label{lem:ModOrgFeas}
Any target SINR feasible for the ROBF algorithm applied to the \emph{modified system} is feasible for the ROBF algorithm applied to the original system.
Hence the feasible SINR targets for the \emph{modified system} will act as a lower bound for the feasible SINR target of the original system.
\end{lemma}
\begin{proof}
The proof is provided in Appendix C, part III.
\end{proof}
Lastly, we establish the condition {\bf [C3]}.
This is rather straight forward and can be seen by the following steps: From 
Appendix C, we have established that for target SINR values that satisfy the condition in Proposition \ref{prop:FeasCond}, $\rho(\Gammam \Deltam) < 1.$
Therefore, using series expansion for the matrix $(\Id-\Gammam \Deltam)^{-1}$ \cite{MtxBook1971}, we have
\begin{align}
(\Id-\Gammam \Deltam)^{-1} \rhov = \Big{(} \sum^{\infty}_{j = 1} (\Gammam \Deltam)^j \Big{)} \rhov.
\end{align}
Since the elements of $\Gammam \Deltam$ and $\rhov$ are positive,
their sum will also be positive. Therefore, the condition 
{\bf [C3]} holds true.

We end this section by providing the feasibility conditions in two special cases in which the feasibility conditions are both necessary and sufficient: \\
{\bf Example 1: Single Cell Case} \\
For the isolated single cell case, \eqref{eqn:feas_cond_final} reduces to the following
\begin{align}
 \frac{1}{N_t}\sum^K_{k = 1} \frac{\gamma_{k}}{1+\gamma_{k}} < 1. \label{eqn:ref001}
\end{align}
When all the UTs are demanding the same SINR $\gamma_k = \gamma \ \forall k,$ the set
of feasible SINR is given by
$\gamma < \LB \frac{K}{\min \{ N_t,K \} }-1 \RB^{-1}.$
The result can be interpreted as follows.
In the case of a single cell, as long as $N_t \geq K,$  any finite SINR target 
is supportable by the ROBF algorithm (for the case when all $\gamma_{i,j}$ are equal). In other words, when $N_t \geq K,$ the BS has enough degrees of freedom to sever all the UTs in the system.
Note that \eqref{eqn:ref001} matches the feasibility conditions derived in \cite{WeiselShamai2006} for the case of a single cell system. 
 \\
{\bf Example 2: 2-Cell Wyner Model} \\
Consider a perfectly symmetric multi-cell system in which the path loss for the intra-cell links are equal to $1,$ and the path loss of the inter-cell links are equal to $\epsilon.$ Every UT demands the same SINR target given by $\gamma.$
In this case, \eqref{eqn:feas_cond_final} reduces to
\begin{align}
\frac{K}{N_t}\LB \frac{\gamma}{1+\gamma} + \frac{ \epsilon \gamma}{1+ \epsilon \gamma} \RB < 1. \label{eqn:Wyner_Rg}
\end{align}
This condition allows us to examine the dependency of the feasible SINR target on various system parameters. 
In particular, we examine the dependency of the feasible $\gamma$ on $N_t$ by a simple numerical example. We plot the downlink transmit power (obtained by running the ROBF algorithm) as a function of $\gamma$ for different values of $N_t$ in Figure \ref{fig:feas_SINR}. We consider $K = 50$ UTs per cell and $\epsilon = 0.5$
It can be seen that beyond a certain cut off value of the target SINR, the downlink power grows unbounded.
This is precisely the value of the target SINR at which the ROBF algorithm becomes infeasible.
It can be verified that the cut off value of $\gamma$ is the one for which the condition in \eqref{eqn:Wyner_Rg} is not satisfied.
Finally, note that higher the number of transmit antennas per BS,  higher is the cut off value of the target SINR. Once again, this is due to the availability of greater number of spatial degrees of freedom. In this cases when $N_t \geq 2 K,$ any finite target SINR is achievable (note that $2K$ is the total number 
of UTs in two cells). 
\begin{figure}[!t]
 \includegraphics[width=3 in]{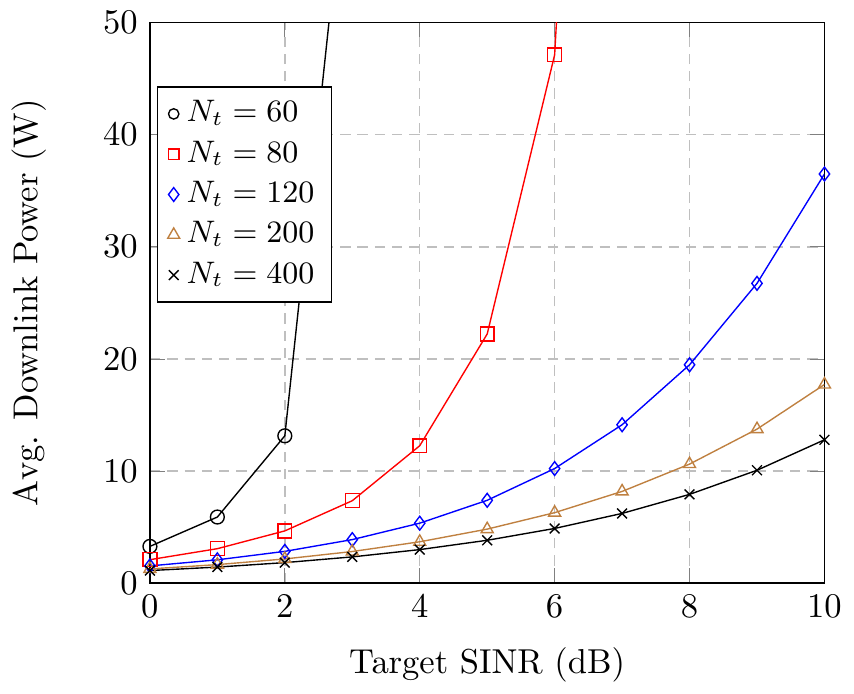} 
 \caption{Downlink power Vs target SINR for $K = 50$ UTs per cell.}
   \label{fig:feas_SINR} 
  \end{figure}
The feasibility conditions for the special case of two cell Wyner model 
was derived in \cite{Randa2010} as well. 

\subsection{Convergence of the Uplink and Downlink
SINR}
Next we focus on the achieved SINR in the uplink and downlink 
for the ROBF algorithm.
We first start with the analysis of the uplink of the ROBF algorithm.
Before we proceed, we make the following observations.
Recall that the achieved SINR in the uplink for the ROBF algorithm is given by
\begin{align}
\Lambda^{\text{UL}}_{i,j} (\muv) = \frac{\frac{\mu_{i,j}}{N_t}|\hat{\vv}^H_{i,j}\hv_{i,i,j}|^2}{\sum_{(n,k) \neq (i,j)} \frac{\mu_{n,k}}{N_t}|\hat{\vv}^H_{i,j}\hv_{i,n,k}|^2+ ||{\hat{\vv}}_{i,j}||_2^2}. \label{eqn:UL_form1}
\end{align}
Note that in our notation $\Lambda^{\text{UL}}_{i,j} (\muv),$
we have explicitly mentioned the achieved uplink SINR as a function of the parameters of the ROBF algorithm $\muv$
(slightly deviating from the 
notation for the uplink SINR introduced in Section \ref{sec:SysmodAlg}).
Since the uplink receive filter in the ROBF algorithm  \eqref{eqn:ULBF_asymp} 
are minimum mean square error (MMSE) form,
the expression for 
the achieved uplink SINR can be given in alternate form as
\begin{align}
\Lambda^{\text{UL}}_{i,j}(\muv) =  \frac{\mu_{i,j}}{N_t} \hv^H_{i,i,j}(\Sigmam^{{\prime}^{\mu}}_i+\Id_{N_t})^{-1}\hv_{i,i,j} \label{eqn:UL_form2}
\end{align}
where $\Sigmam^{\mu}_i = \sum_{i,j} \frac{\mu_{i,j}}{N_t}\hv_{i,i,j}\hv^H_{i,i,j}$ and $\Sigmam^{{\prime}^\mu}_i = \Sigmam^{\mu}_i-\frac{\mu_{i,j}}{N_t}\hv_{i,i,j}\hv^H_{i,i,j}.$ 
Similarly since the uplink receive filter \eqref{eqn:ULBF} in the CBF algorithm are MMSE, the achieved uplink SINR for the CBF algorithm
is given by
\begin{align}
\Lambda^{\text{UL}}_{i,j}(\lambdav) = \frac{ \lambda_{i,j}}{N_t} \hv^H_{i,i,j}(\Sigmam^{{\prime}^\lambda}_i+\Id_{N_t})^{-1}\hv_{i,i,j}  \qquad \forall i,j. \label{eqn:sinruse2}
\end{align}
Also, since the CBF algorithm is optimal, 
\begin{align}
\Lambda^{\text{UL}}_{i,j}(\lambdav) = \gamma_{i,j} \qquad \forall i,j.
\end{align}
Also, recall that the downlink SINR for the ROBF algorithm is given by
\begin{align}
\Lambda^{\text{DL}}_{i,j}(\muv) = \frac{|\vv^H_{i,j}\hv_{i,i,j}|^2}{\sum_{k \neq j}|\vv^H_{i,k}\hv_{i,i,j}|^2+\sum_{n \neq i,k}|\vv^H_{n,k}\hv_{n,i,j}|^2+N_0}. \label{eqn:DLSINR_ROBF}
\end{align}
We now provide the following convergence result:
\begin{theorem} 
\label{thm:ULDLSINRConv}
In the large system regime,
the achieved uplink and downlink SINR for the ROBF algorithm
converge almost surely to the target SINR $\gamma_{i,j}.$
Mathematically stating
\begin{align}
& \Lambda^{\text{UL}}_{i,j}(\muv)  \xrightarrow[N_t,K \to \infty]{\text{a.s.}} {\gamma}_{i,j} \qquad \forall i,j \label{eqn:ULSINRConv}, \\
&\Lambda^{\text{DL}}_{i,j}(\muv) \xrightarrow[N_t,K \to \infty]{\text{a.s.}} {\gamma}_{i,j} \qquad \forall i,j \label{eqn:DLSINRConv}.
\end{align}
\end{theorem}
\begin{proof}
The details of the proof for the convergence of the achieved uplink SINR \eqref{eqn:ULSINRConv} can be found in Appendix D.

Next we proceed to the convergence proof for the downlink SINR
\eqref{eqn:DLSINRConv}.
The proof utilizes the following lemma on
the convergence result for the downlink interference terms.
\begin{lemma}
\label{lem:DLInt}
The downlink interference term corresponding to the 
intra-cell and inter-cell interference converge 
in the large system regime
to the following:
\begin{align}
\sum_{(n,k) \neq (i,j)}|\vv^H_{n,k}\hv_{n,i,j}|^2 & \asymp \sum_{(n,k) \neq (i,j)} \frac{\bar{\delta}_{n,k}\bar{G}_{n,i,j} \bar{G}_{n,n,k} {\bar{m}}^\prime_n}{N_t} .
\end{align}
\end{lemma}
The lemma is proved in Appendix E.
We now proceed to the convergence of the downlink SINR.
Using Lemma \ref{lem:DLInt} and Lemma \ref{lem:peacock},
it can be concluded that the downlink SINR asymptotically
converges to
\begin{align}
& \Lambda^{\text{DL}}_{i,j}(\muv) \xrightarrow[N_t,K \to \infty]{\text{a.s.}} \nonumber \\
& \frac{\bar{\delta}_{i,j}  \sigma_{i,i,j} \bar{G}_{i,i,j} \bar{m}^2_{i} } {\frac{1}{N_t} \sum_{ (n,k) \neq (i,j)} \bar{\delta}_{n,k}  \bar{G}_{n,i,j} \bar{G}_{n,n,k} {\bar{m}}^\prime_n +N_0} \label{eqn:here1111}.
\end{align}
It can be easily verified from \eqref{eqn:lin_eq} 
that the right hand side of \eqref{eqn:here1111} is 
equal to the target SINR $\gamma_{i,j},$ thus completing the proof.
\end{proof}

\subsection{Asymptotic optimality of the Uplink and Downlink Power Allocation}
We now focus on the optimality of the uplink and downlink
power allocation of the ROBF algorithm in the large system regime. 

Consider the Lagrangian of the downlink minimization
problem in its two forms as in equations \eqref{eqn:LagDL}
and \eqref{eqn:LagUL}. Our proof proceeds by 
plugging in the solution obtained by the ROBF
algorithm $(\vv,\muv)$ into the Lagrangian 
and examining its properties in the large system 
regime. 
\begin{lemma}
\label{lem:LagConv}
The following results hold true for
the Lagrangian in the large system limit:
\begin{align}
& \lim_{N_t,K \to \infty} L(\vv,\muv) =
\lim_{N_t,K \to \infty} \sum_{i,j} \vv^H_{i,j} \vv_{i,j}, 
\label{eqn:LagDLConv}\\
& \lim_{N_t,K \to \infty}L(\vv,\muv)   =
\lim_{N_t,K \to \infty} \sum_{i,j}  \frac{\mu_{i,j} N_0}{N_t}. 
\label{eqn:LagULConv}
\end{align}.
\end{lemma}
\begin{proof}
Please refer to Appendix F.
\end{proof}
We can draw the following inference from the result of 
Lemma \ref{lem:LagConv}. \\

{\bf Corollary 2:}
The uplink and downlink power allocations yielded 
by the ROBF algorithm are equal in the asymptotic limit:
\begin{align}
\lim_{N_t,K \to \infty} \sum_{i,j} \vv^H_{i,j} \vv_{i,j} = 
\lim_{N_t,K \to \infty} \sum_{i,j}  \frac{\mu_{i,j} N_0}{N_t}.
\end{align}
%Consider the Lagrangian.
%\begin{align}
%\max_{\wv} \min_{\lambda} L (\wv,\lambdav) \geq \min_{\lambdav} \max_{\wv} L (\wv,\lambdav).
%\end{align}
We know that the uplink power allocation of the CBF algorithm
satisfies
\begin{align}
\sum_{i,j}  \frac{\lambda_{i,j} N_0}{N_t} = 
\min_{\wv} \max_{\lambda} L (\wv,\lambdav)
\end{align}
If we show that in the large system regime
\begin{align}
\lim_{N_t,K \to \infty} \sum_{i,j}  \frac{\mu_{i,j} N_0}{N_t} = 
\lim_{N_t,K \to \infty} \min_{\wv} \max_{\lambda} L (\wv,\lambdav),
\end{align}
and the duality gap is zero, then this implies that the solution provided by the ROBF is optimal in the asymptotic domain.
In other words, the optimal downlink power (which is the solution of the primal problem) is equal to $\lim_{N_t,K \to \infty} \min_{\wv} \max_{\lambda} L (\wv,\lambdav).$
In order to do so, we prove the following result:
\begin{lemma}
In the large system regime, the sum of uplink power allocation
of the ROBF algorithm converges to 
the sum of uplink power allocation of the CBF algorithm.
\begin{align}
\sum_{i,j} \frac{\mu_{i,j}}{N_t} \asymp \sum_{i,j} \frac{\lambda_{i,j}}{N_t} \qquad \forall i,j.
\end{align}
\end{lemma}
\begin{proof}
The result is proved in Appendix G.
\end{proof}
Consequently, the solution provided by the ROBF algorithm is optimal to $\lim_{Nt,K \to \infty} \max_{\wv} \min_{\lambda} L (\wv,\lambdav).$ Also, from the result of Corollary 2,
it follows that the downlink power allocation of the ROBF algorithm is also optimal in the large system regime.
This concludes the proof.

\section{Incorporating Individual BS Transmit Power Constraints}
\label{sec:PeakPower}
In this section, we consider a practical constraint in MIMO multi-cell systems, namely the individual BS transmit power constraints.
Recall that the basic optimization problem considered in this work stated in \eqref{eqn:OPT_basic} does not impose this constraint.
In what follows, we propose a heuristic extension of the ROBF algorithm to incorporate this constraint.

The optimization problem in \eqref{eqn:OPT_basic} along with the individual BS transmit power constraints can be stated as follows:
\beqa
 &\dsp  \min_{\wv_{i,j} \ \forall i,j} & \sum_{i,j}  \wv^H_{i,j}\wv_{i,j} \label{eqn:OPT_basic_maxpower}\\
& s.t. & \Lambda^{\text{DL}}_{i,j} \geq \gamma_{i,j} \qquad \forall i,j, \nonumber \\
& &  \sum_{j}  \wv^H_{i,j}\wv_{i,j} \leq P_{i,\max} \qquad \forall i, \nonumber
 \eeqa
 where $P_{i,\max}$ denotes the peak power of BS$_i.$ In what follows, we first state the heuristic extension of the ROBF algorithm to incorporate the individual BS transmit power constraints,
 and then provide the main intuition behind the development of the algorithm. 
\vspace{0.05in} \hrule
\vspace{0.01in}\hrule\vspace{0.05in}
\begin{algorithm}[ROBF with individual BS transmit power constraints] Perform the following steps.
\label{alg:decentralized_maxpow}
\begin{itemize}
\item[1.] Initialize $\tau = 0,$ and $\alpha_i(0) \geq 0 , \forall i$.
\item[2.] Starting from any initial $\mu^0_{i,j}(\tau) > 0 \ \forall i,j$ compute the uplink power allocation as ${\mu}_{i,j}(\tau) \defines \lim_{t \to \infty} \mu^t_{i,j}(\tau),$ where 
\begin{align}
\mu^{t+1}_{i,j}(\tau) =   \frac{\gamma_{i,j}}{\sigma_{i,i,j}\bar{m}^t_{i} (\tau)} \qquad \forall i,j \label{eqn:lam_stil_maxpow}
\end{align}
and $\bar{m}^t_i(\tau)$ is evaluated as  $\bar{m}^{t}_i(\tau) \defines \lim_{p \to \infty} \bar{m}^{t,p}_i (\tau)$  (initializing with any $\bar{m}^{t,0}_i(\tau) > 0, \forall i$)
\begin{align}
& \bar{m}^{t,p}_i(\tau) = \label{eqn:FPeqn_maxpow} \\ &  \LB \frac{1}{N_t}\sum_{n,k }  \frac{\sigma_{i,n,k} \mu^t_{n,k} (\tau)}{1+\sigma_{i,n,k}\mu^{t}_{n,k}(\tau) \bar{m}^{t,{p-1}(\tau)}_i}+1 + \alpha_i(\tau)\RB^{-1}. \nonumber
\end{align}
\item[3.]  Set the receive uplink filter as
\begin{align} \label{eqn:ULBF_asymp_maxpow}
& \hat{\vv}_{i,j}(\tau) = \\ &  \frac{1}{N_0\sqrt{N_t}}\Big( \sum_{n,k} \frac{\mu_{n,k}(\tau) }{N_t} \hv_{i,n,k} \hv^H_{i,n,k}+ (1 + \alpha_i(\tau)) \Id\Big)^{-1} \hv_{i,i,j}.
\end{align}
\item[4.]  Set the transmit downlink beamforming vectors as
$\vv_{i,j}(\tau) = \sqrt{\frac{\bar{\delta}_{i,j} (\tau)}{N_t}}\hat{\vv}_{i,j},$ where
the scaling factors $\bar{\delta}_{i,j}(\tau)$ are calculated as in  \eqref{eqn:lin_eq}.
\item[5.]  . Set $\tau = \tau+1$ and update $\alpha_i(\tau+1)$ as
\begin{align}
& \alpha_i(\tau+1) = \nonumber \\ & \LB \alpha_i(\tau) + \zeta \Big{(} \sum_{j} \vv^H_{i,j}(\tau) \vv_{i,j}(\tau) -  P_{i,\max}\Big{)} \RB^+ \ \forall i,
\label{eqn:alphaupdate}
\end{align}
where $\zeta > 0$ is a small step size. If $\alpha_i = 0, \forall i \in \{1,\dots,N\},$ then terminate. Else if
for all $ i \in \{1,\dots,N\},$ for which $\alpha_i > 0 $ if
\begin{align}
&  \Big{|}  \sum_{j} \vv^H_{i,j}(\tau) \vv_{i,j}(\tau) -  P_{i,\max} \Big{|} \leq \delta,
\end{align}
where $\delta >0$ is a small non zero quantity, then terminate. Else return to Step 2.
\end{itemize}
\end{algorithm}
\hrule 
\vspace{0.01in} \hrule \vspace{0.1in}
We now proceed to provide the main intuition behind this algorithm.

\subsection{ROBF with individual BS transmit power constraints: Intuition}
Consider the optimization problem in \eqref{eqn:OPT_basic_maxpower}. In order to solve this problem, we proceed by considering the Lagrangian associated with \eqref{eqn:OPT_basic_maxpower}, given by
\begin{align}
L^{\prime}(\wv,\lambdav,\alphav) & = \sum_{i,j}  \wv^H_{i,j} \wv_{i,j}  - \sum_{i,j} \frac{\lambda_{i,j}}{N_t} \Big{[}  \frac{|\wv^H_{i,j} \hv_{i,i,j}|^2}{\gamma_{i,j}} \nonumber \\ & -\sum_{(n,k) \neq (i,j)} |\wv_{n,k} \hv_{n,i,j}|^2 -N_0 \Big{]} \nonumber \\ 
& + \sum_i \alpha_i  \Big{[} \sum_{j}  \wv^H_{i,j}\wv_{i,j} - P_{i,\max} \Big{]}\label{eqn:LagDL_maxpow},
\end{align}
where $\alpha_i$ are the Lagrange multipliers associated with the individual BS transmit power constraints, and $\alphav = [\alpha_1,\dots,\alpha_N]^T$ ($\lambda_{i,j}$ has the same
interpretation as in \eqref{eqn:LagDL}). From \eqref{eqn:LagDL}, we can rewrite \eqref{eqn:LagDL_maxpow} as
\begin{align}
L^{\prime}(\wv,\lambdav,\alphav) = L(\wv,\lambdav) +  \sum_i \alpha_i  \Big{[} \sum_{j}  \wv^H_{i,j}\wv_{i,j} - P_{i,\max} \Big{]}. \label{eqn:LagRelation}
\end{align}
Using duality theory, the solution of the dual problem can be given by $\max_{\alphav,\lambdav} \min_{\wv}L^{\prime}(\wv,\lambdav,\alphav).$ 
%Note that
%\begin{align}
%\max_{\lambdav,\alphav} \min_{\wv} L^{\prime}(\wv,\lambdav,\alphav) = \max_{\alphav} \LSB \max_{\lambdav} \min_{\wv} L^{\prime}(\wv,\lambdav,\alphav) \RSB.
%\end{align}
Let us denote $$g(\alphav) = \max_{\lambdav} \min_{\wv} L^{\prime}(\wv,\lambdav,\alphav).$$
Our approach proceeds by showing that $\sum_{j}  \wv^H_{i,j}\wv_{i,j} - P_{i,\max}$ 
is a sub-gradient direction of the function $g(\alphav).$ Then, the dual problem 
$\max_{\alphav} [\max_{\lambdav} \min_{\wv}L^{\prime}(\wv,\lambdav,\alphav)]$ can be solved by 
updating the Lagrange multiplier  $\alphav$ in the direction of the sub-gradient \cite{BoydConvexOptbook}.

In order to do so, consider two vectors $\alphav^{(1)}$ and $\alphav^{(2)}.$ Note that
\begin{align*}
g(\alphav^{(1)}) = \max_{\lambdav} \min_{\wv} L^{\prime}(\wv,\lambdav,\alphav^{(1)}) = L^{\prime}(\wv^{(1)},\lambdav^{(1)},\alphav^{(1)}),
\end{align*}
where we use $(\wv^{(1)},\lambdav^{(1)})$ to denote the solution to $\max_{\lambdav} \min_{\wv} L^{\prime}(\wv,\lambdav,\alphav^{(1)}).$
Similarly, let $(\wv^{(2)},\lambdav^{(2)})$ denote the solution to $\max_{\lambdav} \min_{\wv} L^{\prime}(\wv,\lambdav,\alphav^{(2)}).$
Consider $g(\alphav^{(1)}).$ Firstly, it can be noted that
\begin{align}
\min_{\wv}  L^{\prime}(\wv,\lambdav,\alphav^{(1)}) \leq  L^{\prime}(\wv^{(2)},\lambdav,\alphav^{(1)}),
\end{align}
and hence 
\begin{align}
\max_{\lambdav}\min_{\wv}  L^{\prime}(\wv,\lambdav,\alphav^{(1)}) \leq \max_{\lambdav} L^{\prime}(\wv^{(2)},\lambdav,\alphav^{(1)}). \label{eqn:refheren5}
\end{align}
Using the definition of $g(\alphav^{(1)})$ in \eqref{eqn:refheren5}, we conclude that 
\begin{align}
g(\alphav^{(1)}) \leq \max_{\lambdav}  L^{\prime}(\wv^{(2)},\lambdav,\alphav^{(1)}). \label{eqn:refherep1}
\end{align}
Using \eqref{eqn:LagRelation} in the right hand side of \eqref{eqn:refherep1}, we obtain, 
\begin{align}
& g(\alphav^{(1)}) \leq \nonumber \\  &\max_{\lambdav}  L(\wv^{(2)},\lambdav,\alphav^{(1)}) +  \sum_i \alpha^{(1)}_i  \Big{[} \sum_{j}  (\wv^{(2)}_{i,j})^H\wv^{(2)}_{i,j} - P_{i,\max} \Big{]}. \label{eqn:refheren2}
\end{align}
Adding and subtracting the term $\sum_i \alpha^{(2)}_i \Big{[} \sum_{j}  (\wv^{(2)}_{i,j})^H\wv^{(2)}_{i,j} - P_{i,\max} \Big{]}$ to the right hand side
of \eqref{eqn:refheren2}, and rearranging the terms, we obtain
\begin{align}
&   g(\alphav^{(1)}) \leq \max_{\lambdav}  L^{\prime}(\wv^{(2)},\lambdav,\alphav^{(2)})  \nonumber\\
& +  \sum_i \LB \alpha^{(1)}_i -  \alpha^{(2)}_i \RB\Big{[} \sum_{j}  (\wv^{(2)}_{i,j})^H\wv^{(2)}_{i,j} - P_{i,\max} \Big{]} \nonumber \\
&   = g(\alphav^{(2)}) +  \sum_i \LB \alpha^{(1)}_i -  \alpha^{(2)}_i \RB\Big{[} \sum_{j}  (\wv^{(2)}_{i,j})^H\wv^{(2)}_{i,j} - P_{i,\max} \Big{]} \label{eqn:refheren3}
\end{align}
where in \eqref{eqn:refheren3}, we have used the definition of $g(\alphav^{(2)}).$
Consequently from \eqref{eqn:refheren3}, it can be concluded that  $\sum_{j}  (\wv^{(2)}_{i,j})^H\wv^{(2)}_{i,j} - P_{i,\max}$ is a sub-gradient direction 
of the function $g(\alphav)$ (by the definition of a sub-gradient). Therefore, updating the Lagrange multiplier $\alpha_i$ in the direction of the sub-gradient solves 
the dual problem $\max_{\alphav} [\max_{\lambdav} \min_{\wv}L^{\prime}(\wv,\lambdav,\alphav)]$ \cite{BoydConvexOptbook}. Therefore, $\alpha_i$ 
must be updated as
\begin{align}
& \alpha_i(\tau+1) = \nonumber \\ & \LB \alpha_i(\tau) + \zeta \Big{(} \sum_{j} \wv^H_{i,j}(\tau) \wv_{i,j}(\tau) -  P_{i,\max}\Big{)} \RB^+ \ \forall i,
\end{align}
where $\zeta > 0$ is a small step size.

Further, for each value of $\alphav(\tau)$ one needs to solve the problem 
$\max_{\lambdav} \min_{\wv} L^{\prime}(\wv,\lambdav,\alphav).$ This can be solved by
repeating the steps of the CBF algorithm for non-zero values of $\alphav$.

Recall that all the above arguments were made considering the case of finite system dimensions. We now follow the same
approach as in the development of ROBF algorithm, i.e., in the large system domain, utilize RMT results to obtain asymptotic approximations of the quantities
involved. The ROBF algorithm with per the BS peak power constraints developed in this section is obtained by following this idea.
The theoretical analysis of the performance of this algorithm along with the proof of convergence will be a topic of future research. 
Herein, we resort to the numerical results stated in
Section \ref{sec:SimRes} in order to show the performance as well as the convergence.

\section{Impact of Imperfect CSI and Pilot Contamination on the performance of ROBF Algorithm}
\label{sec:Pilot}
In this section, we investigate the impact of CSI estimation errors
and pilot contamination on the performance of the ROBF algorithm.
Throughout this section, we assume that the slow fading co-efficient (path loss information) can be accurately estimated at the BS 
(since they remain constant for 
a long period of time, they are easy to estimate, see for e.g. \cite{AshikminMarzetta2012}).
Further, for the fast fading co-efficients, we assume reciprocity between uplink and downlink channels,
and consider the time division duplexing (TDD) model of 
channel estimation, i.e., estimation via uplink pilots.

A comprehensive design of the optimal beamforming vectors in the presence of 
CSI estimation errors and pilot contamination issue is out of the scope of this paper. 
Alternately, we take the following approach: First, we assume that the BS treats the CSI estimate as the true
CSI and implements the ROBF algorithm directly. In this case, we derive the asymptotic equivalent 
of the downlink SINR achieved by the ROBF algorithm. Using numerical results, we investigate the impact of CSI estimation errors 
and pilot contamination on the performance of the ROBF algorithm both in terms of the achieved SINR and downlink power.
Then, exploiting the fact that the BS has the accurate knowledge of the slow fading co-efficient, we propose a heuristic adaptation of the ROBF algorithm in the presence of imperfect CSI and pilot contamination named as the 
modified ROBF (MROBF) algorithm. We show that under the massive MIMO regime, an algorithm in which parameters can be 
computed based on the channel statistics (rather than the fast fading CSI)  such as the MROBF algorithm 
is more robust to CSI estimation and pilot contamination effects.

\subsection{CSI Estimation}
We now describe the CSI training phase for the estimation of
the fast fading co-efficients.
Let $T$ be the length of the channel coherence interval,
a part of which is dedicated for CSI estimation, and $T_{\text{Tr}}$ be the number of symbols used for pilots.
Therefore, the UTs in every cell $i$ transmit $T_{\text{Tr}}$ mutually 
orthogonal pilot symbols to their respective BSs during the training phase.
We represent the pilot sequences used by the $K$ 
UTs in each cell by the matrix $\sqrt{P_{\text{Tr}}} \Phim \in \CC^{T_{\text{Tr}} \times K}, \ (T_{\text{Tr}} \geq K).$ The pilot sequences are repeated in each cell, hence leading
to the issue of pilot contamination.
The matrix $\Phim$ satisfies $\Phim^H \Phim = \Id_K.$

The signal received during the channel training phase denoted by 
$\Ym_{i,\text{Tr}}  \in \CC^{N_t \times K}$ can be written as 
\begin{align}
\Ym_{i,\text{Tr}} =  \Hm_{i,i}\Phim^T+ \big{(} \sum_{n \neq i} \Hm_{i,n}  \big{)} \Phim^T  + \Nm,
\end{align}
where 
$\Hm_{i,n} = [\hv_{i,n,1},\dots,\hv_{i,n,K}]$ and
$\Nm \in \CC^{N_t \times T_{\text{Tr}}}$ with i.i.d. $\mathcal{CN}(0,1)$
elements represents the noise during channel training phase.
The MMSE estimate of $\hv_{i,n,k}$ given $\Ym_{i,\text{Tr}}$
can be given by \cite{MMSE1993}
\begin{align}
\hat{\hv}_{i,n,k} = {\sigma}^\prime_{i,n,k}\LB \sum^N_{b = 1} \hv_{i,b,k} + \frac{\nv_{\text{Tr}}}{\sqrt{P_{\text{Tr}}}} \RB, \label{eqn:MMSEEstimate}
\end{align}
where
${\sigma}^\prime_{i,n,k} = \sigma_{i,n,k} \LB \sum^N_{b = 1} \sigma_{i,b,k} + \frac{1}{P_{\text{Tr}}} \RB^{-1}.$
It can be verified that $\hat{\hv}_{i,n,k}$ is distributed as
\begin{align}
\hat{\hv}_{i,n,k} \sim \mathcal{CN} \LB 0, \sigma^2_{i,n,k} \LB \sum^N_{b = 1} \sigma_{i,b,k} + \frac{1}{P_{\text{Tr}}} \RB^{-1} \RB.
\end{align}
For notational simplicity, let us denote
$\hat{\sigma}_{i,n,k} = \sigma^2_{i,n,k} \LB \sum^N_{b = 1} \sigma_{i,b,k} + \frac{1}{P_{\text{Tr}}} \RB^{-1}.$
Note that
\begin{align}
\hat{\sigma}_{i,n,k} = {\sigma}^\prime_{i,n,k} \sigma_{i,n,k}. \label{eqn:tilderelation}
\end{align}

\subsection*{ROBF Algorithm With CSI Estimates}
Throughout this subsection, we assume that
the BSs assume the CSI estimates to the true channel vector,
and implement the ROBF algorithm. Note however
that the computation of $\mu_{i,j}$ and $\bar{\delta}_{i,j}$ in \eqref{eqn:lam_stil} and
\eqref{eqn:lin_eq} do not require the estimates of the fast fading co-efficient, and hence they 
can be implemented directly.
%Therefore, these two parameters can be computed accurately (as compared to the case with perfect CSI).
The uplink receive filter with imperfect CSI denoted by $\hat{\vv}^{\text{est}}_{i,j}$ can be 
formulated as
\begin{align}
\hat{\vv}^{\text{est}}_{i,j} = \sqrt{\frac{1}{N_t}}  \Psi^{-1}_i \hat{\hv}_{i,i,j} \label{eqn:MMSE}
\end{align}
where
\begin{align}
\Psi_i = \sum_{n,k} \frac{\mu_{i,k}}{N_t}\hat{\hv}_{i,n,k} \hat{\hv}^H_{i,n,k} + \Id_{N_t} \label{eqn:PsiDef}.
\end{align}
%\nonumber \\ = \sum^K_{k = 1} \mu_{i,k}\hat{\xv}_{i,i,k} \hat{\xv}^H_{i,i,k} + \Id_{N_t}.  
The dowlink beamforming vector denoted by ${\vv}^{\text{est}}_{i,j}$ can be then computed as
\begin{align}
{\vv}^{\text{est}}_{i,j} = \sqrt{\frac{\bar{\delta}_{i,j}}{N_t}}\hat{\vv}^{\text{est}}_{i,j}. \label{eqn:RZFBF}
\end{align}

We now investigate the achieved SINR in the downlink 
under the ROBF algorithm with imperfect CSI. The expression
for the dowlink SINR can be given as in \eqref{eqn:DLSINR_ROBF}, by replacing 
$\vv_{i,j}$ with $\hat{\vv}^{\text{est}}_{i,j}.$ The asymptotic equivalent of the dowlink SINR can be obtained by analyzing it in the large system regime. 
\begin{theorem}
The achieved SINR in the downlink by the ROBF algorithm in the 
presence of imperfect CSI converges almost surely to
the right hand side of \eqref{eqn:DLSINRImpCSIConv} in the large
system regime, where the term
$\hat{G}_{n,n,k}$ is defined as 
\begin{align}
\hat{G}_{n,n,k} = \frac{1}{1+\xi_{n,k} \hat{\sigma}_{n,n,k} \bar{m}^{\text{est}}_n} \ \text{and} \
\xi_{i,k} = \frac{ \sum^N_{n = 1} \sigma_{i,n,k} \mu_{n,k} } { \sigma_{i,i,k}} .
\end{align}
Further, $\bar{m}^{\text{est}}_n$ can be computed as the solution 
to the fixed point equation 
\begin{align}
\bar{m}^{\text{est}}_n = \LB \frac{1}{N_t} \sum^K_{k = 1} \frac{\xi_{n,k}  \hat{\sigma}_{n,n,k}}{1+\xi_{n,k} \hat{\sigma}_{n,n,k} \bar{m}^{\text{est}}_n}+1\RB^{-1}, \label{eqn:StilImpCSI}
\end{align}
and $(\bar{m}^{\prime}_n)^{\text{est}}$ can be computed from 
$\bar{m}^{\text{est}}_n$
as
\begin{align}
(\bar{m}^{\prime}_n)^{\text{est}} = \frac{(\bar{m}^{\text{est}}_n)^2}{1-\frac{1}{N_t}  \sum^K_{k=1} \frac{(\hat{\sigma}_{n,n,k} \xi_{n,k}\bar{m}^{\text{est}}_n)^2}{(1+\hat{\sigma}_{n,n,k} \xi_{n,k}\bar{m}^{\text{est}}_n)^2}}. \label{eqn:stilderImpCSI}
\end{align}
\end{theorem}
\begin{proof}
The proof is provided in Appendix H.
\end{proof}

\begin{figure*}
\begin{align}
& (\Lambda^{\text{DL}}_{i,j}(\muv))_{\text{est}}  \asymp \frac{\bar{\delta}_{i,j} (\hat{\sigma}_{i,i,j} \hat{G}_{i,i,j} \bar{m}^{\text{est}}_i)^2 }{\sum^N_{ \substack{n = 1 \\ n \neq i}} \bar{\delta}_{n,j}  \LB \sigma_{n,i,j} {\sigma}^\prime_{n,n,j} \hat{G}_{n,n,j} \bar{m}^{\text{est}}_n \RB^2 + \frac{1}{N_t}\sum^N_{n = 1} \sum^K_{ \substack{k = 1 \\ k \neq j}} \bar{\delta}_{n,k} \hat{\sigma}_{n,n,k}  \hat{G}^2_{n,n,k} (\bar{m}^{\prime}_n)^{\text{est}} B_{n,k} }, \label{eqn:DLSINRImpCSIConv} \\
& \text{where} \nonumber \\
& B_{n,k} = \sigma_{n,i,j}  + \hat{\sigma}_{n,n,j}  ( \xi_{n,j}  {\sigma}_{n,i,j}  {\sigma}^\prime_{n,n,j} \hat{G}_{n,n,j} {\bar{m}^{\text{est}}}_n)^2  
- 2 \mu_{n,j}({\sigma}_{n,i,j}  {\sigma}^\prime_{n,n,j})^2 \hat{G}_{n,n,j} {\bar{m}^{\text{est}}}_n. \\
& \frac{ \mu_{i,j}(\hat{\sigma}_{i,i,j} \hat{G}_{i,i,j} \bar{m}^{\text{est}}_i)^2 }{\sum^N_{ \substack{n = 1 \\ n \neq i}} \mu_{n,j}\LB \sigma_{i,i,j} {\sigma}^\prime_{i,n,j} \hat{G}_{i,i,j} \bar{m}^{\text{est}}_i \RB^2 + \frac{1}{N_t}\sum^N_{n = 1} \sum^K_{ \substack{k = 1 \\ k \neq j}}  \mu_{n,k}\hat{\sigma}_{i,i,j}  \hat{G}^2_{i,i,j} (\bar{m}^{\prime}_i)^{\text{est}} B^\prime_{n,k} + \hat{\sigma}_{i,i,j} \hat{G}^2_{i,i,j} (\bar{m}^{\prime}_i)^{\text{est}} } , 
 \label{eqn:ULSINRImpCSIConv} \\
 & \text{where} \nonumber \\
& B^\prime_{n,k} = \sigma_{i,n,k}  + \hat{\sigma}_{i,i,k}  ( \xi_{i,k}  {\sigma}_{i,n,k}  {\sigma}^\prime_{i,i,k} \hat{G}_{i,i,k} {\bar{m}^{\text{est}}}_i)^2  
- 2 \mu_{i,k}({\sigma}_{i,n,k}  {\sigma}^\prime_{i,i,k})^2 \hat{G}_{i,i,k} {\bar{m}^{\text{est}}}_i.
\end{align}
\end{figure*}

\subsection{Modified ROBF (MROBF) Algorithm}
In this subsection, we propose a heuristic adaptation of the ROBF algorithm addressed as Modified ROBF (MROBF) algorithm
which is designed to accommodate the effects of imperfect CSI and pilot contamination. 

First, note that naive application of the ROBF algorithm may not yield good performance in 
the presence of imperfect CSI and pilot contamination. The following reasoning provides an intuitive understanding for this:
Consider the interference arising from the signal of $UT_{n,k},$ ($n \neq i, k \neq j)$ at the UT$_{i,j}$ (which
does not use the same pilot as UT$_{i,j}$), i.e., the term $|\wv_{n,k} \hv_{n,i,j}|^2.$
Following the derivation of Appendix H, it can be verified that this interference terms converges to 
 \begin{align}
 |\wv_{n,k} \hv_{n,i,j}|^2 \asymp \frac{\bar{\delta}_{n,k}}{N_t}\hat{\sigma}_{n,n,k}  \hat{G}^2_{n,n,k} (\bar{m}^{\prime}_n)^{\text{est}} B_{n,k}.
 \label{eqn:IntwoPilot}
 \end{align}
Now consider the interference arising from the signal of $UT_{n,j},$ ($n \neq i)$ at the UT$_{i,j}$ (which
reuses the same pilot as UT$_{i,j}$), i.e., the term $|\wv_{n,j} \hv_{n,i,j}|^2.$
It can be verified that asymptotically, this interference terms converges to 
 \begin{align}
 |\wv_{n,j} \hv_{n,i,j}|^2 \asymp \bar{\delta}_{n,j} \LB \sigma_{n,i,j} {\sigma}^\prime_{n,n,j} \hat{G}_{n,n,j} \bar{m}^{\text{est}}_n \RB^2 \label{eqn:IntwithPilot}.
 \end{align}
We note that the two interference terms are significantly different. In particular, \eqref{eqn:IntwoPilot}  is scaled by a factor of $\frac{1}{N_t},$
(and for large $N_t,$ this has a very low value). However, the term in \eqref{eqn:IntwithPilot} is not scaled. 
This is due to the fact that the BS cannot distinguish between the channels of 
UT$_{i,j}$ and UT$_{n,j}$ due to the pilot contamination effect.
This implies that naive application of the ROBF algorithm may significantly underestimate the
interference arising out of the UTs that reuse the same pilots. The algorithm performance can be enhanced by 
carefully accounting for these issues. This is indeed the main intuition behind the MROBF algorithm.

In the MROBF algorithm, we consider that the BS does not alter the structure of the beamforming 
vector, i.e., the BS retains the RZF beamforming vector as in \eqref{eqn:RZFBF}.
Nevertheless, the performance gain can still be obtained by redesigning the uplink power allocation,
and computation of $\bar{\deltav}.$ We redesign the uplink power allocation on similar lines as that of the ROBF algorithm.

First note that in the case of perfect CSI, the achieved uplink SINR in the asymptotic limit is given by
$\sigma_{i,i,j} \mu_{i,j} \bar{m}_{i}.$
Further, the uplink power allocation is chosen to satisfy $\gamma_{i,j} = \sigma_{i,i,j} \mu_{i,j} \bar{m}_{i},$ or
\begin{align}
\mu_{i,j} =   \frac{\gamma_{i,j}}{\sigma_{i,i,j}\bar{m}_{i}} \qquad \forall i,j \label{eqn:mucomp}.
\end{align}

We use a similar argument for the computation of the uplink power allocation in the case of imperfect CSI.
In order to do so, consider the uplink SINR with uplink receive filter formulated as in \eqref{eqn:MMSE} 
\begin{align}
(\Lambda^{\text{UL}}_{i,j}(\muv))_{\text{est}} = \frac{\frac{\mu_{i,j}}{N_t}|\hat{\vv}^{\text{est}}_{i,j}\hv_{i,i,j}|^2}{\sum_{(n,k) \neq (i,j)}\frac{\mu_{n,k}}{N_t}|\hat{\vv}^{\text{est}}_{i,j}\hv_{i,n,k}|^2+ ||\hat{\vv}^{\text{est}}_{i,j}||_2^2} \label{eqn:ULSINR_MROBF}.
\end{align}
Let us examine the uplink SINR in the large system domain. 
Using the derivation similar to Appendix H, by replacing the individual terms of \eqref{eqn:ULSINR_MROBF} by their asymptotic equivalents,
the uplink SINR in the large system limit can be approximated by \eqref{eqn:ULSINRImpCSIConv}\footnote{Note the this approach does not constitute a formal proof of convergence 
of the uplink SINR. Since the algorithm developed itself is heuristic, the proof of convergence has been omitted.}.
For convenience, let us denote that denominator of  \eqref{eqn:ULSINRImpCSIConv} by $\mathcal{I}^{\text{UL}}_{\text{asymp}}.$
Following the same approach as in \eqref{eqn:mucomp}, we compute the uplink power allocation as the solution 
to the following set of equations:
\begin{align}
 \mu_{i,j} = \frac{\gamma_{i,j} \mathcal{I}^{\text{UL}}_{\text{asymp}}}{(\hat{\sigma}_{i,i,j} \hat{G}_{i,i,j} \bar{m}^{\text{est}}_i)^2} \qquad \forall i,j. \label{eqn:mucompImp}
\end{align}
The $\mu_{i,j}$ that satisfies the set of equations \eqref{eqn:mucompImp} can be computed using the following iterative method: 
Starting from any initial $\mu^0_{i,j} > 0 \ \forall i,j$ the uplink power allocation is given by ${\mu}_{i,j} \defines \lim_{t \to \infty} \mu^t_{i,j},$ where 
\begin{align}
 \mu^{t+1}_{i,j} = \frac{\gamma_{i,j} (\mathcal{I}^{\text{UL}}_{\text{asymp}})^t }{(\hat{\sigma}_{i,i,j} \hat{G}^t_{i,i,j} (\bar{m}^{\text{est}}_i)^t)^2} \qquad \forall i,j.  \label{eqn:IterationsMROBF}
\end{align}
In \eqref{eqn:IterationsMROBF}, $(\mathcal{I}^{\text{UL}}_{\text{asymp}})^t, \hat{G}^t_{i,i,j}$ and  $(\bar{m}^{\text{est}}_i)^t$ denote the respective
quantities computed at $\mu^t_{i,j}, \ \forall i,j.$ We observe that numerically that the iterations of \eqref{eqn:IterationsMROBF}
converges to the solution of \eqref{eqn:mucompImp} (see numerical results in Section \ref{sec:SimRes}).
%\begin{figure*}
%\begin{align}
%& (\Lambda^{\text{DL}}_{i,j}(\muv))_{\text{est}}  \asymp \frac{\bar{\delta}_{i,j} (\hat{\sigma}_{i,i,j} \hat{G}_{i,i,j} \bar{m}^{\text{est}}_i)^2 }{\sum^N_{ \substack{n = 1 \\ n \neq i}} \bar{\delta}_{n,j}  \LB \sigma_{n,i,j} \tilde{\sigma}_{n,n,j} \hat{G}_{n,n,j} \bar{m}^{\text{est}}_n \RB^2 + \frac{1}{N_t}\sum^N_{n = 1} \sum^K_{ \substack{k = 1 \\ k \neq j}} \bar{\delta}_{n,k} \hat{\sigma}_{n,n,k}  \hat{G}^2_{n,n,k} (\bar{m}^{\prime}_n)^{\text{est}} B_{n,k} }, \label{eqn:DLSINRImpCSIConv} \\
%& \text{where} \nonumber \\
%& B_{n,k} = \sigma_{n,i,j}  + \hat{\sigma}_{n,n,j}  ( \mu_{n,j}  {\sigma}_{n,i,j}  \tilde{\sigma}_{n,n,j} \hat{G}_{n,n,j} {\bar{m}^{\text{est}}}_n)^2  
%- 2 \mu_{n,j}({\sigma}_{n,i,j}  \tilde{\sigma}_{n,n,j})^2 \hat{G}_{n,n,j} {\bar{m}^{\text{est}}}_n.
%\end{align}
%\end{figure*}

Finally, $\bar{\deltav}$ can be computed as a solution to the following linear equations
\begin{align}
\bar{\deltav}  = \Deltam^{-1} \gammav \label{eqn:DLAdaptMROBF}
\end{align}
where the matrix $\Deltam \in \CC^{NK \times NK} $ is defined as
\begin{equation}
\Deltam = \left(
\begin{array}{ccc}
\Deltam^{1,1} & \ldots & \Deltam^{1,N} \\
\vdots  & \ddots & \vdots \\
\Deltam^{N,1} & \ldots & \Deltam^{N,N}\\
\end{array} \right)
\end{equation}
where each submatrix $\Deltam^{i,j} \in \CC^{K \times K}$ is given by 
\begin{equation} 
\Deltam^{i,n}_{j,k} \defines \begin{cases} (\hat{\sigma}_{i,i,j} \hat{G}_{i,i,j} \bar{m}^{\text{est}}_i)^2 , & n=i ,\ k = j \\
                     -\LB \sigma_{n,i,j} {\sigma}^\prime_{n,n,j} \hat{G}_{n,n,j} \bar{m}^{\text{est}}_n \RB^2 ,& n \neq i,  k = j \\
                     \frac{-1}{N_t} \hat{\sigma}_{n,n,k}  \hat{G}^2_{n,n,k} (\bar{m}^{\prime}_n)^{\text{est}} B_{n,k} ,& n = i,  k \neq j \\  & \text{and}  \ n \neq i,  k \neq j.
                   \end{cases} 
\end{equation}

\section{Numerical Results}
\label{sec:SimRes}
In this section, we present some numerical results to demonstrate the performance of the ROBF algorithm in a massive MIMO setting.

\begin{figure}[!t]
 \includegraphics[width=3 in]{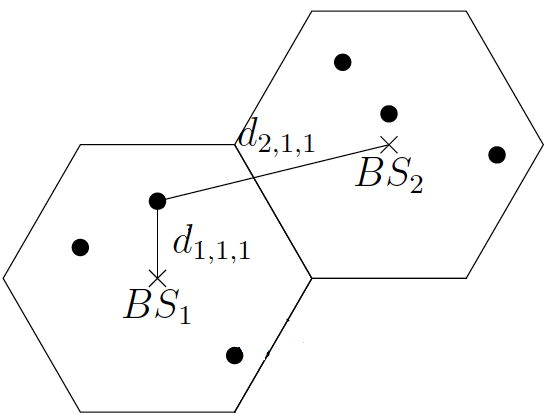} 
 \caption{Hexagonal cellular network consisting of $2$ cells.}
   \label{fig:2CellNwk} 
  \end{figure}
  
  \begin{figure*}[!t]
\begin{minipage}[t]{3.4 in}
\centering
\includegraphics[width=3.6 in, height  = 2.7 in]{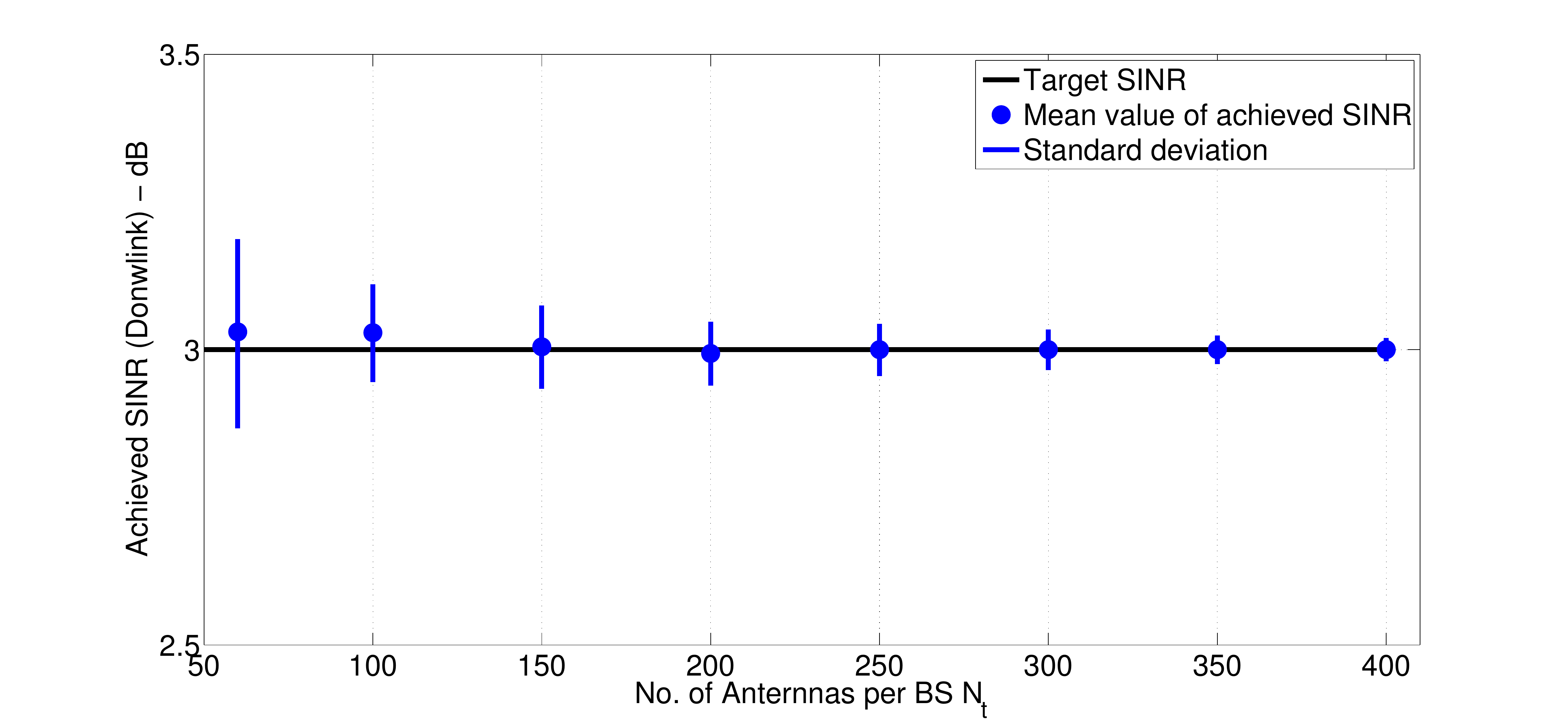}  
\caption{Fluctuations of the downlink SINR (ROBF algorithm) around the target value. $K = 25$ UTs per cell and target SINR = $3$ dB per UT.}
\label{fig:DLSINR_Variance}
\end{minipage}%\hfill
\begin{minipage}[t]{3.4 in}
\centering
\includegraphics[width=3.6 in, height  = 2.7 in]{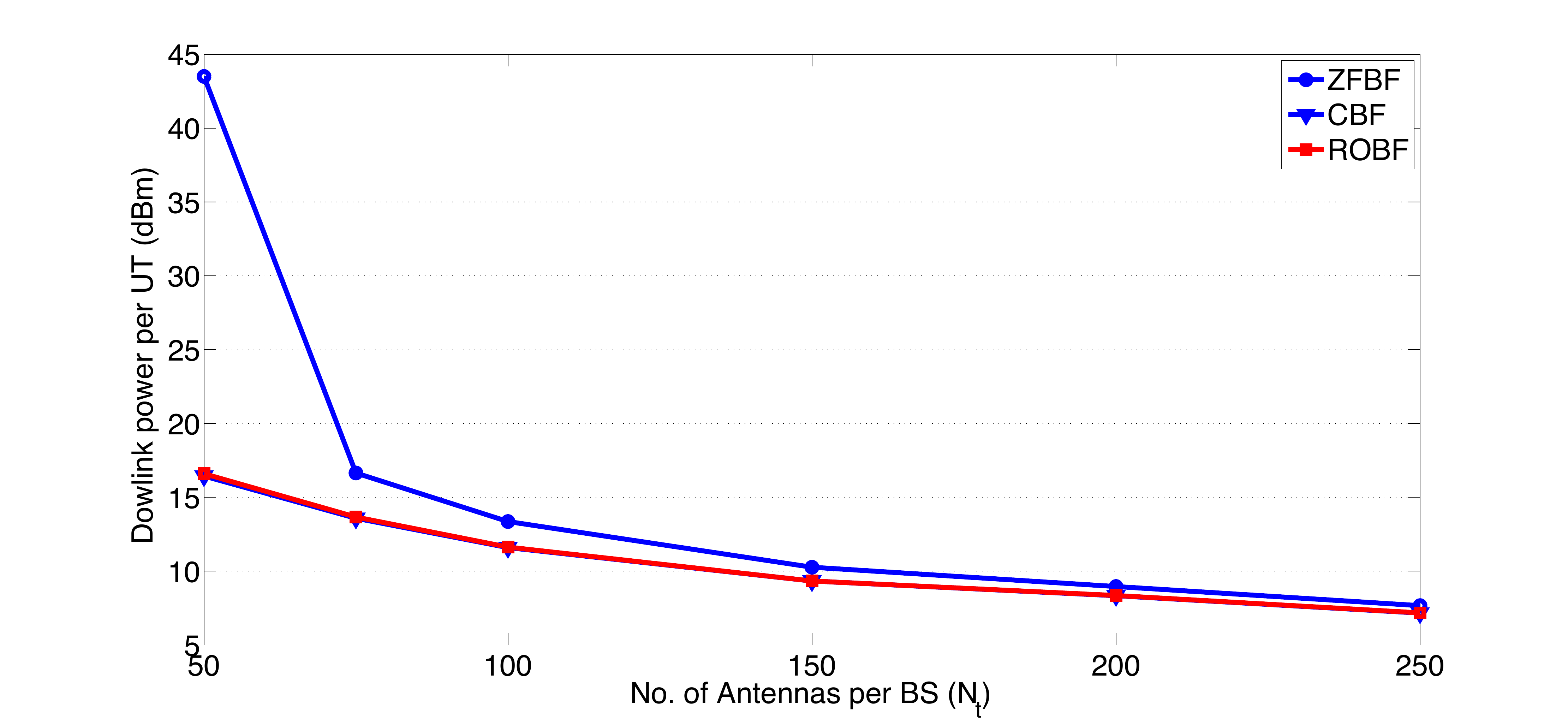}  
\caption{Comparison of downlink power per UT as a function of the number of antennas per BS. $K = 25$ UTs per cell and target rate = $3$ bits/s/Hz ($\log(1+\gamma_{i,j})$) per UT.}
\label{fig:ZFBFVsROBF_Nt}
\end{minipage}
\begin{minipage}[t]{3.4 in}
\centering
\includegraphics[width=3.6 in, height  = 2.7 in]{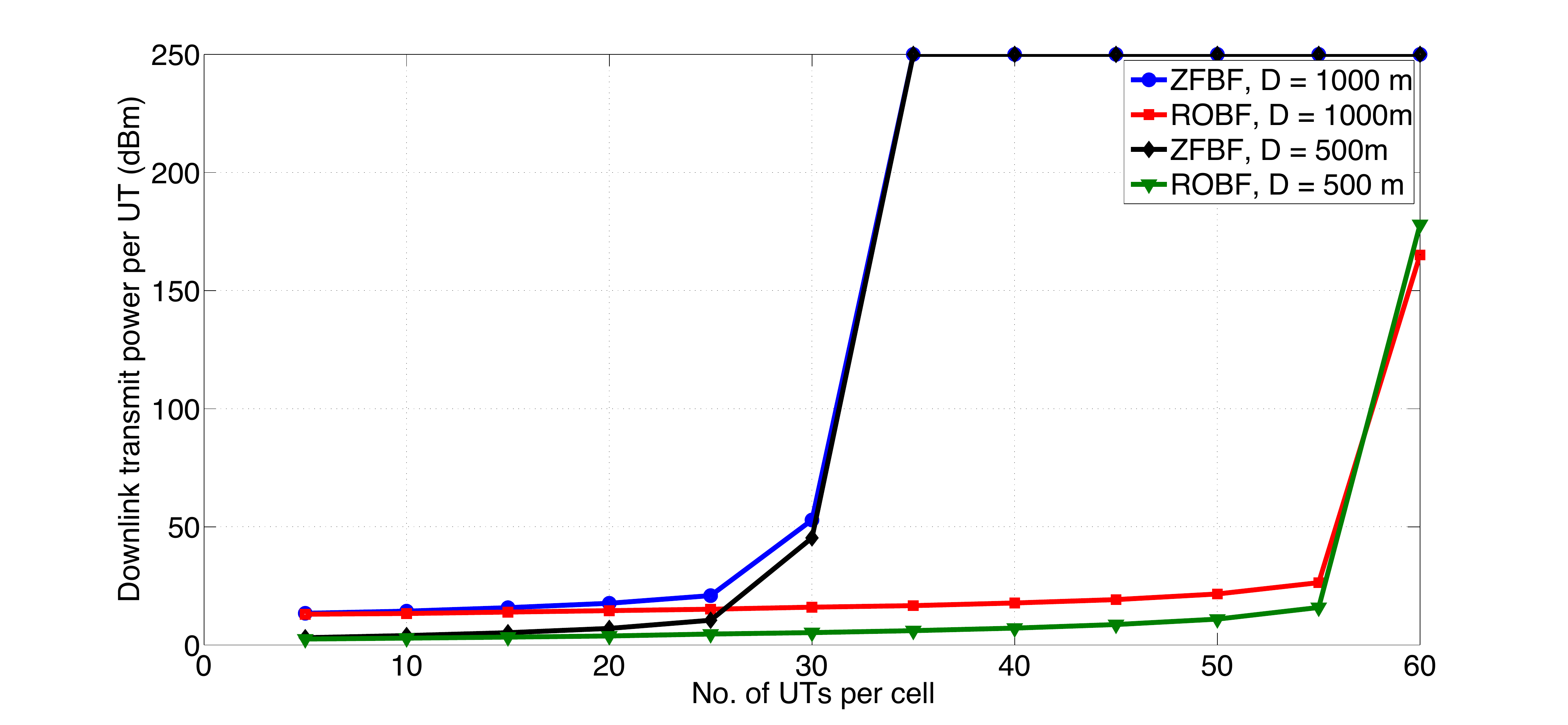}  
\caption{Comparison of downlink power per UT as a function of the number of UTs per cell. $N_t = 60$ antennas/ BS and target rate = $3$ bits/s/Hz ($\log(1+\gamma_{i,j})$) per UT. $D$ represents the distance between the two BSs.}
\label{fig:ROBFVsZFBF_Pathloss_60UTs}
\end{minipage}
\begin{minipage}[t]{3.4 in}
\centering
 \includegraphics[width=3.6 in, height  = 2.7 in]{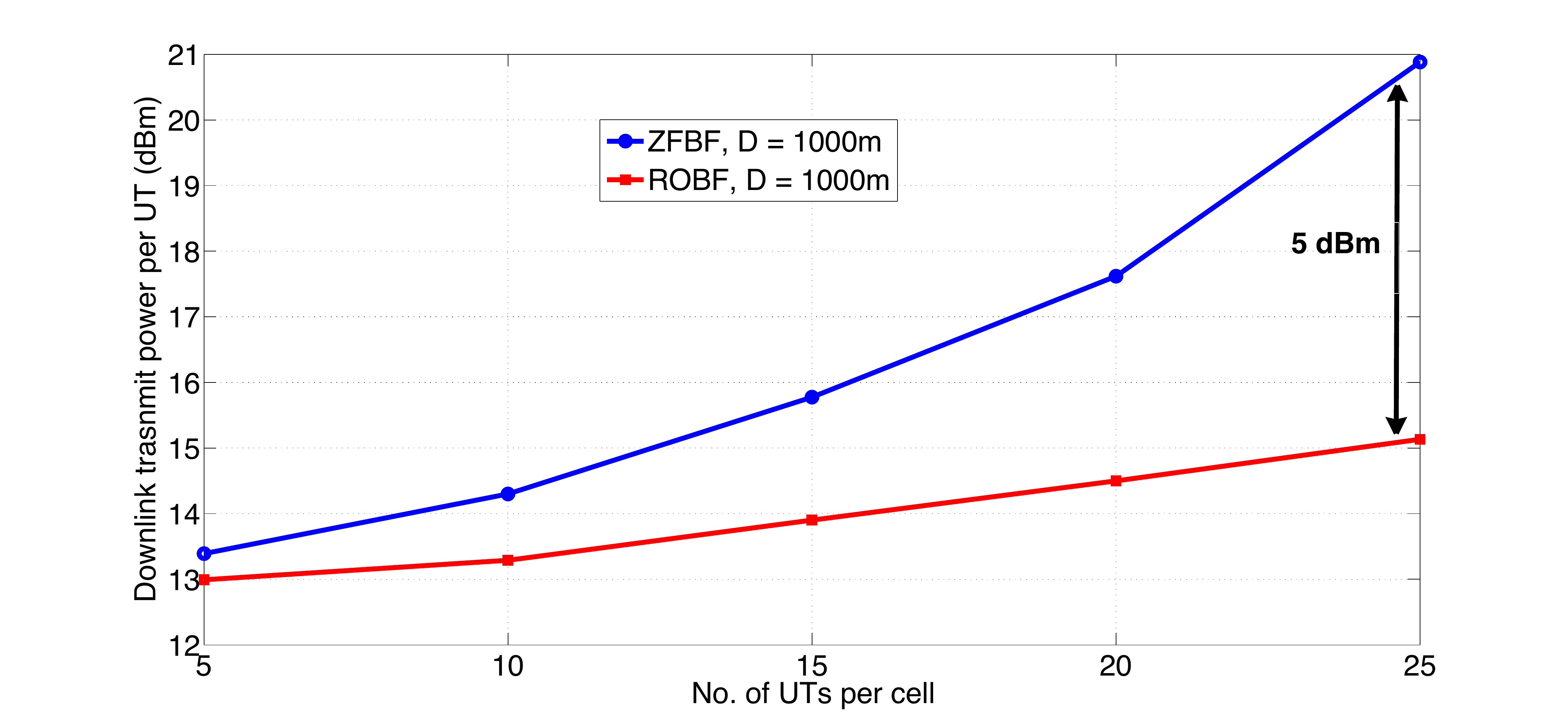}  
\caption{\ \ Comparison of downlink power per UT as a function of the $K$. \ \ Settings identical to Figure \ref{fig:ROBFVsZFBF_Pathloss_60UTs}.}
\label{fig:WithArrowROBFVsZFBF_Pathloss_25UTs}
\end{minipage}
\end{figure*}
We consider a hexagonal cellular system consisting of $2$ cells as shown in Figure~\ref{fig:2CellNwk}.
We use the distance dependent path loss model in which the
path loss from BS $i$ to UT$_{j,k}$ is given by
$$\sigma_{i,j,k} = \frac{d_0}{{d_{i,j,k}}^{\beta}},$$ 
where $d_{i,j,k}$ represents the distance between BS $i$ to UT$_{j,k}.$
$\beta$ represents the path loss exponent which is taken to be $3.6$
in all the simulation scenarios. $d_0$ represents the channel attenuation at reference point and is taken to $10^{-3.53}.$
Location of the UTs are obtained by generating uniform 
random numbers inside each hexagonal cell.
The distance $d_{\min} \leq d_{i,j,k} \leq d_{\max},$ where
$d_{\min}$ and $d_{\max}$ represent the minimum and the maximum distance 
between the UTs to the BSs of their respective cells. In our simulation, 
$d_{\min} = 20$ m and $d_{\max} = \{500,1000\}$ m depending on the distance between the two BSs. The noise power is taken to be $-104$ dBm over the operating bandwidth. All numerical results are plotted 
by varying the positions of the UTs inside the cell over $500$ iterations. 

First, we examine the performance of the ROBF algorithm in satisfying 
the UT SINR constraints. 
Accordingly, we plot the variation of the uplink and downlink SINR (averaged across the UTs) as a function of the number of antennas per 
BS for $1000$ channel realizations in Figure \ref{fig:DLSINR_Variance}. 
Herein, $K = 50$ UTs per cell and target SINR is $3$ dB per UT.
The horizontal line represents the target SINR which is $3$ dB. 
The bubbles represent the average achieved SINR values (averaged over the channel realizations), and 
the vertical lines around this bubble represent the variation of the achieved downlink SINR around the average value.
It can be observed that even for moderate number of antennas, e.g. $60$ antennas per BS (comparable to the number of UTs), the target SINR constraints are satisfied for almost every channel realization (since the fluctuations are small). This implies that the ROBF algorithm nearly optimal under this setting in terms of satisfying the SINR constraints. 
%Further, although insignificant, it can be seen that the variation of the downlink SINR is on a slightly greater range a compared to the uplink SINR. 
%This is due the the fact that a greater number of parameters need to be calculated (specifically the additional scaling parameter $\delta_{i,j}$ must be approximated using the RMT results) to compute the
%downlink beamforming vectors, which gives rise to the greater variation.

Next, we investigate the downlink power of the ROBF algorithm and compare it with the CBF algorithm and ZF beamforming\footnote{Note that for eigen beamforming and RZF, after fixing the beamforming direction, the power allocation needed to satisfy the SINR constraints does not yield to a simple structure (as in the case of ZF). In order to avoid the additional complexities associated with power allocation, only ZF is used for comparison purposes. Nevertheless, we expect the gains of the ROBF algorithm to hold true with respect to the other beamforming techniques as well.} (denoted by $\sqrt{P^{\text{ZF}}_{i,j}} \wv^{\text{ZF}}_{i,j}$, where $\wv^{\text{ZF}}_{i,j}$ is unit norm vector).
Note that our simulation setting, 
the BS performs ZF beamforming to null the interference UTs in both the cells (and not merely the UTs in its own cells). After nulling the interference, 
BS performs appropriate power allocation in order to meet the target SINR constraint of the UTs $$P^{\text{ZF}}_{i,j} = \frac{\gamma_{i,j} N_0 }{|\hv_{i,i,j} \wv^{\text{ZF}}_{i,j} |^2}.$$ 
We plot the downlink power per UT (i.e. $\frac{\text{Sum downlink power}}{K}$) as a function of the number of antennas per BS in Figure \ref{fig:ZFBFVsROBF_Nt}. Herein, $K = 25$ UTs per cell and target rate = $3$ bits/s/Hz ($\log(1+\gamma_{i,j})$) per UT.
The following conclusions can be drawn. First, the downlink power expended 
by ROBF very closely matches that of CBF, indicating that the ROBF is nearly optimal in terms of minimizing the downlink power as well.
Next, it can be seen that for moderate number of antennas per BS, e.g. $50-100$ antennas, ROBF provides substantial power gains 
as compared to ZF beamforming. This result highlights the gains obtained by optimizing the power allocation (as compared to naively nulling out interference to all the UTs in the system). When
$N_t$ is very large, the scaling of $\frac{1}{N_t}$ starts to play a dominant role, and ROBF ceases to provide substantial gains over ZF. 

Finally, we examine the efficacy of ROBF algorithm in terms of its ability to support greater number of UTs per cell. This is accomplished by fixing the number of antennas per BS to $60,$ and plotting the downlink power as a function of the number of UTs per cell in Figure \ref{fig:ROBFVsZFBF_Pathloss_60UTs}. It can be seen that beyond a certain number of UTs, the downlink power corresponding to the ZF beamforming becomes unbounded. In fact, this happens at $K = 30$ UTs per cell (note that $N_t = 2K$ at this point indicating that the BS has used up all its degrees of freedom).
At the same time, the downlink power of the ROBF algorithm grows unbounded at a much later stage, i.e. around $55$  UTs per cell. 
Moreover, we can also see that as $D$ (distance between the BSs) increases from $500$ m to $1000$ m, the transmit power decreases as expected. This is due to the reduced effect of inter-cell interference. 
Lastly, to quantify the transmit power gains obtained by the ROBF algorithm over ZF, we zoom into the previous plot in Figure \ref{fig:WithArrowROBFVsZFBF_Pathloss_25UTs}. Once again, as noted before, ROBF algorithm provides substantial power gains compared to ZF, a gain of 
$5$ dBm per UT for $K = 25$ UTs per cell.

\subsection*{ROBF with Individual BS Transmit Power Constraints}
In this subsection, we illustrate the performance of the ROBF algorithm with the individual BS transmit power constraints developed in Section \ref{sec:PeakPower}. In particular, we show the convergence as well as the performance results of this algorithm.
For illustration purposes, we consider a $2$ cell scenario and an identical system set up as in the previous case.
The peak power of BS$_1$ is set to $20$ dB and the peak power of BS$_2$ is set to $10$ dB.
We run the algorithm of Section \ref{sec:PeakPower} until the stopping criteria mentioned in Step 5 is satisfied. 
The transmit power per BS is plotted against the number of iterations in Figure \ref{fig:ROBF_With_MaxPower}.
The horizontal dashed lines in this figure correspond to the transmit powers obtained by running the original ROBF algorithm (without the individual BS transmit power constraints).
It can be seen that for this case, the peak transmit power of BS$_2$ is greater than $10$ dB, and hence its peak power of constraint is violated.
The solid lines represent the transmit powers obtained by running the ROBF algorithm with the individual BS transmit power constraints. 
It can be seen that the transmit powers of both the BSs converge within a few iterations. Further, the transmit power of BS$_1$ is equal to $10$ dB, and hence its
peak transmit power constraint is respected. The transmit power of BS$_2$ is greater for this case as compared to the original ROBF algorithm.
Further, we observed from the numerical results that the total downlink power of the system (summing the transmit power of both the BSs) obtained by running the
original ROBF algorithm is $14 $ dB where as its value for the case of ROBF algorithm with the individual BS transmit power constraints is $14.65$ dB. 
This illustrates the fact that the minimum downlink power with the individual BS transmit power constraints is greater than the minimum downlink power of the original optimization 
in \eqref{eqn:OPT_basic}.
\begin{figure}[!t]
\centering
\includegraphics[width=3.6 in, height  = 2.7 in]{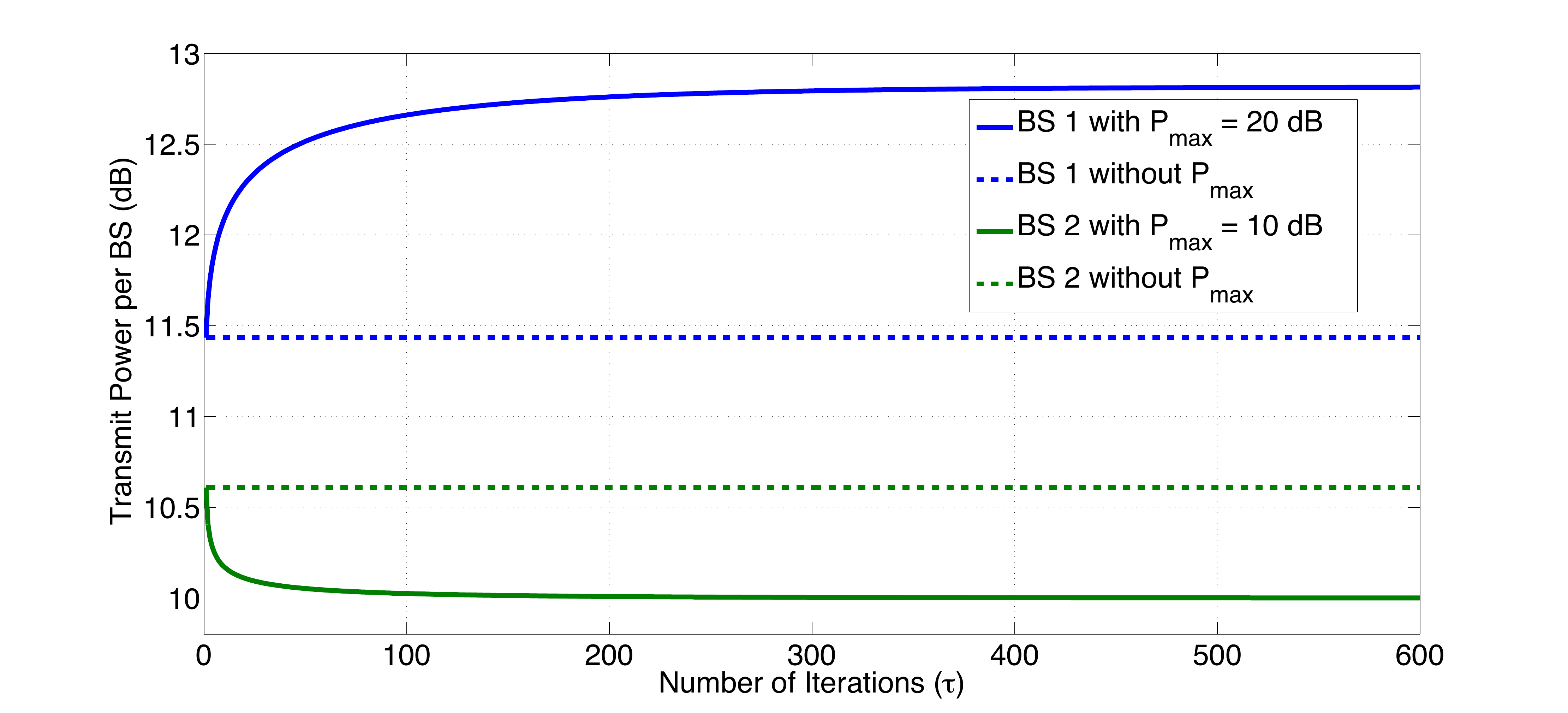}  
\caption{Iterations of the downlink power per BS. $N_t  = 100$ antennas/ BS, $K = 50$ UTs per cell, $P_{1,\max} = 20$ dB and $P_{2,\max} = 10.$ Target rate per user = $3$ bits/s/Hz.}
\label{fig:ROBF_With_MaxPower}
\end{figure}

\subsection*{Impact of Pilot Contamination Effect}
%\cite{PilotResureBjo2014}
We finally investigate the impact of pilot contamination effect on the performance of the ROBF algorithm, and
the improvements obtained by the MROBF algorithm. We consider an identical system set up as in the case of 
perfect CSI.
Since our main objective is to characterize the performance
loss due to the pilot contamination effect,  we ignore the errors associated with the CSI estimation process, i.e., $P_{\text{Tr}}$ in \eqref{eqn:MMSEEstimate}  is set to a very high value.

First, we compare the achieved SINR in the downlink by the MROBF algorithm and the CBF algorithm in Figure \ref{fig:MROBFVsCBF}. 
For the implementation of CBF, we consider that Algorithm 1 is implemented with all the CSI values replaced 
by their estimates (i.e. $\hv_{i,n,k}$ 
replaced by $\hat{\hv}_{i,n,k} \ \forall i,n,k$ computed as in \eqref{eqn:MMSEEstimate}).
It can be seen that the MROBF algorithm very closely matches the target SINR values, where as the CBF algorithm
does not achieve the target SINR. This is due to the fact that in the MROBF algorithm, parameters can be computed
based on the slow fading co-efficient, which can be estimated accurately at the BS. However, the CBF algorithm requires knowledge
of the fast fading co-efficient, whose estimation suffers from the pilot contamination effect (in addition to the CSI estimation error
that has been ignored in our numerical results).

Finally, we provide a comparison of the downlink power reduction provided by the MROBF algorithm.
Notice that the CBF and ROBF do not achieve the target SINR. For the sake of comparison, we devise an algorithm that achieves the target SINR. One can think 
of retaining the uplink power allocation of the ROBF algorithm, and only adapting the computation of $\bar{\deltav}$ as in \eqref{eqn:DLAdaptMROBF},
such that the target SINR is achieved in the downlink. We address this algorithm by the name ROBF with downlink adaptation.
We plot the transmit power per UT as a function of the target SINR for both the MROBF algorithm and the ROBF with downlink adaptation in Figure \ref{fig:TxmitPower}.
It can be seen that the transmit power per UT is significantly lower for the MROBF algorithm. Moreover, beyond a certain target SINR value,
the downlink power for the ROBF with downlink adaptation becomes very high, implying that the target SINR cannot be supported by this algorithm. 
In contrast, the MROBF algorithm can support a greater range of target SINR values. 
This is due to the fact that the ROBF with downlink adaptation algorithm
significantly under estimates the interference arising out of the UTs which reuse the same pilot symbols, where as the MROBF algorithm accurately accounts
for this quantity.
 \begin{figure*}[!t]
\begin{minipage}[t]{3.4 in}
\centering
\includegraphics[width=3.6 in, height  = 2.7 in]{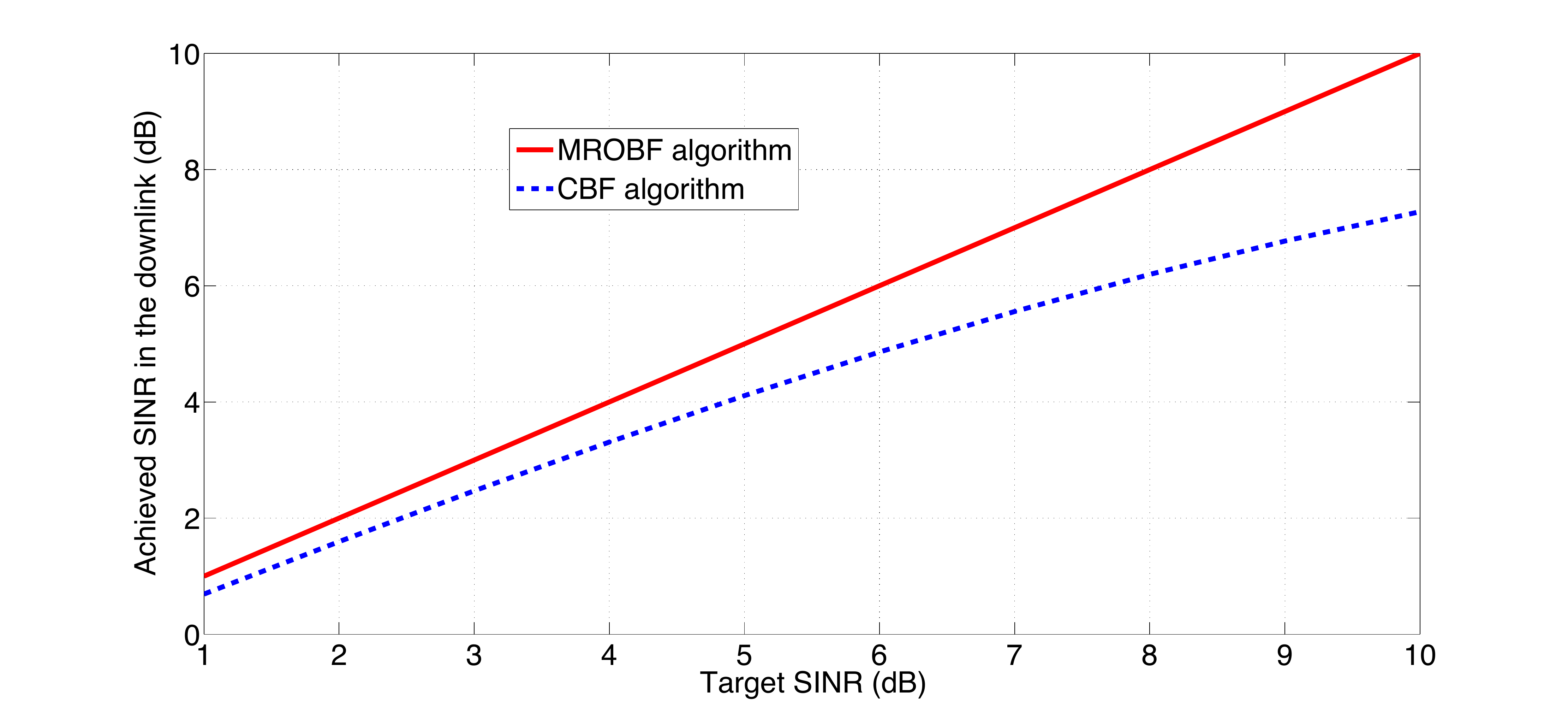}  
\caption{Achieved SINR Vs Target SINR for MROBF and CBF algorithm in the presence of pilot contamination effect for $N_t = 60$ antennas/BS and $K = 30$ UTs per cell.}
\label{fig:MROBFVsCBF}
\end{minipage}%\hfill
\begin{minipage}[t]{3.4 in}
\centering
\includegraphics[width=3.6 in, height  = 2.7 in]{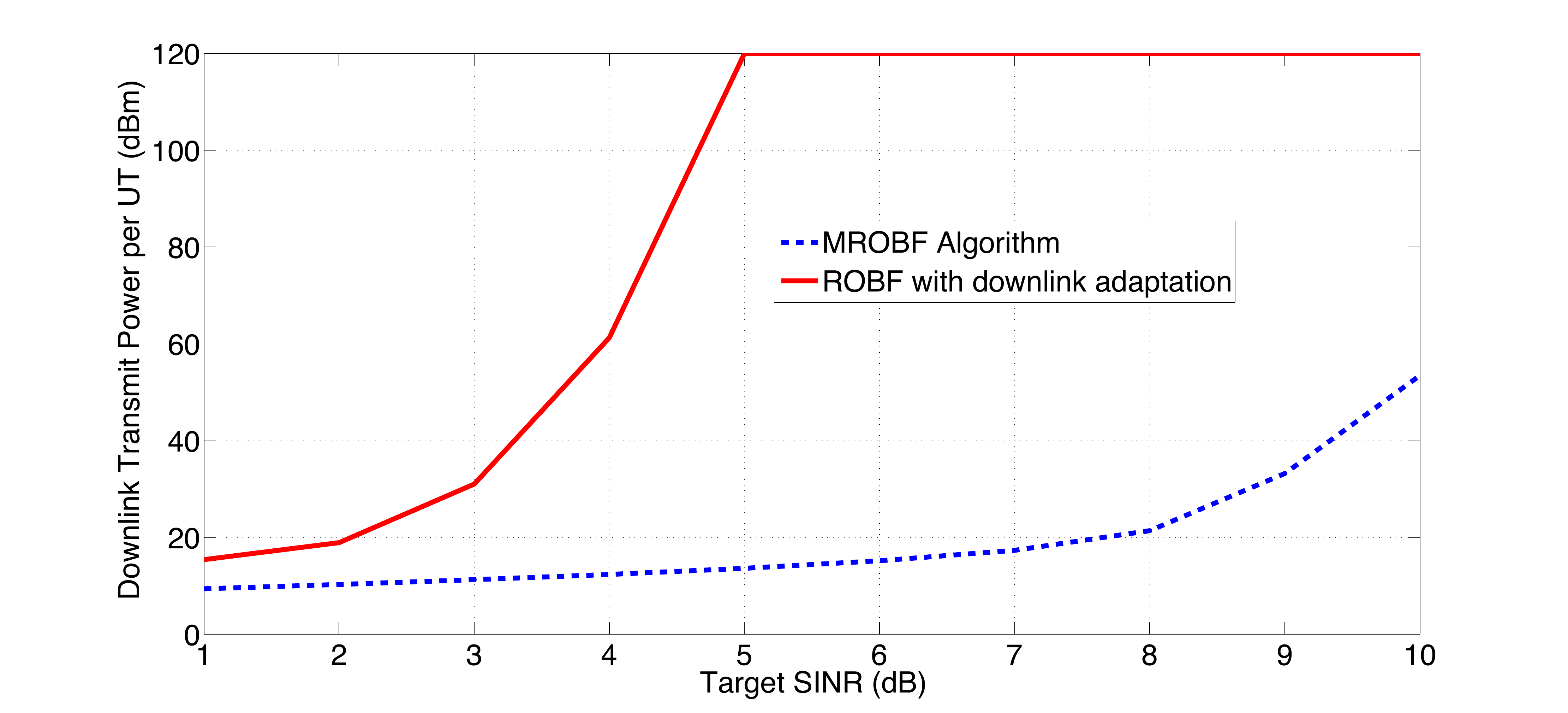}  
\caption{Comparison of downlink power per UT Vs Target SINR for the MROBF and ROBF algorithm with downlink adaptation.}
\label{fig:TxmitPower}
\end{minipage}
\end{figure*}

\section{Conclusion}
\label{sec:Conc}
In this work, we formulated a decentralized multi-cell beamforming algorithm for massive MIMO systems with the objective of minimizing the aggregate transmit power
across all the BSs subject to UT SINR constrains.
The algorithm requires only locally available CSI at the BSs, and some statistical side information of the channel gains to other UTs, and incurs a lower burden both in terms of the information exchange between BSs and the number of computations. 
Further we proved theoretically that this algorithm is asymptotically optimal for a large number of BS antennas and UTs. 
We also characterized a lower bound on the range of SINR values for which this algorithm is feasible.
We confirmed using numerical results that
the algorithm is nearly optimal in the regime of massive MIMO
systems. It also provides substantial power savings as compared to zero-forcing beamforming when the number of antennas per BS is of the same orders of magnitude as the number of UTs per cell. 
We proposed a heuristic extension of this algorithm to incorporate a practical constraint in multi-cell systems, i.e. the individual BS transmit power
constraints. We verified the convergence as well as the performance of this algorithm numerically.
Finally, we investigated the impact of CSI estimation errors and pilot contamination on the performance of the ROBF algorithm. Further, we proposed a heuristic modification of the ROBF algorithm, and showed with the help of numerical results that 
such an algorithm is more robust to the effects of CSI estimation errors and pilot contamination (as compared to an algorithm that uses
fast fading CSI values to compute the system parameters).

\section*{Appendix A: Some Relevant Results}
\label{sec:RandMtx}
\begin{lemma} (\cite{Rombaldi2000})
\label{lem:Rombaldi}
For two matrices $\Xm$ and $\Ym,$ if $0 \leq \Xm \leq \Ym,$
then $\rho(\Xm) \leq \rho(\Ym).$
\end{lemma}

\begin{lemma} \label{lem:lemma1} (Equation 2.2, \cite{jackbai95})
Let $\Am$ be a Hermitian invertible matrix of size $N \times N$, then for any vector $\xv \in \mathbb{C}^N$ and scalar $\tau \in \mathbb{C}$ for which $\Am+\tau \xv \xv^H$ is invertible,
\begin{align*}
\LB \Am + \tau \xv \xv^H \RB^{-1} = \Am^{-1} - \frac{\Am^{-1} \tau \xv \xv^H \Am^{-1}}{1+ \tau \xv^H \Am^{-1}  \xv},
\end{align*}
and
\begin{align*}
\xv^H(\Am+\tau \xv\xv^H)^{-1} = \frac{\xv^H\Am^{-1}}{1+\tau \xv^H \Am^{-1} \xv}.
\end{align*}
\end{lemma}

\begin{lemma} \label{lem:lemma2} (Lemma $2.6,$ \cite{jackbai98})
Let $\xv,\yv \sim \mathcal{CN}(0,\frac{1}{N} \Id_{N}) \in \mathbb{C}^{N},$ $\Am \in \CC^{N \times N}$  Hermitian matrix, and that $\xv$ and $\yv$ are mutually independent, and independent of $\Am.$ Consider $m \geq 2.$ Then there exists a constant $C_m$ independent 
of $N$ and $\Am$ such that
\begin{align*}
\mathbb{E} \LSB \big{|} \xv^H\Am\xv - \frac{1}{N}\trace(\Am) \big{|}^m\RSB \leq \frac{C_m} {N^{m/2}}||\Am||^m.
\end{align*}
This implies by the Markov inequality and the Borel Cantelli lemma for  $||\Am|| < \infty,$ that
$$ \xv^H\Am\xv - \frac{1}{N}\trace(\Am) \xrightarrow[N \to \infty]{\text{a.s.}} 0.$$
Additionally, $$ \xv^H\Am\yv  \xrightarrow[N \to \infty]{\text{a.s.}} 0.$$
\end{lemma}

\begin{lemma} \label{lem:lemma3} (Lemma $2.6,$ \cite{jackbai95}) %Lemma2.6
Let $z \in \mathbb{C}^+$ with $v = \Im(z)$ and $\Am$ and $\Bm$ are $N \times N$ matrices with $\Bm$ being Hermitian, $\tau \in \mathbb{R},$ and
$\qv \in \mathbb{C}^N,$ then
\begin{align*} 
|\trace\big((\Bm-z\Id)^{-1}-(\Bm+\tau \qv \qv^H -z\Id)^{-1})\Am\big)| \leq \frac{||\Am||}{v}.
\end{align*}
\end{lemma}

\begin{theorem} (\cite{jackbai95})
\label{thm:deteqm1}
Consider the matrix $\Bm_{N_t} = \Xm_{N_t}\Tm_{N_t}\Xm^H_{N_t},$ where $\Xm_{N_t} = \frac{1}{\sqrt{N_t}} \Ym^{N_t} \in \CC^{N_t \times NK}$ with entries $\Ym^{N_t}(p,q) \sim \mathcal{CN}(0,1)$, and the matrix $\Tm_{N_t}$ a non random diagonal matrix given by
$\Tm_{N_t} = \text{diag}(t_1,\dots,t_{NK}) \in \RR^{NK \times NK}.$ 
Let 
$m(z) = \frac{1}{{N_t}}\trace(\Bm_{N_t}+z\Id)^{-1}, \ z > 0,$
Then, $m(z) - \bar{m}(z) \xrightarrow[N_t,K \to \infty]{\text{a.s.}} 0,$
where $\bar{m}(z)$ can be evaluated as the unique solution to the fixed point equation
\begin{align*}
\bar{m}(z) = \LB \frac{1}{N_t}\sum^{NK}_{i=1} \frac{t_i}{1+t_i \bar{m}(z)}+z\RB ^{-1}.
\end{align*}
\end{theorem}
Further from \cite{sebastian2010}, under the same assumptions as Theorem \ref{thm:deteqm1}, let
$m^{\prime} (z) = \frac{1}{{N_t}} \trace(\Bm_{N_t}+z\Id)^{-2}.$
Then,
\begin{align}
m^{\prime}(z) - \bar{m}^{\prime}(z) \xrightarrow[N_t,K \to \infty]{\text{a.s.}} 0 \label{eqn:DerConv}
\end{align}
where $\bar{m}^{\prime}(z)$ can be evaluated as 
\begin{align}
\bar{m}^{\prime}(z) = \frac{\bar{m}^2(z)}{ 1- \frac{\bar{m}^2(z)}{N_t}\sum^{NK}_{i=1} \frac{t^2_i}{\LB 1+t_i \bar{m}(z)\RB^2}}
\label{eqn:stilderivative}. 
\end{align}
%\label{eqn:stilder}
At this point, we remark that all the above results also
hold for a more general non-Gaussian vectors/ matrices satisfying some moment conditions.

\begin{lemma} (Lemma $1$,\cite{peacock2008})
\label{lem:peacock}
Let $a_N$, $b_N$, $x_N$ and $y_N$ denote four
infinite sequences of complex random variables indexed by
$N.$
If $a_N \asymp b_N$ and $x_N \asymp y_N$ and if $|a_N|,|y|^{-1}_N$ and/or
$|a_N|,|y|^{-1}_N$ are uniformly bounded above over $N,$ then 
$a_N/x_N \asymp b_N/y_N.$ The requirement of uniform bound can
be relaxed to boundedness almost surely.
\end{lemma}

\section*{Appendix B: Proof of the Fixed Point Equation}
\label{sec:ProofProp1}
In order to prove the convergence of the iterative equation $\eqref{eqn:lam_stil}$,
we use the arguments of standard function \cite{Yates1995}.
A $K-$variate function $\gv(\xv) = \LSB g_1(\xv),\dots,g_K(\xv) \RSB \in \RR^K$ for $\xv \in \CC^K$ is said to be standard if it fulfills the following conditions:
%If the function $g_k(x_1,\dots,x_K)$ satisfies the following properties,
\begin{itemize}
\item Positivity: $\gv(\xv) > 0,$ for $\xv \geq 0.$
\item Scalability: For $\beta > 1,$ 
$\beta \gv(\xv) >  \gv(\beta \xv).$
\item Monotonicity: For $\xv^{\prime} \geq \xv,$
$\gv(\xv^{\prime}) \geq \gv(\xv).$
\end{itemize}
Now consider the following iterative algorithm given by,
\begin{align}
\xv^{t+1} = \gv(\xv^{t}), \qquad t \geq 1 \label{eqn:jaja680} 
\end{align}
If the $K-$variate function is standard, then it ensures the 
convergence of \eqref{eqn:jaja680} to its unique fixed point solution $\xv = \gv(\xv)$, if the solution exists.

Consider the iterative equations in $\eqref{eqn:lam_stil}$ which can be represented as 
\begin{align}
\mu_{i,j}^{t+1}  = f_{i,j} (\muv^t) \label{eqn:jaja865}
\end{align}
where $\muv^t_i = [\mu^t_{i,1},\dots,\mu^t_{i,K}]^T$ and 
 $\muv^t = [\muv^t_{i},\dots,\muv^t_n]^T$
and
\begin{align}
f_{i,j} (\muv^t) \defines \frac{\gamma_{i,j}}{\sigma_{i,i,j}\bar{m}^t_{i}} \qquad \forall i,j . 
\end{align}
Let us also define the $NK-$ variate function
\begin{align}
\fv_i (\muv^t) & \defines [f_{i,1}(\muv^t),\dots,f_{i,K}(\muv^t)]^T \in \RR^{K} \nonumber \\
\fv(\muv^t) & \defines [\fv_1 (\muv^t),\dots,\fv_N(\muv^t)]^T \in \RR^{NK}.
\end{align}
The existence of the fixed point to the equation $\muv = \fv(\muv)$
follows from Lemma 1.
We now prove that the $NK-$ variate function $\fv(\muv)$
is a standard function and hence the iterations of  \eqref{eqn:jaja865} 
converges to its unique fixed point.

In the subsequent part of this proof, 
we introduce the notation $\bar{m}^t_i(z)$ to denote 
the solution of the fixed point equation 
\begin{align}
& \bar{m}^t_i(z)  = \LB \frac{1}{N_t}\sum^N_{n = 1} \sum^K_{k = 1} \frac{\sigma_{i,n,k} \mu^t_{n,k}}{1+\sigma_{i,n,k}\mu^t_{n,k} \bar{m}^t_i(z)}+z \RB^{-1}. \label{eqn:FPeqn1}
\end{align}
With this notation, the solution to the fixed point equation of \eqref{eqn:FPeqn} 
can be written as $\bar{m}^t_i(1).$
For notational convenience, we also drop the 
superscript $t.$  \\
{\bf Positivity:}
The positivity result follows directly since $\bar{m}_i(1)$ is positive whenever $\muv \geq 0.$ 
Hence $f_{i,j}(\muv) > 0.$ \\
{\bf Scalability: }
Let us consider the difference between the following quantities.
\begin{align}
 \beta f_{i,j}(\muv) -  f_{i,j}(\beta \muv)
& = \frac{\beta \gamma_{i,j}}{\sigma_{i,i,j}\bar{m}^{(1)}_i(1)} -  \frac{ \gamma_{i,j}}{\sigma_{i,i,j}\bar{m}^{(2)}_i(1)} \nonumber \\
& = \frac{ \gamma_{i,j}}{\sigma_{i,i,j}}\LB \frac{\beta}{\bar{m}^{(1)}_i(1)} - \frac{1}{\bar{m}^{(2)}_i(1)}\RB \label{eqn:scale_diff}
\end{align}
where $\beta>1$ and $\bar{m}^{(1)}_i(1)$ and $\bar{m}^{(2)}_i(1)$ are the unique solutions to $\eqref{eqn:FPeqn}$ evaluated at $\muv$ and 
  $\beta\muv.$
In order to evaluate $\bar{m}^{(2)}_i(1)$, we go back to the definition $\bar{m}^{(2)}_i(1)$ evaluated at
$\beta\muv.$
\begin{align}
\bar{m}^{(2)}_i(1) = \frac{1}{N} \trace(\Hm_i \beta \Mm \Hm^H_i+\Id)^{-1} \asymp \frac{1}{\beta} \bar{m}_i\LB\frac{1}{\beta}\RB \label{eqn:m2z}
\end{align}
where $\Mm = \diag [\muv].$ Clearly $1/\beta < 1.$ 
$\bar{m}^{(2)}_i(1)$ can be evaluated 
as the solution  to the fixed point equation $\eqref{eqn:FPeqn}$ evaluated
at the point $1/\beta$ and then scaling the result by $\beta.$
From $\eqref{eqn:scale_diff}$ and $\eqref{eqn:m2z}$, it can be concluded that in order to prove the scalability result, it is sufficient to show that $\bar{m}_i(z)$ is a decreasing function 
of $z.$ 

In order to prove the same, let us consider an extended version of the channel matrix which is
constructed as follows. Defining,
$\Rm_{i,j,k} = \sigma_{i,j,k} \Id_{N_t} \in \RR^{N_t \times N_t}$ and
$\Rm^L_{i,j,k} = \sigma_{i,j,k} \Id_{N_t L} \in \RR^{N_t L \times N_t L}.$ 
The matrix $\Hm^L_{i}$ is constructed as follows:
\begin{align}
\Hm^L_{i,j} & = \frac{1}{\sqrt{L}} \LSB \Rm^{1/2}_{i,j,1}\Xm^L_{i,j},\dots,\Rm^{1/2}_{i,j,K}\Xm^L_{i,j} \RSB \in \CC^{N_t L \times KL} \nonumber \\
\Hm^L_{i} & = \LSB \Hm^L_{i,1},\dots,\Hm^L_{i,N} \RSB \in \CC^{N_t L \times NKL} \label{eqn:jaja776}
\end{align}
where the matrix $\Xm^L_{i,j} \in \CC^{Nt \times K},$ whose elements
$\Xm^L_{i,j}(p,q) \sim \mathcal{CN} (0,\frac{1}{N_t}).$ 
Also, let us define the following,
\begin{align}
& \lambdav^L_{i,j} = \Big{[} \lambda_{i,j},\dots,\lambda_{i,j}\Big{]}^T \in \RR^{L \times 1} \nonumber \\
& \lambdav^L_{i} = \LSB \lambdav^L_{i,1},\dots,\lambdav^L_{i,K}\RSB^T \in \RR^{KL \times 1} \nonumber \\
& \Lambdam^L = \diag \LB \lambdav^L_{1},\dots,\lambdav^L_{N} \RB \in \RR^{NKL \times NKL} 
\end{align}
and let $\Qm_i(z) = (\Hm_i \Lambdam \Hm^H_i+z \Id_{N_t L})^{-1}$ and $\Qm^L_i(z) = (\Hm^L_i \Lambdam^L (\Hm^L)^H_i+z \Id_{N_t L})^{-1}.$
Let us denote
\begin{align}
m^L_i(z) = \frac{1}{N_t L} \trace (\Qm^L_i(z)).
\end{align}
It can be verified that for any fixed $N,N_t$ and $K,$ the following limit holds,
\begin{align}
m^L_i(z) - \bar{m}_i(z) \xrightarrow[L \to \infty]{\text{a.s.}} 0. \label{eqn:jaja675}
\end{align}
%where $\bar{m}_i(z)$ is the deterministic approximation of 
%the random quantity $\frac{1}{N_t} \trace (\Qm_i(z))$ (note Theorem 3) which can be evaluated as in \eqref{eqn:lam_stil}.
Now consider the difference between the following two quantities, for any $z_2 > z_1 > 0$ and for any positive $L,$
we have
\begin{align}
& m^L_i(z_1) -m^L_i(z_2)  \nonumber
\\ & = \frac{1}{N_t L} \trace \LB \Qm^L_i(z_1) - \Qm^L_i(z_2)\RB \nonumber
\\ & \stackrel{(a)}{>}  0 \label{eqn:jaja670}
\end{align}
where inequality $(a)$ follows by the following identity:
For invertible matrices $\Am$ and $\Bm,$ 
\begin{align}
\Am^{-1} - \Bm^{-1} = -\Am^{-1} (\Am- \Bm) \Bm^{-1}.
\end{align}
We now show that for any fixed $N,N_t$ and $K,$ the inequality in \eqref{eqn:jaja670} of the random quantities ($m^L_i(z)$) also hold for their respective deterministic approximations ($\bar{m}_i(z)$). This can be argued as follows.
Consider the difference
\begin{align}
m^L_i(z_1) - m^L_i&(z_2)  = m^L_i(z_1) - \bar{m}_i(z_1) - m^L_i(z_2) \nonumber \\ & + \bar{m}_i(z_2) + \bar{m}_i(z_1)- \bar{m}_i(z_2). \label{eqn:jaja672}
\end{align}
First note that since the matrix $\Hm^L_{i}$ has bounded spectral norm almost surely,
it follows that, almost surely,
\begin{align}
\lim_{L \to \infty} m^L_i(z) > 0. \label{eqn:jaja755}
\end{align}
From \eqref{eqn:jaja755}, further applying the result of \eqref{eqn:jaja675} in the right hand side of \eqref{eqn:jaja672} 
and the inequality of \eqref{eqn:jaja670}, it follows that 
\begin{align}
\bar{m}_i(z_1)- \bar{m}_i(z_2) > 0. \label{eqn:jaja756}
\end{align}
Thus, $\bar{m}_i(z)$ is a decreasing function of $z$. 
We remark that the strict positivity of the $m^L_i(z)$ in \eqref{eqn:jaja755} is essential for the strict
positivity argument of \eqref{eqn:jaja756}.\\
{\bf Monotonicity: }
Consider $\muv^{\prime} \geq \muv.$ 
In this case, we denote $\bar{m}^{(1)}_i(1)$ and $\bar{m}^{(2)}_i(1)$ as the solutions to the fixed point equations in $\eqref{eqn:FPeqn}$ evaluated at $\muv^{\prime}$ and $\muv$ respectively.
As before, let us consider the 
difference between the quantities, 
\begin{align}
f_{i,j}(\muv^{\prime}) - f_{i,j}(\muv) =  \frac{ \gamma_{i,j}}{\sigma_{i,i,j}\bar{m}^{(1)}_i(1)} -  \frac{ \gamma_{i,j}}{\sigma_{i,i,j}\bar{m}^{(2)}_i(1)}.
\end{align}
We now have to show that $\bar{m}^{(1)}_i(1) \leq \bar{m}^{(2)}_i(1)$ in order to prove the monotonicity result.
This can be shown by constructing the extended matrices as in \eqref{eqn:jaja776} and noting the monotonicity property for the associated random quantities and
extending the result to their deterministic approximations for any system dimensions. The proof is similar to the scalability result and hence omitted here.

\section*{Appendix C: Feasibility Conditions}
\subsection{Part I: Proof of Lemma 2}
Let us first consider the expression for the achieved uplink SINR
by the ROBF algorithm given by
\begin{align}
\Lambda_{i,j} (\muv) = \frac{\frac{\mu_{i,j}}{N_t}|\hat{\vv}^H_{i,j}\hv_{i,i,j}|^2}{\sum_{(n,k) \neq (i,j)} \frac{\mu_{n,k}}{N_t}|\hat{\vv}^H_{i,j}\hv_{i,n,k}|^2+ ||{\hat{\vv}}_{i,j}||_2^2}. 
\end{align}
%Since the beamforming vector corresponding 
%to the ROBF algorithm is a MMSE beamformer \eqref{eqn:ULBF_asymp}, it can be shown easily that the 
%uplink SINR can be equivalently represented as
%\begin{align}
%\Lambda_{i,j} (\mu_{i,j}) = \frac{\mu_{i,j}}{N_t} \hv^H_{i,i,j} (\Sigmam^{\prime}_i+ \Id_{N_t})^{-1} \hv_{i,i,j}. \label{eqn:UL_form2}
%\end{align}
%As we will show later, 
%Now recall the fixed point equation 
We examine the expression for the uplink SINR 
in the large system regime.
It can be shown that the expression
for the uplink SINR converges asymptotically to the following (the proof of this is very similar to the convergence result of the downlink interference and downlink SINR provided in Lemma 4 and Appendix E, and hence omitted in order to avoid repetition): %\label{eqn:ULasymp_form1}
\begin{align}
\Lambda_{i,j} (\muv) & \asymp \frac{\mu_{i,j} \sigma_{i,i,j} \bar{G}_{i,i,j} \bar{m}^2_i}{\sum_{(n,k) \neq (i,j)} \frac{1}{N_t} \mu_{n,k} \bar{G}_{i,n,k} \bar{G}_{i,i,j} \bar{m}^\prime_i + \alpha_i \bar{G}_{i,i,j} \bar{m}^\prime_i} \nonumber \\
& = \frac{\mu_{i,j} \sigma_{i,i,j}  \bar{m}^2_i}{\sum_{(n,k) \neq (i,j)} \frac{1}{N_t} \mu_{n,k} \bar{G}_{i,n,k}  \bar{m}^\prime_i + \alpha_i  \bar{m}^\prime_i}. \label{eqn:UL_asymp_form1simp}
\end{align}
Let us we consider a slightly modified version 
of the expression \eqref{eqn:UL_asymp_form1simp} 
as follows:
\begin{align}
\frac{\mu_{i,j} \sigma_{i,i,j}  \bar{m}^2_i}{\sum_{n,k} \frac{1}{N_t} \mu_{n,k} \bar{G}_{i,n,k}  \bar{m}^\prime_i + \alpha_i  \bar{m}^\prime_i}. \label{eqn:UL_asymp_form2simp}
\end{align}

We now examine the expression \eqref{eqn:UL_asymp_form2simp}
in detail.
Recall the fixed point equation for the computation of $\bar{m}_i$ in \eqref{eqn:FPeqn}. Upon
rearranging the terms, we have
\begin{align}
 \bar{m}_i  = 1  - \frac{1}{N_t} \sum_{n,k} \frac{\sigma_{i,n,k} \mu_{n,k} \bar{m}_i}{1+\sigma_{i,n,k}\mu_{n,k} \bar{m}_i} .
\end{align}
Substituting the expressions for $\bar{G}_{i,n,k},$ $\bar{m}_i$ and $\bar{m}^\prime_i$ in \eqref{eqn:UL_asymp_form2simp}, we get
\begin{align}
 \eqref{eqn:UL_asymp_form2simp} = \frac{\sigma_{i,i,j} \mu_{i,j} \LB 1-\frac{1}{N_t} \sum_{n,k} \frac{(\sigma_{i,n,k} \mu_{n,k}\bar{m}_i)^2}{(1+\sigma_{i,n,k} \mu_{n,k}\bar{m}_i)^2 } \RB}{1+\frac{1}{N_t} \sum_{n,k} \frac{\sigma_{i,n,k} \mu_{n,k}}{(1+\sigma_{i,n,k} \mu_{n,k}\bar{m}_i)^2}}. \label{eqn:jaja403}
\end{align}
Multiplying and diving by $\bar{m}_i$ in \eqref{eqn:jaja403},
we obtain
\begin{align}
&  \eqref{eqn:UL_asymp_form2simp} = \frac{\sigma_{i,i,j} \mu_{i,j} \bar{m}_i \LB 1-\frac{1}{N_t} \sum_{n,k} \frac{(\sigma_{i,n,k} \mu_{n,k}\bar{m}_i)^2}{(1+\sigma_{i,n,k} \mu_{n,k}\bar{m}_i)^2 } \RB}{\bar{m}_i+\frac{1}{N_t} \sum_{n,k} \frac{\sigma_{i,n,k} \mu_{n,k}\bar{m}_i}{(1+\sigma_{i,n,k} \mu_{n,k}\bar{m}_i)^2}} \nonumber \\
& = \frac{\sigma_{i,i,j} \mu_{i,j} \bar{m}_i \LB 1-\frac{1}{N_t} \sum_{n,k} \frac{(\sigma_{i,n,k} \mu_{n,k}\bar{m}_i)^2}{(1+\sigma_{i,n,k} \mu_{n,k}\bar{m}_i)^2 } \RB}{1  - \frac{1}{N_t} \sum_{n,k} \frac{\sigma_{i,n,k} \mu_{n,k} \bar{m}_i}{1+\sigma_{i,n,k}\mu_{n,k} \bar{m}_i} +\frac{1}{N_t} \sum_{n,k} \frac{\sigma_{i,n,k} \mu_{n,k}\bar{m}_i}{(1+\sigma_{i,n,k} \mu_{n,k}\bar{m}_i)^2}} \nonumber \\
& = \frac{\sigma_{i,i,j} \mu_{i,j} \bar{m}_i \LB 1-\frac{1}{N_t} \sum_{n,k} \frac{(\sigma_{i,n,k} \mu_{n,k}\bar{m}_i)^2}{(1+\sigma_{i,n,k} \mu_{n,k}\bar{m}_i)^2 } \RB}{1  - \frac{1}{N_t} \sum_{n,k} \frac{(\sigma_{i,n,k} \mu_{n,k} m_i)^2}{(1+\sigma_{i,n,k}\mu_{n,k} \bar{m}_i)^2}} \nonumber \\
& = \sigma_{i,i,j} \mu_{i,j} \bar{m}_i. \label{eqn:jaja405}
\end{align}
We notice that that the expression at \eqref{eqn:jaja405}
is exactly in the same form as the right hand side of the
fixed point equation \eqref{eqn:FPconv}.
Therefore, we can rewrite \eqref{eqn:FPconv} as follows:
\begin{align*}
\gamma_{i,j} & = \frac{\mu_{i,j} \sigma_{i,i,j} \bar{G}_{i,i,j} \bar{m}^2_i}{\sum_{n,k} \frac{1}{N_t}\mu_{n,k} \bar{G}_{i,n,k} \bar{G}_{i,i,j} \bar{m}^\prime_i +  \bar{G}_{i,i,j} \bar{m}^\prime_i} \qquad \forall i,j.
\end{align*}
Rearranging, we have
\begin{align*}
& \frac{1}{\gamma_{i,j}}\mu_{i,j} \sigma_{i,i,j} \bar{G}_{i,i,j} \bar{m}^2_i=\sum_{n,k} \frac{\mu_{n,k}}{N_t} \bar{G}_{i,n,k} \bar{G}_{i,i,j} \bar{m}^\prime_i \\ & \qquad \qquad \qquad \qquad \qquad+ \bar{G}_{i,i,j} \bar{m}^\prime_i \qquad \forall i,j \\
& \implies \mu_{i,j}= \sum_{n,k}  \frac{\gamma_{i,j} \mu_{n,k}}{N_t}  \frac{ \bar{G}_{i,n,k} \bar{G}_{i,i,j} \bar{m}^\prime_i}{\sigma_{i,i,j}\bar{G}_{i,i,j} \bar{m}^2_i} \\ & \qquad \qquad \qquad \qquad \qquad+ \frac{  \bar{G}_{i,i,j} \bar{m}^\prime_i}{\sigma_{i,i,j}\bar{G}_{i,i,j} \bar{m}^2_i} \qquad \forall i,j.
\end{align*} 
The equation in matrix form can be written as 
\begin{align}
& \muv = \Gammam (\Deltam^\prime)^T \muv + {\bf \kappa} \label{eqn:UL_lineq}
\end{align}
where ${\bf \kappa}_i = \LSB \frac{\bar{m}^{\prime}_i}{\sigma_{i,i,1} \bar{m}^2_i},\dots, \frac{\bar{m}^{\prime}_i}{\sigma_{i,i,K} \bar{m}^2_i}\RSB^T$ and ${\bf \kappa} = \LSB {\bf \kappa}_1,\dots,{\bf \kappa}_N\RSB^T.$
The matrix $\Gammam$ is the same as the matrix $\Gammam$  defined in the description of the ROBF algorithm.
The matrix $\Deltam^\prime$ is defined as follows.
\begin{equation}
\Deltam^\prime = \left(
\begin{array}{ccc}
(\Deltam^\prime)^{1,1} & \ldots & (\Deltam^\prime)^{1,N} \\
\vdots  & \ddots & \vdots \\
(\Deltam^\prime)^{N,1} & \ldots & (\Deltam^\prime)^{N,N}\\
\end{array} \right)
\end{equation}
where each submatrix $(\Deltam^\prime)^{i,j} \in \CC^{K \times K}$ is given by 
\begin{equation} 
(\Deltam^\prime)^{i,n}_{j,k} \defines \begin{cases} \frac{1}{N_t}\bar{G}_{i,i,j} \bar{G}_{i,i,j} \bar{m}^{\prime}_{i} , & n=i ,\ k = j \\
                      \frac{1}{N_t}\bar{G}_{i,i,j} \bar{G}_{i,i,k} \bar{m}^{\prime}_{i} ,& n = i,\ k \neq j \\
                    \frac{1}{N_t}\bar{G}_{n,i,j} \bar{G}_{n,n,k}  \bar{m}^{\prime}_{n}  ,& n\neq i.
                   \end{cases} \label{eqn:coeff_mtx1}
\end{equation}
Notice that if the SINR targets are feasible for the uplink solution, the linear equations in \eqref{eqn:UL_lineq}
must have a solution. 
Problems with similar structure has been studied before in the context of power allocation in wireless networks \cite{Chaing2008}. 
It has been established that a feasible solution exists for such problems if and only if
$\rho(\Gammam (\Deltam^\prime)^T) < 1$ (where $\rho$ is the spectral radius of the matrix) and the matrix $\Id-\Gammam (\Deltam^\prime)^T$ is invertible \cite{Chaing2008}.
Now recall the matrix $\Deltam$ defined in equation \eqref{eqn:coeff_mtx}.
Observe that the matrix $\Deltam$ and $\Deltam^\prime$ only differ in the diagonal element.
Also, we can note that $\Deltam \leq \Deltam^\prime$ (denotes element wise inequality, refer to the notations), and hence
$\Gammam \Deltam \leq \Gammam \Deltam^\prime.$
Hence, it follows from Lemma \ref{lem:Rombaldi}, that $\rho(\Gammam (\Deltam)^T) \leq \rho(\Gammam (\Deltam^\prime)^T) < 1.$ 

Now recall the linear equations for computing $\bar{\deltav}$ given by
\begin{align*}
 \bar{\deltav} = \Gammam \Deltam \bar{\deltav} + \rhov. 
 \end{align*}
  
It has been established in works before \cite{ZanderFrodigh1994}, the eigen values 
of the matrices $\Gammam \Deltam$ and $\Gammam \Deltam^T$
are the same. For completeness, this can be argued as follows,
\begin{align*}
|\Gammam \Deltam - \lambda \Id| & = |\Gammam||\Deltam-\lambda\Gammam^{-1}| = |\Gammam||\LB \Deltam-\lambda\Gammam^{-1} \RB^T| \\
& = |\Gammam| | \Deltam^T-\lambda\Gammam^{-1} | = | \Gammam \Deltam^T-\lambda \Id | = 0.
\end{align*}
From the above discussion, we conclude that the spectral radius of $\Gammam \Deltam$ and $\Gammam \Deltam^T$ are the same
and, hence, if $\rho ( \Gamma \Deltam) < 1$ then $\rho (\Gammam \Deltam^T) < 1$ and vice versa.
Consequently, the matrix $\Id-\Gammam \Deltam$ is invertible.
This proves the result of Lemma \ref{lem:updownfeas}.

\subsection{Part II: Feasible SINR targets for the modified system}
Consider the equation \eqref{eqn:ref668}.
Rearranging \eqref{eqn:ref668}, we obtain
\begin{align}
1  & = \frac{1}{N_t} \sum^K_{k = 1}  \frac{\sigma_{i,i,k} \mu^{\text{mod}}_{i,k} \bar{m}^{\text{mod}}_i}{1+\sigma_{i,i,k}\mu^{\text{mod}}_{i,k} \bar{m}^{\text{mod}}_i} \nonumber \\ & \qquad \qquad +\frac{1}{N_t} \sum^N_{\stackrel{n = 1} { n \neq i}} \sum^K_{k = 1} \frac{\sigma_{max} (n) \mu^{\text{mod}}_{n,k}}{1+\sigma_{max} (n)\mu^{\text{mod}}_{n,k} \bar{m}^{\text{mod}}_i} +  \bar{m}^{\text{mod}}_i \nonumber \\
& = \frac{1}{N_t} \sum^K_{k = 1} \frac{\gamma_{i,k}}{1+\gamma_{i,k}} \nonumber \\ & \qquad \qquad +\frac{1}{N_t}  \sum^N_{\stackrel{n = 1} { n \neq i}} \sum^K_{k = 1} \frac{\sigma_{max} (n) \mu^{\text{mod}}_{n,k} \bar{m}^{\text{mod}}_i}{1+\sigma_{max} (n) \mu^{\text{mod}}_{n,k} \bar{m}^{\text{mod}}_i} +  \bar{m}^{\text{mod}}_i. \label{eqn:refhere401}
\end{align}
Equation \eqref{eqn:refhere401} is true for all $\bar{m}^{\text{mod}}_i, i = 1,\dots,N.$ Therefore, by symmetry of the fixed point equation $\bar{m}^{\text{mod}}_i = \bar{m}^{\text{mod}}, i = 1,\dots,N.$
Thus, 
\begin{align}
\sigma_{max} (n) \mu^{\text{mod}}_{n,k} \bar{m}^{\text{mod}}_i = \sigma_{max} (n) \mu^{\text{mod}}_{n,k} \bar{m}^{\text{mod}} = \frac{\sigma_{max} (n)}{\sigma_{n,n,k}} \gamma_{n,k} \label{eqn:ref667}. 
\end{align}
Substituting \eqref{eqn:ref667} in \eqref{eqn:refhere401} yields
\begin{align}
 1 = \frac{1}{N_t} \sum^K_{k = 1} \frac{\gamma_{i,k}}{1+\gamma_{i,k}}+\frac{1}{N_t} \sum^N_{\stackrel{n = 1} { n \neq i}} \sum^K_{k = 1} \frac{\frac{\sigma_{max} (n)}{\sigma_{n,n,k}} \gamma_{n,k}}{1+\frac{\sigma_{max} (n)}{\sigma_{n,n,k}} \gamma_{n,k}} + \bar{m}_i \label{eqn:jaja0031}. 
\end{align}
Now recall \eqref{eqn:jaja666}. Rearranging, we obtain
\begin{align}
\mu^{\text{mod}}_{i,j} = \frac{\gamma_{i,j}}{ \sigma_{i,i,j}\bar{m}^{\text{mod}}_i}. \label{eqn:jaja00301}
\end{align}
The feasibility condition {\bf [C1]} requires that
$\mu^{\text{mod}}_{i,j} < \infty, \ \forall i,j$ This implies from \eqref{eqn:jaja00301} that $\bar{m}^{\text{mod}}_i > 0.$
Using this in \eqref{eqn:jaja0031}, we obtain the
following condition on the target SINR for feasibility:
\begin{align}
 \frac{1}{N_t} \sum^K_{k = 1} \frac{\gamma_{i,k}}{1+\gamma_{i,k}}+\frac{1}{N_t} \sum^N_{\stackrel{n = 1} { n \neq i}} \sum^K_{k = 1} \frac{\frac{\sigma_{max} (n)}{\sigma_{n,n,k}} \gamma_{n,k}}{1+\frac{\sigma_{max} (n)}{\sigma_{n,n,k}} \gamma_{n,k}} < 1. 
\end{align}

\subsubsection{Proof of Corollary 1}
Let us rewrite \eqref{eqn:jaja0031} as
\begin{align}
\bar{m}_i =  1- \frac{1}{N_t} \sum^K_{k = 1} \frac{\gamma_{i,k}}{1+\gamma_{i,k}}+\frac{1}{N_t} \sum^N_{\stackrel{n = 1} { n \neq i}} \sum^K_{k = 1} \frac{\frac{\sigma_{max} (n)}{\sigma_{n,n,k}} \gamma_{n,k}}{1+\frac{\sigma_{max} (n)}{\sigma_{n,n,k}} \gamma_{n,k}}  \label{eqn:jaja00310}. 
\end{align}
The feasibility of the uplink problem in the asymptotic domain implies that $\limsup_{N_t,K \to \infty} \mu^{\text{mod}}_{i,j} < \infty.$
This in turn implies that $\liminf_{N_t,K \to \infty} \bar{m}^{\text{mod}}_i > 0$ (strictly positive).
Thus, we consider \eqref{eqn:jaja00310} in the asymptotic regime, and take $\liminf_{N_t,K \to \infty}.$ This yields,
\begin{align}
& \liminf_{N_t,K \to \infty}  \bar{m}_i  =  \liminf_{N_t,K \to \infty} \Big{[} 1- \frac{1}{N_t} \sum^K_{k = 1} \frac{\gamma_{i,k}}{1+\gamma_{i,k}} \nonumber \\ & \qquad \qquad
\qquad -\frac{1}{N_t} \sum^N_{\stackrel{n = 1} { n \neq i}} \sum^K_{k = 1} \frac{\frac{\sigma_{max} (n)}{\sigma_{n,n,k}} \gamma_{n,k}}{1+\frac{\sigma_{max} (n)}{\sigma_{n,n,k}} \gamma_{n,k}} \Big{]} . \\
& 
=  1- \limsup_{N_t,K \to \infty} \Big{[}  \frac{1}{N_t} \sum^K_{k = 1} \frac{\gamma_{i,k}}{1+\gamma_{i,k}} \nonumber \\
& \qquad \qquad
\qquad  -\frac{1}{N_t} \sum^N_{\stackrel{n = 1} { n \neq i}} \sum^K_{k = 1} \frac{\frac{\sigma_{max} (n)}{\sigma_{n,n,k}} \gamma_{n,k}}{1+\frac{\sigma_{max} (n)}{\sigma_{n,n,k}} \gamma_{n,k}} \Big{]} \label{eqn:jaja00311}
\end{align}
Since $\liminf_{N_t,K \to \infty} \bar{m}^{\text{mod}}_i > 0,$  
and \eqref{eqn:jaja00311}, it follows that the feasibility conditions are given by,
\begin{align}
 \limsup_{N_t,K \to \infty} \LSB \frac{1}{N_t} \sum^K_{k = 1} \frac{\gamma_{i,k}}{1+\gamma_{i,k}}+\frac{1}{N_t} \sum^N_{\stackrel{n = 1} { n \neq i}} \sum^K_{k = 1} \frac{\frac{\sigma_{max} (n)}{\sigma_{n,n,k}} \gamma_{n,k}}{1+\frac{\sigma_{max} (n)}{\sigma_{n,n,k}} \gamma_{n,k}} \RSB \nonumber \\ < 1 \qquad \forall i. \label{eqn:jaja406}
\end{align}

\subsection{Part III: Proof of Lemma 3}
Let us denote $\muv^{\text{mod}}_i = [\mu^{\text{mod}}_{i,1},\dots,\mu^{\text{mod}}_{i,K}]^T$
and $\muv^{\text{mod}} = [\muv^{\text{mod}}_1,\dots,\muv^{\text{mod}}_N]^T.$
The main idea behind this proof is the following:
Consider the pair of vector $\{ \gammav,\muv^{\text{mod}} \}$
that satisfies \eqref{eqn:jaja666}.
We show that the when the power allocation $\muv^{\text{mod}}$
is used in the original system, then 
\begin{align}
\sigma_{i,i,j}\mu^{\text{mod}}_{i,j} \bar{m}^{\text{mod1}}_i \geq \gamma_{i,j} \qquad \forall i,j \label{eqn:jaja1161}
\end{align}
where $\bar{m}^{\text{mod1}}_i$ satisfies
\begin{align}
 \frac{1}{\bar{m}^{\text{mod1}}_i} & =  \frac{1}{N_t}  \frac{\sigma_{i,i,k}  \mu^{\text{mod}}_{i,k}}{1+\sigma_{i,i,k}\mu^{\text{mod}}_{i,k} \bar{m}^{\text{mod1}}_i} \nonumber \\ & +\frac{1}{N_t} \sum^N_{\stackrel{n = 1} { n \neq i}} \sum^K_{k = 1} \frac{\sigma_{i,n,k} \mu^{\text{mod}}_{n,k}}{1+\sigma_{i,n,k} \mu^{\text{mod}}_{n,k} \bar{m}^{\text{mod1}}_i}+1 \label{eqn:jaja1171},
\end{align}
where \eqref{eqn:jaja1171} represents the fixed point equation 
calculated by utilizing the power allocation $\muv^{\text{mod}}$
in the original system.
If \eqref{eqn:jaja1161} is true, then this implies that 
$\gammav$ is a feasible SINR target vector for the original system 
(since there exists a power allocation $\muv^{\text{mod}}$ that can achieve this SINR target). 

We prove this as follows:
Firstly, it is easy to see that 
\begin{align}
\sigma_{max} (n) \mu^{\text{mod}}_{n,k} \geq \sigma_{i,n,k} \mu^{\text{mod}}_{n,k}  \qquad \forall n,k \label{eqn:jaja407}
\end{align}
since $\sigma_{max} (n) \geq \sigma_{i,n,k}$.
Recall that the function
$f_{i,j}  \defines \frac{\gamma_{i,j}}{\sigma_{i,i,j} \bar{m}_i}$
is a standard function. From \eqref{eqn:jaja407} and the monotonicity property of standard functions, we have
\begin{align}
\frac{\gamma_{i,j}}{\sigma_{i,i,j} \bar{m}^{\text{mod}}_i} \geq \frac{\gamma_{i,j}}{\sigma_{i,i,j} \bar{m}^{\text{mod1}}_i}. \label{eqn:jaja408}
\end{align}
Rearranging \eqref{eqn:jaja408} we have,
$\sigma_{i,i,j} \bar{m}^{\text{mod1}}_i \geq \sigma_{i,i,j} \bar{m}^{\text{mod}}_i.$
Multiplying by $\mu^{\text{mod}}_{i,j},$ it follows that
\begin{align}
\sigma_{i,i,j} \mu^{\text{mod}}_{i,j} \bar{m}^{\text{mod1}}_i \geq \sigma_{i,i,j} \mu^{\text{mod}}_{i,j} \bar{m}^{\text{mod}}_i. \label{eqn:jaja409}
\end{align}
From \eqref{eqn:jaja666}, 
 $\sigma_{i,i,j} \mu^{\text{mod}}_{i,j} \bar{m}^{\text{mod}}_i = \gamma_{i,j}.$ Therefore, we can conclude that
\begin{align*}
\sigma_{i,i,j} \mu^{\text{mod}}_{i,j} \bar{m}^{\text{mod1}}_i \geq \gamma_{i,j}.
\end{align*}

\section*{Appendix D: Convergence of the Uplink SINR}
In this Appendix, we prove the result of Lemma \ref{thm:ULDLSINRConv},
statement \eqref{eqn:ULSINRConv}.
Recall that the expression of uplink SINR for the ROBF algorithm is given by \eqref{eqn:UL_form2}.
We make a variable change and
denote $\xv_{i,i,j} = \frac{\hv_{i,i,j}}{\sqrt{N_t}}.$
This implies that 
\begin{align}
\Lambda_{i,j}(\muv) =  \mu_{i,j}\xv^H_{i,i,j}(\Sigmam^{{\prime}^{\mu}}_i+\Id_{N_t})^{-1}\xv_{i,i,j} \qquad \forall i,j. \label{eqn:ULscaled}
\end{align}
Note that $\mu_{i,j}$ is independent of the elements of the channel matrix.
%Also note that since $\mu_{i,j}$ is the solution to the fixed point equation \eqref{eqn:lam_stil}, it must
%be bounded from above.
Applying Lemma \ref{lem:lemma2} to the quadratic term of $\xv^H_{i,i,j}(\Sigmam^{{\prime}^{\mu}}_i+\Id_{N_t})^{-1}\xv_{i,i,j}$  yields to,
\begin{align}
\mathbb{E} & \LSB \Big| \xv^H_{i,i,j}(\Sigmam^{{\prime}^{\mu}}_i+\Id_{N_t})^{-1}\xv_{i,i,j} - \frac{1}{N_t} \trace(\Sigmam^{{\prime}^{\mu}}_i+\Id_{N_t})^{-1}\Big|^k \RSB \nonumber \\& \qquad \qquad \qquad \qquad \qquad \leq \frac{C_1}{N^\frac{k}{2}_t}  \qquad \forall i,j \label{eqn:bound1}
\end{align}
for $k \geq 2,$  and constant $C_1$ independent of $N_t$ and $K.$
Additionally, from the result of Lemma 6.1, \cite{hachem07}, we have
\begin{align}
\mathbb{E} \LSB \Big|\frac{1}{N_t} \trace(\Sigmam^{{\prime}^{\mu}}_i+\Id_{N_t})^{-1}-\bar{m}_{i}\Big|^k \RSB \leq \frac{C_2}{N^\frac{k}{2}_t}, \ \forall i,j \label{eqn:bound2}
\end{align}
for $k \geq 2$ constant $C_2$ independent of $N_t$ and $K.$
Therefore, from $\eqref{eqn:bound1},$ $\eqref{eqn:bound2}$ and Holder's inequality
$(|x+y|^k \leq 2^{k-1}(|x|^k+|y|^k))$ we conclude that for some constant $C_3,$
\begin{align}
\mathbb{E} \LSB \Big| \xv^H_{i,i,j}(\Sigmam^{{\prime}^{\mu}}_i+\Id_{N_t})^{-1}\xv_{i,i,j} - \bar{m}_{i} \Big|^k\RSB \leq \frac{C_3}{N^\frac{k}{2}_t}, \ \forall i,j. \label{eqn:bound3}
\end{align}
From \eqref{eqn:ULscaled}, we have, $\xv^H_{i,i,j}(\Sigmam^{{\prime}^{\mu}}_i+\Id_{N_t})^{-1}\xv_{i,i,j} = \frac{\Lambda_{i,j}(\mu_{i,j})}{\sigma_{i,i,j} \mu_{i,j}}.$ Moreover, at the convergence of the fixed point equation \eqref{eqn:lam_stil}, $\bar{m}_i = \frac{\gamma_{i,j}}{\sigma_{i,i,j} \mu_{i,j}}.$ Substituting in \eqref{eqn:bound3}
we have,
\begin{align}
\mathbb{E} \LSB \Big|\Lambda_{i,j}(\mu_{i,j})-\gamma_{i,j}\Big|^k \RSB \leq \frac{C_4}{N_t^\frac{k}{2}}, \ \forall i,j \label{eqn:jaja555}
\end{align}
where $C_4 = C_3 \sigma_{i,i,j} \mu_{i,j}.$ Note that $\sigma_{i,i,j}$ is bounded. 
Moreover for any target SINR $\gammav$ 
satisfying the feasibility conditions of \eqref{eqn:feas_cond_final}, $\limsup_{N_t,K \to \infty} \mu_{i,j}$ is bounded.  
Therefore $\limsup_{N_t,K \to \infty} C_4$ is bounded as well.
%being the solution to fixed point equation of \eqref{eqn:lam_stil} is bounded as well
In order to prove convergence results, we examine the properties of the supremum
over all the indices $i,j.$
\begin{align}
 \mathbb{E} \LSB \sup_{i,j} \Big|\Lambda_{i,j}(\mu_{i,j})-\gamma_{i,j}\Big|^k \RSB & \stackrel{(a)}{\leq}  \sum_{i,j} \mathbb{E} \LSB \Big|\Lambda_{i,j}(\mu_{i,j})-\gamma_{i,j}\Big|^k \RSB  \nonumber \\ 
 & \stackrel{(b)}{\leq} NK \frac{C_4}{N_t^\frac{k}{2}} =  \frac{C_5}{N_t^{\frac{k}{2}-1}} \label{eqn:junk}
\end{align}
where $C_5 = C_4 N \beta$ (where $\beta = \frac{K}{N_t},$ a finite value).
Inequality $(a)$ follows from the linearity of expectation operation and $(b)$ follows from
the bound in \eqref{eqn:jaja555}.
Finally, we make use of the following inequality (which can be shown easily, details omitted here)
\begin{align}   
   \sup_{i,j} \mathbb{E} \LSB |\Lambda_{i,j}(\mu_{i,j})-\gamma_{i,j}|^k \RSB & \leq  \mathbb{E} \LSB \sup_{i,j} |\Lambda_{i,j}(\mu_{i,j})-\gamma_{i,j}|^k \RSB \label{eqn:ineq1}.
\end{align}
From $\eqref{eqn:junk}$ and $\eqref{eqn:ineq1}$ we deduce,
\begin{align}
\sup_{i,j} \mathbb{E} \LSB |\Lambda_{i,j}(\mu_{i,j})-\gamma_{i,j}|^k \RSB \leq \frac{C_5}{N_t^{\frac{k}{2}-1}}. \label{eqn:haha2112}
\end{align}
By taking $k$ to be sufficiently high ($k \geq 6$ in this case), the right hand side of \eqref{eqn:haha2112}
is summable.
By Markov Inequality ($(5.31)$ of \cite{billing95}) and the Borel Cantelli lemma (Theorem $4.3$ of \cite{billing95}),
it follows that
%\begin{align}
$\Lambda_{i,j}(\mu_{i,j})-\gamma_{i,j} \xrightarrow[N_t,K \to \infty]{a.s.} 0 , \ \forall i,j.$

\section*{Appendix E: Convergence proof of the Downlink Interference}
We only focus on the convergence of inter-cell interference term
(the convergence of the intra-cell interference follows 
in a similar manner).
Throughout this section, we denote
$\Phi_n = \sum^N_{b = 1} \sum^K_{k = 1} \frac{\mu_{b,k}}{N_t} \hv_{n,b,k} \hv^H_{n,b,k} + \Id.$
Consider the inter-cell interference term
$ \sum_{k}|\vv^H_{n,k}\hv_{n,i,j}|^2 \ (n \neq i).$
Substituting for $\vv_{n,k} = \frac{\sqrt{\bar{\delta}_{n,k}}}{N_t} \Phi_n^{-1} \hv_{n,n,k},$ we obtain,
\begin{align}
\sum_k  |{\vv}^H_{n,k} & \hv_{n,i,j}|^2  =
\sum_k \frac{\bar{\delta}_{n,k}}{N^2_t} \hv^H_{n,i,j} \Phi_n^{-1} \hv_{n,n,k} \hv^H_{n,n,k}\Phi_n^{-1}\hv_{n,i,j}.
\end{align}
Performing a change of variable
$\xv_{n,i,j}  = \frac{1}{\sqrt{N_t}} \hv_{n,i,j}$
and $\xv_{n,n,k}  = \frac{1}{\sqrt{N_t}} \hv_{n,n,k},$
we obtain,
\begin{align}
\sum_k  |&{\vv}^H_{n,k}  \hv_{n,i,j}|^2 \nonumber \\ & =   \xv^H_{n,i,j} \Phi_n^{-1} \sum_k \LB \bar{\delta}_{n,k} \xv_{n,n,k} \xv^H_{n,n,k} \RB \Phi_n^{-1}\xv_{n,i,j}.
\end{align}
We use Lemma \ref{lem:lemma1} to remove the column
$\xv_{n,i,j}$ from the matrix $\Phi_n.$
Denoting $\Phi^{\prime}_n = \Phi_n-\mu_{i,j}\xv_{n,i,j} \xv^H_{n,i,j},$ 
we obtain
\begin{align}
\sum_k & |{\vv}^H_{n,k}  \hv_{n,i,j}|^2 \nonumber \\ & = \frac{\xv^H_{n,i,j} {\Phi_n^{\prime}}^{-1} \LB \sum_k  \bar{\delta}_{n,k} \xv_{n,n,k} \xv^H_{n,n,k} \RB {\Phi_n^{\prime}}^{-1}\xv_{n,i,j}}{(1+\mu_{i,j} \xv^H_{n,i,j} {\Phi_n^{\prime}}^{-1} \xv_{n,i,j})^2} \label{eqn:refhere001}.
\end{align}
First, it can be easily verified that the denominator term in \eqref{eqn:refhere001} converges to
\begin{align}
(1+\mu_{i,j}\xv^H_{n,i,j} {\Phi_n^{\prime}}^{-1} \xv_{n,i,j} )^2
\asymp (1+\mu_{i,j}\sigma_{n,i,j}\bar{m}_n)^2 \label{eqn:finres1}.
\end{align}
We now focus on the terms of the numerator 
of \eqref{eqn:refhere001}. 
Let us denote the matrix
$\Am = {\Phi_n^{\prime}}^{-1} \LB \sum_k  \bar{\delta}_{n,k} \xv_{n,n,k} \xv^H_{n,n,k} \RB {\Phi_n^{\prime}}^{-1}.$
Applying Lemma \ref{lem:lemma2} on the term $\xv^H_{n,i,j} \Am \xv_{n,i,j}$, we have
\begin{align}
\left| \xv^H_{n,i,j} \Am \xv_{n,i,j} - \frac{\sigma_{n,i,j}}{N_t} \trace({\Am}) \right|^m \leq \frac{C_m}{N^{\frac{m}{2}}_t} ||\Am||^m.
\end{align}
If $||\Am|| < \infty,$ then we
can conclude that 
\begin{align}
\xv^H_{n,i,j} \Am \xv_{n,i,j} \asymp \frac{\sigma_{n,i,j}}{N_t} \trace({\Am})
\label{eqn:finres2}.
\end{align}
The proof of the fact that $||\Am||$ is bounded 
is deferred till the end of this section (see Subsection D at the end of this appendix).

Let us now focus on the term $\frac{1}{N_t}\trace(\Am).$
First, we perform some straightforward manipulations.
\begin{align}
\frac{1}{N_t}\trace(\Am) & = \frac{1}{N_t}\trace \LB {\Phi_n^{\prime}}^{-1} \LB \sum_k  \bar{\delta}_{n,k} \xv_{n,n,k} \xv^H_{n,n,k} \RB{\Phi_n^{\prime}}^{-1} \RB \nonumber \\
& \stackrel{(a)}{=} \frac{1}{N_t}\trace \LB  \LB \sum_k  \bar{\delta}_{n,k} \xv_{n,n,k} \xv^H_{n,n,k} \RB{\Phi_n^{\prime}}^{-2} \RB \nonumber \\
& \stackrel{(b)}{=} \frac{1}{N_t} \sum_k  \bar{\delta}_{n,k}  \LB  \xv^H_{n,n,k} {\Phi_n^{\prime}}^{-2} \xv_{n,n,k} \RB, \label{eqn:intbd1a}
\end{align}
were $(a)$ follows from $\trace(\Am \Bm ) = \trace(\Bm \Am)$
and $(b)$ follows by noting $\trace(\sum_i \Am_i) = \sum_i \trace(\Am_i)$
and $\trace(\xv \xv^H \Am) = \xv^H \Am \xv.$
Applying Lemma \ref{lem:lemma1} to 
extract the column vector $\xv_{n,n,k}$ 
from the matrix ${\Phi^{\prime}_n}^{-1},$
and denoting ${\Phi^{\prime \prime}_n} = {\Phi^{\prime}_n} - \mu_{n,k} \xv_{n,n,k} \xv^H_{n,n,k},$ we obtain,
\begin{align}
\xv^H_{n,n,k} {\Phi_n^{\prime}}^{-2} \xv_{n,n,k} = 
\frac{\xv^H_{n,n,k} {\Phi_n^{\prime \prime}}^{-2} \xv_{n,n,k}}{(1+\mu_{n,k} \xv^H_{n,n,k} {\Phi_n^{\prime}}^{-1} \xv_{n,n,k})^2}. \label{eqn:refhere002}
\end{align}
Note that the denominator term of \eqref{eqn:refhere002}
converges to the following:
\begin{align}
(1+\mu_{n,k} \xv^H_{n,n,k} {\Phi_n^{\prime}}^{-1} \xv_{n,n,k})^2 
\asymp (1+\sigma_{n,n,k}\mu_{n,k} \bar{m}_n)^2 \label{eqn:finres3}.
\end{align}
Let us focus on the numerator term
$\frac{1}{N_t} \sum_k  \bar{\delta}_{n,k}  \LB  \xv^H_{n,n,k} {\Phi_n^{\prime}}^{-2} \xv_{n,n,k} \RB.$
Using the fact that $\frac{K}{N_t}$ is bounded, and applying the result of Theorem \ref{thm:deteqm1}, eq. \eqref{eqn:DerConv}, it can be shown that (the exact details of the derivation are omitted)
\begin{align}
\frac{\sum_k  \bar{\delta}_{n,k}  \LB  \xv_{n,n,k} {\Phi_n^{\prime \prime}}^{-2} \xv^H_{n,n,k} \RB}{N_t}  \asymp \frac{\sum_k \bar{\delta}_{n,k} \sigma_{n,n,k} {\bar{m}}^\prime_n}{N_t} \label{eqn:finres4}.
\end{align}
From the results of \eqref{eqn:refhere001}
%\eqref{eqn:finres1}, \eqref{eqn:finres2}, \eqref{eqn:finres3}, 
and \eqref{eqn:finres4}, it can be concluded that
\begin{align}
& \sum_k |\vv^H_{n,k} \hv_{n,i,j}|^2 
\asymp \sum_k \frac{\bar{\delta}_{n,k}\bar{G}_{n,i,j} \bar{G}_{n,n,k} {\bar{m}}^\prime_n}{N_t}, \ \ n \neq i.
\end{align}

\subsection{Boundedness of $||\Am||$}
We will complete the proof by showing that
$||\Am|| < \infty.$
First, it can be easily noticed that the matrix 
$\Am$ is Hermitian.
We now show that the matrix $\Am$ is also
positive semi-definite.
%\begin{align}
%{\Phi_n^{\prime}}^{-1} \LB \sum_k  \delta_{n,k} \xv_{n,n,k} \xv^H_{n,n,k} \RB {\Phi_n^{\prime}}^{-1}
%\end{align}
For any vector $\gv \in \CC^{N_t \times 1},$
we examine $\gv^H \Am \gv.$
\begin{align}
& \gv^H {\Phi_n^{\prime}}^{-1} \LB \sum_k  \bar{\delta}_{n,k} \xv_{n,n,k} \xv^H_{n,n,k} \RB {\Phi_n^{\prime}}^{-1} \gv \nonumber \\
& = \sum_k  \bar{\delta}_{n,k} \gv^H {\Phi_n^{\prime}}^{-1} \xv_{n,n,k} 
\xv^H_{n,n,k} {\Phi_n^{\prime}}^{-1} \gv
\end{align}
Denoting $\yv = {\Phi_n^{\prime}}^{-1} \gv,$
\begin{align}
\sum_k  \bar{\delta}_{n,k} \gv^H \yv \yv \gv^H 
 = \sum_k  \bar{\delta}_{n,k} |\gv^H \yv|^2 \geq 0
\end{align}
since $\bar{\delta}_{n,k} \geq 0$ and $|\gv^H \yv|^2 \geq 0.$
Therefore, the matrix $\Am$ is semi-definite.
Next, we note that (since $\Am$ is Hermitian and positive semi-definite)
\begin{align}
||\Am|| = \lambda_{\max} (\Am).
\end{align}
Consider $\zv,$ the eigen vector corresponding to the
maximum eigen value of the matrix $\Am.$
\begin{align}
||\Am|| & \leq \zv^H {\Phi_n^{\prime}}^{-1} \LB \sum_k  \bar{\delta}_{n,k} \xv_{n,n,k} \xv^H_{n,n,k} \RB {\Phi_n^{\prime}}^{-1} \zv \nonumber \\
& = \sum_k  \bar{\delta}_{n,k} \zv^H {\Phi_n^{\prime}}^{-1} \xv_{n,n,k} 
\xv^H_{n,n,k} {\Phi_n^{\prime}}^{-1} \zv \nonumber \\
& = \sum_k  \LB \frac{\bar{\delta}_{n,k}}{\mu_{n,k}} \RB \mu_{n,k} \zv^H {\Phi_n^{\prime}}^{-1} \xv_{n,n,k} 
\xv^H_{n,n,k} {\Phi_n^{\prime}}^{-1} \zv \nonumber 
\end{align}
Note that $\beta = \max_{n,k} \frac{\bar{\delta}_{n,k}}{\mu_{n,k}}$
is a bounded value.
Substituting, we have
\begin{align}
||\Am|| & \leq \beta \sum_k   \mu_{n,k} \zv^H {\Phi_n^{\prime}}^{-1} \xv_{n,n,k} 
\xv^H_{n,n,k} {\Phi_n^{\prime}}^{-1} \zv \nonumber \\
& \leq \beta \sum_k   \mu_{n,k} \zv^H {\Phi_n^{\prime}}^{-1} \xv_{n,n,k} 
\xv^H_{n,n,k} {\Phi_n^{\prime}}^{-1} \zv + \beta \zv^H {\Phi_n^{\prime}}^{-2}\zv \nonumber.
\end{align}
where the last inequality follows by noting that
$\beta \zv^H {\Phi_n^{\prime}}^{-2}\zv \geq 0.$
Therefore,
\begin{align}
||\Am||  & \leq \beta \zv^H {\Phi_n^{\prime}}^{-1} \LB \sum_k   \mu_{n,k} \xv_{n,n,k} 
\xv^H_{n,n,k}+ \Id \RB {\Phi_n^{\prime}}^{-1} \zv \nonumber \\
& = \beta \zv^H {\Phi_n^{\prime}}^{-1}  {\Phi_n^{\prime}}  {\Phi_n^{\prime}}^{-1} \zv \nonumber \\
& \leq \beta \zv^H {\Phi_n^{\prime}}^{-1} \zv
\leq \beta \lambda_{\max} ({\Phi_n^{\prime}}^{-1} ) < \infty \label{eqn:BoundSpecNorm}.
\end{align}

\section*{Appendix F: Convergence of the Lagrangian}
\label{sec:Proof1}
First we start with the proof of \eqref{eqn:LagDLConv}.
Consider the Lagrangian given in the form of
\eqref{eqn:LagDL}. For simplicity of notations, let us define
\begin{align}
U_{i,j} & \defines  \frac{|\vv^H_{i,j} h_{i,i,j}|^2}{\gamma_{i,j}} \nonumber \\
I^{(1)}_{i,j} & \defines \sum_{k \neq j} |\vv^H_{i,k} \hv_{i,i,j}|^2, \  
\bar{I}^{(1)}_{i,j}  \defines  \sum_{k \neq i} \frac{\bar{\delta}_{i,k}\bar{G}_{i,i,j} \bar{G}_{i,i,k} {\bar{m}}^\prime_i}{N_t}  \nonumber \\
I^{(2)}_{i,j} & \defines \sum_{n,k} |\vv^H_{n,k} \hv_{n,i,j}|^2, \nonumber \  
\bar{I}^{(2)}_{i,j}  \defines  \sum_{n,k} \frac{\bar{\delta}_{n,k}\bar{G}_{n,i,j} \bar{G}_{n,n,k} {\bar{m}}^\prime_n}{N_t}  \nonumber \\
& \qquad \qquad \qquad \qquad \qquad
n \neq i.
\end{align}
We examine the asymptotic convergence of the terms of
the Lagrangian, i.e.,
$\lim_{N_t,K \to \infty} L(\vv,\mu).$
By using arguments similar to the derivation 
of \eqref{eqn:finres4}, the following can be shown:
\begin{align}
\sum_{i,j} & \frac{\mu_{i,j}}{N_t} \LSB \frac{U_{i,j}}{\gamma_{i,j}}-I^{(1)}_{i,j}-I^{(2)}_{i,j} - N_0\RSB \asymp 
\nonumber \\ & \sum_{i,j} \frac{\mu_{i,j}}{N_t} \LSB \frac{\bar{U}_{i,j}}{\gamma_{i,j}}-\bar{I}^{(1)}_{i,j}-\bar{I}^{(2)}_{i,j} -N_0\RSB \label{eqn:refhere0014}.
\end{align}
Therefore,
\begin{align}
\lim_{N_t,K \to \infty} & L^{\text{DL}}(\vv,\mu) 
 \asymp \lim_{N_t,K \to \infty} \sum_{i,j}\vv^H_{i,j} \vv_{i,j} \nonumber \\
& +  \sum_{i,j} \frac{\mu_{i,j}}{N_t} \LSB \frac{\bar{U}_{i,j}}{\gamma_{i,j}}-\bar{I}^{(1)}_{i,j}-\bar{I}^{(2)}_{i,j}-N_0\RSB.
\end{align}
Following the argument of \eqref{eqn:lin_eq}, we note that
\begin{align}
\frac{\bar{U}_{i,j}}{\gamma_{i,j}}-\bar{I}^{(1)}_{i,j}-\bar{I}^{(2)}_{i,j} -N_0 = 0 \qquad \forall i,j \label{eqn:refhere0015}.
\end{align}
From \eqref{eqn:refhere0014} and \eqref{eqn:refhere0015}, we conclude that 
\begin{align}
\lim_{N_t,K \to \infty} L^{\text{DL}}&(\vv,\mu) \asymp
\lim_{N_t,K \to \infty} \sum_{i,j}\vv^H_{i,j} \vv_{i,j}.
\end{align}
This completes the proof of \eqref{eqn:LagDLConv}.

Next, we proceed to the proof of \eqref{eqn:LagULConv}.
Consider the Lagrangian in the form of \eqref{eqn:LagUL}.
Using \eqref{eqn:refhere0011}, we can rewrite the term
$\sum_{i,j}\vv^H_{i,j} \Bm_{i,j} \vv_{i,j}$ as
\begin{align}
\sum_{i,j}  \vv^H_{i,j} \Bm_{i,j} \vv_{i,j} &  = 
\sum_{i,j} \Big{[} \vv^H_{i,j} \vv_{i,j} - \frac{\mu_{i,j}}{\gamma_{i,j} N_t} |\vv^H_{i,j} \hv_{i,i,j}|^2 \nonumber \\
& + \sum_{(n,k) \neq (i,j)}  \frac{\mu_{n,k}}{N_t}
|\vv^H_{i,j} \hv_{i,n,k}|^2 \Big{]} \label{eqn:refhere0012}. 
\end{align}
Substituting for $\vv_{i,j} = \sqrt{\frac{\bar{\delta}_{i,j}}{N_t}} \hat{\vv}_{i,j}$ in \eqref{eqn:refhere0012}, we obtain
\begin{align}
& \sum_{i,j}  \vv^H_{i,j} \Bm_{i,j} \vv_{i,j}   = 
\sum_{i,j} \frac{\bar{\delta}_{i,j}}{N_t} \Big{[} \hat{\vv}^H_{i,j} \hat{\vv}_{i,j} \nonumber \\ & \ - \frac{\mu_{i,j}}{\gamma_{i,j} N_t} |\hat{\vv}^H_{i,j} \hv_{i,i,j}|^2 
+ \sum_{(n,k) \neq (i,j)}  \frac{\mu_{n,k}}{N_t}
|\hat{\vv}^H_{i,j} \hv_{i,n,k}|^2  \Big{]} \label{eqn:refhere0013}.
\end{align}
We now examine the asymptotic convergence of
the terms on the right hand side of \eqref{eqn:refhere0013}.
The following convergence results can be shown (the details are
omitted here as they follow from steps similar to the proof of previous results).
\begin{align}
& \hat{\vv}^H_{i,j}\hat{\vv}_{i,j}  \asymp   \frac{\sigma_{i,i,j} \bar{m}^\prime_{i}}{(1+\mu_{i,j}\sigma_{i,i,j}\bar{m}_{i})^2 } 
\label{eqn:here1006}.   \\
& \frac{\mu_{i,j}}{N_t}|\hat{\vv}^H_{i,j}\hv_{i,i,j}|^2  \asymp \mu_{i,j} \LB \frac{\sigma_{i,i,j}\bar{m}_{i}}{1+\mu_{i,j}\sigma_{i,i,j}\bar{m}_{i}} \RB^2 
\label{eqn:here1003}.   \\
& \sum_{(n,k) \neq (i,j)} \frac{\mu_{n,k}}{N_t}  |\hat{\vv}^H_{i,j}\hv_{i,n,k}|^2  \asymp \nonumber \\ & \sum_{(n,k) \neq (i,j)}\frac{\mu_{n,k}\sigma_{i,i,j} \sigma_{i,n,k}\bar{m}^\prime_{i}}{N_t(1+\mu_{i,j}\sigma_{i,i,j}\bar{m}_{i})^2 (1+\mu_{n,k}\sigma_{i,n,k}\bar{m}_{i})^2} .   
\end{align}
From \eqref{eqn:here1003} and using the fact that $\gamma_{i,i,j} = \sigma_{i,i,j} \mu_{i,j}\bar{m}_i,$
it follows that
\begin{align}
\frac{\mu_{i,j}}{\gamma_{i,j}N^2_t}|\hat{\vv}^H_{i,j}\hv_{i,i,j}|^2 \asymp   \frac{\sigma_{i,i,j}\bar{m}_{i}}{(1+\mu_{i,j}\sigma_{i,i,j}\bar{m}_{i})^2} .
\label{eqn:here1003a}.   
\end{align}
Additionally, similar to the proof of \eqref{eqn:refhere0014},
it can be shown that
\begin{align}
& \sum_{i,j} \frac{\bar{\delta}_{i,j}}{N_t} \Big{[}\frac{1}{N_t} \hat{\vv}^H_{i,j} \hat{\vv}_{i,j} - \frac{\mu_{i,j}}{\gamma_{i,j} N^2_t} |\hat{\vv}^H_{i,j} \hv_{i,i,j}|^2  \nonumber \\
&  \sum_{(n,k) \neq (i,j)}  \frac{\mu_{n,k}}{N^2_t}
|\hat{\vv}^H_{i,j} \hv_{i,n,k}|^2  \Big{]} \asymp \nonumber \\
&\qquad  \sum_{i,j} \frac{\bar{\delta}_{i,j}}{N_t} \Big{[} \frac{\sigma_{i,i,j} \bar{m}^\prime_{i}}{(1+\mu_{i,j}\sigma_{i,i,j}\bar{m}_{i})^2 }   
 -    \frac{\sigma_{i,i,j}\bar{m}_{i}}{(1+\mu_{i,j}\sigma_{i,i,j}\bar{m}_{i})^2}  \nonumber \\
&\qquad + \sum_{(n,k) \neq (i,j)}\frac{\mu_{n,k}\sigma_{i,i,j} \sigma_{i,n,k}\bar{m}^\prime_{i}}{N_t(1+\mu_{i,j}\sigma_{i,i,j}\bar{m}_{i})^2 (1+\mu_{n,k}\sigma_{i,n,k}\bar{m}_{i})^2} \Big{]}.
\end{align}
Therefore, we conclude that
\begin{align}
& \lim_{N_t,K \to \infty}L^{\text{UL}}(\vv,\mu)   \asymp 
\lim_{N_t,K \to \infty} \sum_{i,j}  \frac{\mu_{i,j} N_0}{N_t} \nonumber \\
&\qquad + \sum_{i,j} \frac{\bar{\delta}_{i,j}}{N_t} \Big{[} \frac{\sigma_{i,i,j} \bar{m}^\prime_{i}}{(1+\mu_{i,j}\sigma_{i,i,j}\bar{m}_{i})^2 }   
 -    \frac{\sigma_{i,i,j}\bar{m}_{i}}{(1+\mu_{i,j}\sigma_{i,i,j}\bar{m}_{i})^2}  \nonumber \\
&\qquad +\sum_{(n,k) \neq (i,j)}\frac{\mu_{n,k}\sigma_{i,i,j} \sigma_{i,n,k}\bar{m}^\prime_{i}}{N_t(1+\mu_{i,j}\sigma_{i,i,j}\bar{m}_{i})^2 (1+\mu_{n,k}\sigma_{i,n,k}\bar{m}_{i})^2} \Big{]}.
\label{eqn:uplink1}
\end{align}
Finally, we show the following result:
\begin{align}
& \frac{\sigma_{i,i,j} \bar{m}^\prime_{i}}{(1+\mu_{i,j}\sigma_{i,i,j}\bar{m}_{i})^2 }   
 -    \frac{\sigma_{i,i,j}\bar{m}_{i}}{(1+\mu_{i,j}\sigma_{i,i,j}\bar{m}_{i})^2}  \nonumber \\
&\ \ + \sum_{(n,k) \neq (i,j)}\frac{\mu_{n,k}\sigma_{i,i,j} \sigma_{i,n,k}\bar{m}^\prime_{i}}{N_t(1+\mu_{i,j}\sigma_{i,i,j}\bar{m}_{i})^2 (1+\mu_{n,k}\sigma_{i,n,k}\bar{m}_{i})^2} = 0
\label{eqn:uplink0}.
\end{align}
The proof of this result is provided next. Consider the
sum of the first and last terms of \eqref{eqn:uplink0}.
\begin{align}
 & \sum_{n,k }\frac{\mu_{n,k}\sigma_{i,i,j} \sigma_{i,n,k}\bar{m}^\prime_{i}}{N_t(1+\mu_{i,j}\sigma_{i,i,j}\bar{m}_{i})^2 (1+\mu_{n,k}\sigma_{i,n,k}\bar{m}_{i})^2} \nonumber \\
&  \qquad \qquad \qquad + \frac{\sigma_{i,i,j} \bar{m}^\prime_{i}}{(1+\mu_{i,j}\sigma_{i,i,j}\bar{m}_{i})^2 }  \nonumber \\
& = \frac{\sigma_{i,i,j} \bar{m}^\prime_{i}}{(1+\mu_{i,j}\sigma_{i,i,j}\bar{m}_{i})^2 } \LB 1+ \frac{1}{N_t} \sum_{n,k}\frac{\mu_{n,k}\sigma_{i,n,k} }{(1+\mu_{n,k}\sigma_{i,n,k}\bar{m}_{i})^2 } \RB.
\label{eqn:refhere0018}
\end{align}
Multiplying and dividing the right hand side of \eqref{eqn:refhere0018} by $\bar{m}_i,$ 
we obtain
\begin{align}
& \eqref{eqn:refhere0018} = \nonumber \\ & \frac{\sigma_{i,i,j} \bar{m}^\prime_{i}}{\bar{m}_i (1+\mu_{i,j}\sigma_{i,i,j}\bar{m}_{i})^2 } \LB \bar{m}_i+ \frac{1}{N_t} \sum_{n,k}\frac{\mu_{n,k}\sigma_{i,n,k} \bar{m}_i }{(1+\mu_{n,k}\sigma_{i,n,k}\bar{m}_{i})^2 } \RB. \label{eqn:here1001}
\end{align}
Recall the fixed point equation for the computation of $\bar{m}_i$ in \eqref{eqn:FPeqn}. Upon rearranging the terms
\begin{align}
 \bar{m}_i  = 1  - \frac{1}{N_t} \sum_{n,k} \frac{\sigma_{i,n,k} \mu_{n,k} \bar{m}_i}{1+\sigma_{i,n,k}\mu_{n,k} \bar{m}_i} .
\end{align}
Using $\bar{m}_i$ for the terms inside the brackets of
the right hand side of \eqref{eqn:here1001}, we have,
\begin{align}
 1-\frac{1}{N_t} & \sum_{n,k} \frac{\sigma_{i,n,k} \mu_{n,k} \bar{m}_i}{1+\sigma_{i,n,k}\mu_{n,k} \bar{m}_i}+ \frac{1}{N_t} \sum_{n,k}\frac{\mu_{m,n}\sigma_{i,n,k} \bar{m}_i }{(1+\mu_{n,k}\sigma_{i,n,k}\bar{m}_{i})^2 } \nonumber \\
= &  1-\frac{1}{N_t} \sum_{n,k} \frac{(\sigma_{i,n,k} \mu_{n,k} \bar{m}_i)^2}{(1+\sigma_{i,n,k}\mu_{n,k} \bar{m}_i)^2} \label{eqn:refhere0020}.
\end{align}
Rearranging the expression for $\bar{m}^\prime_i$ in 
\eqref{eqn:stilderivative}, we have,
\begin{align}
\frac{\bar{m}^2_i}{\bar{m}^\prime_i} = 1-\frac{1}{N_t}
\sum_{n,k} \LB \frac{\sigma_{i,n,k} \mu_{n,k} \bar{m}_i}{1+\sigma_{i,n,k}\mu_{n,k} \bar{m}_i} \RB^2. \label{eqn:refhere0019}
\end{align}
Using \eqref{eqn:refhere0020} and \eqref{eqn:refhere0019} in \eqref{eqn:here1001}, it can be verified that
\begin{align} 
& \eqref{eqn:here1001} =  \frac{\sigma_{i,i,j} \bar{m}_i}{(1+\sigma_{i,n,k}\mu_{n,k} \bar{m}_i)^2}. \label{eqn:here1005}
\end{align}
From \eqref{eqn:here1005}, the result of 
\eqref{eqn:uplink0} follows.

Finally, using \eqref{eqn:uplink0} in \eqref{eqn:uplink1}
we conclude that 
\begin{align}
& \lim_{N_t,K \to \infty}L^{\text{UL}}(\vv,\mu)   \asymp 
\lim_{N_t,K \to \infty} \sum_{i,j}  \frac{\mu_{i,j} N_0}{N_t}.
\end{align}
This completes the proof of \eqref{eqn:LagULConv}.

\section*{Appendix G: Convergence of uplink power allocation}
The proof proceeds by an approach similar to \cite{Randa2010}
of finding the upper and lower bounds of the optimal solution and show that these bounds coincides asymptotically with the solution obtained by the ROBF algorithm.
Recall the fixed point equation for the computation
of the uplink power allocation in the CBF algorithm in
\eqref{eqn:lam_org}. It can be rewritten as
\begin{align} 
f_{i,j}(\betav) = \frac{1}{\frac{1}{\gamma_{i,j} N_t}\hv^H_{i,i,j}(\Sigmam^{{\prime}^\beta}_i+\Id_{N_t})^{-1}\hv_{i,i,j}} \ \forall i,j \label{eqn:FP_Colremoved}
\end{align}
Also recall from the arguments of \cite{Dahrouj2010}
that $f_{i,j}(\betav)$ is a standard function.

Let us denote $\mu_{i,j}(\delta_{i,j}),$ as the solution 
provided by the ROBF algorithm with target SINR 
$\gamma_{i,j}+\delta_{i,j} \ \forall i,j,$ 
where $\delta_{i,j} \geq 0$ is a small positive constant
(i.e. the solution provided by the fixed point equation $\eqref{eqn:lam_stil}$ with $\gamma_{i,j}$ replaced by 
$\gamma_{i,j}+\delta_{i,j}, \ \forall i,j).$ 
We now examine the achieved SINR in the uplink 
with a power allocation of $\mu_{i,j}(\delta_{i,j}),$
i.e.,
\begin{align}
\Lambda_{i,j} & (\muv(\deltav)) \nonumber \\ & = \frac{\mu_{i,j}(\delta_{i,j})}{N_t} \hv^H_{i,i,j}(\Sigmam^{{\prime}^{\mu_{i,j}(\delta_{i,j})}}_i+\alpha_i\Id_{N_t})^{-1}\hv_{i,i,j}. \label{eqn:refimmd}
\end{align}
Similar to the result of Theorem \ref{thm:ULDLSINRConv} (convergence of the uplink SINR), it can be proved that
\begin{align}
\Lambda_{i,j}(\mu_{i,j}(\delta_{i,j}))-(\gamma_{i,j}+\delta_{i,j}) \xrightarrow[N_t,K \to \infty]{a.s.} 0 , \ \forall i,j.
\end{align}
%Therefore, we conclude that with uplink power allocation given by $\mu_{i,j}(\delta_{i,j}),$ the achieved  SINR in the uplink is almost surely equal to $\gamma_{i,j}+\delta_{i,j}$ when the dimensions of the system become large. 
Since $\delta_{i,j} \geq 0$, we have 
\begin{align}
\Lambda_{i,j}(\mu_{i,j}(\delta_{i,j}))-\gamma_{i,j} \stackrel{a.s.}{\geq} 0 , \ \forall i,j \label{eqn:asconv1}
\end{align}
Since the achieved SINR with an uplink power
$\mu_{i,j}(\delta_{i,j})$ allocation is asymptotically greater than $\gamma_{i,j},$ the power allocation
$\mu_{i,j}(\delta_{i,j})$ is a feasible solution to the
uplink problem with target SINR $\gamma_{i,j}.$
From the monotonically property of the standard function 
$f(\betav),$ we can conclude that 
\begin{align}
\mu_{i,j}(\delta_{i,j}) \geq \lambda_{i,j} \qquad \forall i,j.
\end{align}

Similarly, let us define $\mu_{i,j}(-\delta_{i,j})$ as the solution to
the fixed point equation $\eqref{eqn:lam_stil}$ with target SINR
$\gamma_{i,j}-\delta_{i,j}, \ \forall i,j.$
By similar argument as above, we have
\begin{align}
\Lambda_{i,j}(\mu_{i,j}(-\delta_{i,j}))-(\gamma_{i,j}-\delta_{i,j}) \xrightarrow[N_t,K \to \infty]{a.s.} 0 , \ \forall i,j \label{eqn:asconv2}.
\end{align}
Using the definition of almost sure convergence, 
we can write that for $\epsilon_{i,j} > 0,$
there exists $N^\prime$ such that for all
$N_t > N^\prime,$ we can have
\begin{align}
|\Lambda_{i,j}(\mu_{i,j}(-\delta_{i,j}))-(\gamma_{i,j}-\delta_{i,j})| < \epsilon_{i,j}.
\end{align}
If we can select a value of $\epsilon_{i,j} < \delta_{i,j},$
then there exists a large enough $N^{\prime \prime}$
such that for all $N_t > N^{\prime \prime},$
the following
\begin{align}
\frac{\mu_{i,j}(-\delta_{i,j})}{N_t} \hv^H_{i,i,j}(\Sigmam^{{\prime}^{\mu_{i,j}(-\delta_{i,j})}}_i+\Id_{N_t})^{-1}\hv_{i,i,j} \leq \gamma_{i,j}, \ \forall i,j \label{eqn:jaja1141}
\end{align}
holds true.
Equation \eqref{eqn:jaja1141} implies that
\begin{align}
\mu_{i,j}(-\delta_{i,j}) & {\leq} \frac{1}{\frac{1}{\gamma_{i,j} N_t}\hv^H_{i,i,j}(\Sigmam^{{\prime}^{\mu_{i,j}(-\delta_{i,j})}}_i+\Id_{N_t})^{-1}\hv_{i,i,j}} , \ \forall i,j . \nonumber \\
& = f_{i,j}(\muv(-\deltav)) \label{eqn:monotonic}.
\end{align}
In other words, $\muv(-\deltav)$ an infeasible point
for the standard function $f_{i,j}(\muv(-\deltav)).$

Let us consider that we want to find a fixed point
to 
\begin{align}
\beta_{i,j} = f_{i,j}(\betav) \qquad \forall i,j.
\end{align}
The fixed point corresponding to this is the optimal 
uplink power allocation of the CBF algorithm $\lambda_{i,j}.$
We start with an initial value of $\beta^0_{i,j} = \mu_{i,j}(-\delta_{i,j}).$
Let us consider the iterations of $\beta^{t+1}_{i,j} = f_{i,j}(\betav^t),$
with $\betav^0 =  \muv(-\deltav).$ Rewriting the first iteration,
we have,
\begin{align}
\beta^{1}_{i,j} = f(\muv(-\deltav)) \stackrel{(a)}{\geq} \mu_{i,j}(-\delta_{i,j}) \qquad \forall i,j,
\end{align}
where $(a)$ follows from \eqref{eqn:monotonic}.
Also note that this is true for all $i,j.$
Using the fact that the function $f$ is standard
and hence monotonic, the sequence $\beta^t_{i,j}, \ t = 0,1,2,\dots,$ monotonically increases and converges to 
$\lambda_{i,j}.$
Therefore, it follows that
\begin{align}
\mu_{i,j} (-\delta_{i,j}) \leq \lambda_{i,j}.
\end{align}
Therefore, from the above arguments, we conclude that
\begin{align}
\mu_{i,j}(-\delta_{i,j}) \stackrel{a.s.}{\leq}
 \lambda_{i,j} \stackrel{a.s.}{\leq} \mu_{i,j}(\delta_{i,j}) \qquad \forall i,j.
\end{align}
By taking $\delta_{i,j}$ arbitrarily small, we have
\begin{align}
\mu_{i,j}-\lambda_{i,j} \xrightarrow[N_t,K \to \infty]{a.s.} 0 \qquad \forall i,j.
\end{align}
Similar to the proofs as before, 
\begin{align}
\frac{\sum_{i,j}\mu_{i,j}}{N_t}-\frac{\sum_{i,j}\lambda_{i,j}}{N_t} \xrightarrow[N_t,K \to \infty]{a.s.} 0 \qquad \forall i,j.
\end{align}

\section*{Appendix H: Impact of Imperfect CSI and Pilot Contamination}
Recall the uplink receive filer with CSI estimate in 
\eqref{eqn:MMSE}.
First note that for the MMSE estimate of the form 
\eqref{eqn:MMSEEstimate}, $\hat{\hv}_{i,n,k}, \ n \neq i$ and 
$\hat{\hv}_{i,i,k}$ are related as 
\begin{align}
\hat{\hv}_{i,n,k} = \frac{\sigma_{i,n,k}}{\sigma_{i,i,k}} \hat{\hv}_{i,i,k} \label{eqn:Relation}
\end{align}
Using \eqref{eqn:Relation} in \eqref{eqn:PsiDef}, we have
\begin{align}
\Psi_i & = \sum^N_{n = 1} \sum^K_{k = 1}  \mu_{n,k} \hat{\hv}_{i,n,k} \hat{\hv}^H_{i,n,k} + \Id  \nonumber \\
& = \sum^K_{k = 1}  \LB \frac{ \sum^N_{n = 1} \sigma_{i,n,k} \mu_{n,k} } { \sigma_{i,i,k}N_t} \RB \hat{\hv}_{i,i,k} \hat{\hv}^H_{i,i,k} + \Id \nonumber \\
& = \sum^K_{k = 1} \xi_{i,k}  \hat{\xv}_{i,i,k} \hat{\xv}^H_{i,i,k} + \Id
\end{align}
where
\begin{align}
\xi_{i,k} = \frac{ \sum^N_{n = 1} \sigma_{i,n,k} \mu_{n,k} } { \sigma_{i,i,k}}. 
\end{align}
%Therefore, the downlink SINR converges to the expression 
%of \eqref{eqn:DLSINRImpCSIConv} with $\mu_{n,k}$ replaced
%by $\xi_{n,k}.$
Additionally, let us denote
\begin{align}
\Psi_{i,j} & = \sum^K_{k = 1, k \neq j} \xi_{i,k} \hat{\xv}_{i,i,k} \hat{\xv}^H_{i,i,k} + \Id_{N_t}  \nonumber \\ 
\Psi_{i,j,l} & = \sum^K_{k = 1, k \neq j, k \neq l} \xi_{i,k}\hat{\xv}_{i,i,k} \hat{\xv}^H_{i,i,k} + \Id_{N_t}.  
\end{align}

First, we focus on the asymptotic equivalent of the useful signal term.
\begin{align}
{\vv^{\text{est}}_{i,j}}^H \hv_{i,i,j} & = \frac{\sqrt{\bar{\delta}_{i,j}}}{N_t} \hat{\hv}^H_{i,i,j} \Psi^{-1}_i {\hv}_{i,i,j} \nonumber \\
& = {\sqrt{\bar{\delta}_{i,j}}} \hat{\xv}^H_{i,i,j} \Psi^{-1}_i {\xv}_{i,i,j} 
\label{eqn:refimm5}
\end{align}
First note that using Lemma \ref{lem:lemma1}, we can remove the column 
$\hat{\xv}_{i,i,j}$ from the matrix $\Psi^{-1}_i$ as follows:
\begin{align}
\eqref{eqn:refimm5} & =  \frac{\sqrt{\bar{\delta}_{i,j}} \hat{\xv}^H_{i,i,j} \Psi^{-1}_{i,j} {\xv}_{i,i,j}}{1+\xi_{i,j}  \hat{\xv}^H_{i,i,j} \Psi^{-1}_{i,j} \hat{\xv}_{i,i,j}} \nonumber \\
& \asymp \frac{ \sqrt{\bar{\delta}_{i,j}} \hat{\sigma}_{i,i,j}  \bar{m}^{\text{est}}_i}{1+\xi_{i,j} \hat{\sigma}_{i,i,j} \bar{m}^{\text{est}}_i} \label{eqn:impCSIusefulpower},
\end{align}
where $\bar{m}^{\text{est}}_i$ can be computed as in \eqref{eqn:StilImpCSI}.
The result in \eqref{eqn:impCSIusefulpower} can be derived as follows:
First, we look at the numerator term 
\begin{align}
& \hat{\xv}^H_{i,i,j} \Psi^{-1}_{i,j}  \xv_{i,i,j} \nonumber \\ & = {\sigma}^\prime_{i,i,j}\LB \sum^N_{b = 1} \xv_{i,b,j} + \frac{\nv_{\text{Tr}}}{\sqrt{P_{\text{Tr}}}} \RB^H {\Psi^{-1}_{i,j}}  \xv_{i,i,j} \nonumber \\
& \stackrel{(a)}{\asymp} {\sigma}^\prime_{i,i,j} \sigma_{i,i,j} \bar{m}^{\text{est}}_i \stackrel{(b)}{=} \hat{\sigma}_{i,i,j}  \bar{m}^{\text{est}}_i,
\end{align}
where ${(a)}$ follows from the results 
of Lemma \ref{lem:lemma2}, Theorem \ref{thm:deteqm1}, and noting that the column vectors $\xv_{i,b,j} \ \forall \ b \neq i, \nv_{\text{Tr}}$ are independent of $\xv_{i,i,j}.$ Step
${(b)}$ follows using the relation in \eqref{eqn:tilderelation}.
Similarly, the denominator term converges to
\begin{align*}
1+\xi_{i,j}  \hat{\xv}^H_{i,i,j} \Psi^{-1}_{i,j} \hat{\xv}_{i,i,j} \asymp
1+\xi_{i,j} \hat{\sigma}_{i,i,j} \bar{m}^{\text{est}}_i.
\end{align*}

We now investigate the asymptotic equivalent of the sum of interference power
at UT$_{i,j}$ from the BS of cell $n$, given by
\begin{align}
\sum_k  |{{\vv}^{\text{est}}_{n,k}}^{H} & \hv_{n,i,j}|^2   =
\sum_k \frac{\bar{\delta}_{n,k}}{N^2_t} \hv^H_{n,i,j} \Psi_n^{-1} \hat{\hv}_{n,n,k} \hat{\hv}^H_{n,n,k}\Psi_n^{-1}\hv_{n,i,j} \nonumber \\
& =  \sum_k  \bar{\delta}_{n,k} {\xv}_{n,i,j} \Psi_n^{-1}       \hat{\xv}_{n,n,k} \hat{\xv}^H_{n,n,k} \Psi_n^{-1} {\xv}_{n,i,j} \nonumber. 
\end{align}
We note that the CSI estimate of UT$_{i,j}$ is contaminated 
from the pilot signal of UT$_{n,j}, \ (n \neq i).$ Therefore, we 
first analyze the interference signal from UT$_{n,j}, \ (n \neq i),$ 
as follows:
\begin{align}
 & \bar{\delta}_{n,j} {\xv}^H_{n,i,j} \Psi_n^{-1}  \hat{\xv}_{n,n,j} \hat{\xv}^H_{n,n,j} \Psi_n^{-1} {\xv}_{n,i,j} \nonumber \\
& \stackrel{(a)}{=} \frac{ \bar{\delta}_{n,j} \hat{\xv}^H_{n,n,j} \Psi^{-1}_{n,j}  \xv_{n,i,j} \xv^H_{n,i,j}  {\Psi^{-1}_{n,j}} \hat{\xv}_{n,n,j}}{(1+\xi_{n,j} \hat{\xv}^H_{n,n,j} {\Psi^{-1}_{n,j}} \hat{\xv}_{n,n,j})^2} \nonumber \\
& \stackrel{(b)}{\asymp} \frac{  \bar{\delta}_{n,j}  \LB \sigma_{n,i,j} {\sigma}^\prime_{n,n,j}  \bar{m}^{\text{est}}_n \RB^2}{(1+\xi_{n,j}  \hat{\sigma}_{n,n,j} \bar{m}^{\text{est}}_n)^2}
\end{align}
where in $(a),$ we have removed the column $\hat{\xv}^H_{n,n,j}$
from the matrix $\Psi_n^{-1}$ using Lemma \ref{lem:lemma1}, and the result of
$(b)$ is derived similar to \eqref{eqn:impCSIusefulpower}.

Next, we consider the sum of interference signals from rest of
the UTs in cell $n,$ i.e.,
\begin{align}
& \sum_{k \neq j} 
\bar{\delta}_{n,k} {\xv}^H_{n,i,j} \Psi_n^{-1}  \hat{\xv}_{n,n,k} \hat{\xv}^H_{n,n,k} \Psi_n^{-1} {\xv}_{n,i,j} \nonumber \\
& =  {\xv}_{n,i,j} \Psi_n^{-1} \Bm_{n,j} \Psi_n^{-1} {\xv}_{n,i,j},
\label{eqn:refimm101}
\end{align}
where $\Bm_{n,j} = \sum_{k \neq j} \bar{\delta}_{n,k} \hat{\xv}_{n,n,k} \hat{\xv}^H_{n,n,k}.$
Using Lemma \ref{lem:lemma1}, we can decompose $\Psi^{-1}_{n}$ as
\begin{align}
\Psi^{-1}_{n} = \Psi^{-1}_{n,j} - \frac{ \xi_{n,j} \Psi^{-1}_{n,j} \hat{\xv}_{n,n,j} \hat{\xv}^H_{n,n,j} \Psi^{-1}_{n,j} }{1+\xi_{n,j} \hat{\xv}^H_{n,n,j}  \Psi^{-1}_{n,j} \hat{\xv}_{n,n,j}} \label{eqn:refimm11}.
\end{align}
Using \eqref{eqn:refimm11} in \eqref{eqn:refimm101}, we obtain
\begin{align}
& {\xv}^H_{n,i,j} \Psi_n^{-1} \Bm_{n,j} \Psi_n^{-1} {\xv}_{n,i,j} \nonumber \\
& = {\xv}^H_{n,i,j} \Psi^{-1}_{n,j} \Bm_{n,j} \Psi^{-1}_{n,j} {\xv}_{n,i,j} \label{eqn:refimm1}\\
& +  \frac{ \xi^2_{n,j} ({\xv}^H_{n,i,j} \Psi^{-1}_{n,j} \hat{\xv}_{n,n,j})^2 \hat{\xv}^H_{n,n,j} \Psi^{-1}_{n,j} \Bm_{n,j}\Psi^{-1}_{n,j} \hat{\xv}_{n,n,j} }{(1+\xi_{n,j} \hat{\xv}^H_{n,n,j}  \Psi^{-1}_{n,j} \hat{\xv}_{n,n,j})^2}  \label{eqn:refimm2} \\
&  - 2 \text{Re} \Big{\{}  \frac{  \xi_{n,j} ( \hat{\xv}^H_{n,n,j}  \Psi^{-1}_{n,j} {\xv}_{n,i,j}) {\xv}^H_{n,i,j} \Psi^{-1}_{n,j} B_{n,j} \Psi^{-1}_{n,j} \hat{\xv}_{n,n,j} }{1+\xi_{n,j} \hat{\xv}^H_{n,n,j}  \Psi^{-1}_{n,j} \hat{\xv}_{n,n,j}} \Big{\}} . \label{eqn:refimm3}
\end{align}
Consider the term \eqref{eqn:refimm1}. It can be proved that the spectral norm of $\Psi_{n,j}^{-1} \Bm_{n,j} \Psi_{n,j}^{-1}$ is bounded (following the steps in \eqref{eqn:BoundSpecNorm}). Therefore, using 
\begin{align}
{\xv}_{n,i,j} \Psi_{n,j}^{-1} \Bm_{n,j} \Psi_{n,j}^{-1} {\xv}_{n,i,j} 
\asymp \frac{\sigma_{n,i,j}}{N_t} \trace (\Psi_{n,j}^{-1} \Bm_{n,j} \Psi_{n,j}^{-1}).
\end{align}
Using some straightforward steps (similar to the derivation of \eqref{eqn:intbd1a}), it can be shown that
\begin{align}
\frac{1}{N_t}\trace (\Psi_{n,j}^{-1} \Bm_{n,j} & \Psi_{n,j}^{-1})  = \frac{1}{N_t} \sum_{k \neq j} \bar{\delta}_{n,k} \hat{\xv}^H_{n,n,k} \Psi_{n,j}^{-2} \hat{\xv}_{n,n,k}. \label{eqn:refimm103}
\end{align}
The right hand side of \eqref{eqn:refimm103} can be further analyzed as follows:
\begin{align}
\eqref{eqn:refimm103} & = \frac{1}{N_t} \sum_{k \neq j} \frac{\bar{\delta}_{n,k} \hat{\xv}^H_{n,n,k} \Psi_{n,j,k}^{-2} \hat{\xv}_{n,n,k} }{(1+\xi_{n,k} \hat{\xv}^H_{n,n,k} \Psi_{n,j,k}^{-1} \hat{\xv}_{n,n,k} )^2} \nonumber \\
& \asymp \frac{1}{N_t} \sum_{k \neq j} \frac{\bar{\delta}_{n,k} \hat{\sigma}_{n,n,k} (\bar{m}^{\prime}_n)^{\text{est}} }{(1+\xi_{n,k} \hat{\sigma}_{n,n,k} \bar{m}^{\text{est}}_n )^2}, \label{eqn:refimm102}
\end{align}
where $(\bar{m}^{\prime}_n)^{\text{est}}$ can be evaluated as in
\eqref{eqn:stilderImpCSI}, and the result of $\eqref{eqn:refimm102}$ has been derived following the approach of \eqref{eqn:finres4}.
Similarly, it can be shown that
\begin{align}
\eqref{eqn:refimm2} & \asymp \sum_{k \neq j}  \frac{ \bar{\delta}_{n,k} \xi^2_{n,j}  ({\sigma}_{n,i,j}  {\sigma}^\prime_{n,n,j} {\bar{m}^{\text{est}}}_n)^2 \hat{\sigma}_{n,n,j} \hat{\sigma}_{n,n,k} (\bar{m}^{\prime}_n)^{\text{est}} }{N_t (1+\xi_{n,j}  \hat{\sigma}_{n,n,j} \bar{m}^{\text{est}}_n )^2 (1+\xi_{n,k}  \hat{\sigma}_{n,n,k} \bar{m}^{\text{est}}_n )^2}, \nonumber \\
\eqref{eqn:refimm3} & \asymp -2\sum_{k \neq j} \frac{{\sigma}^2_{n,i,j}  ({\sigma}^\prime_{n,n,j})^2 {\bar{m}^{\text{est}}}_n \hat{\sigma}_{n,n,k} (\bar{m}^{\prime}_n)^{\text{est}} }{N_t (1+\xi_{n,j} \hat{\sigma}_{n,n,j} \bar{m}^{\text{est}}_n) (1+\xi_{n,k} \hat{\sigma}_{n,n,k} \bar{m}^{\text{est}}_n)^2}. \nonumber
\end{align}
Combining the convergence results of the useful signal and the interference signal terms derived in this Appendix, we obtain \eqref{eqn:DLSINRImpCSIConv}.

\bibliographystyle{IEEEtran}
\bibliography{IEEEabrv,refs}

\begin{IEEEbiography}
[{\includegraphics[width=1in,height=1.25in,clip,keepaspectratio]{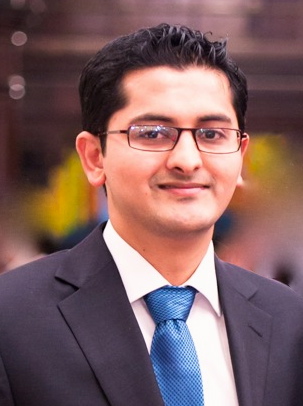}}]
{Subhash Lakshminarayana}
(S'07-M'12)
received his M.S. degree in Electrical
and Computer Engineering from The Ohio State
University in 2009, and his Ph.D. from the
Alcatel Lucent Chair on Flexible Radio
and the Department of Telecommunications at \'Ecole Sup\'erieure d'\'Electricit\'e (Sup\'elec), France in 2012.
Currently he is with The Singapore University of Technology
and Design (SUTD). He has held visiting research appointments at Princeton University from Aug-Dec
2013 and May-Nov 2014. He has also been has been a student researcher at the Indian Institute of Science, Bangalore during
2007.

His work was selected among the ``Best 50 papers" of IEEE Globecom 2014. He has served as TPC member for
several top international conferences.
His research interests broadly spans wireless communication and signal processing with emphasis on small cell networks
(SCNs), cross-layer design wireless networks, MIMO systems, stochastic network optimization, energy harvesting and smart grid systems.
\end{IEEEbiography}

\begin{IEEEbiography}
[{\includegraphics[width=1in,height=1.25in,clip,keepaspectratio]{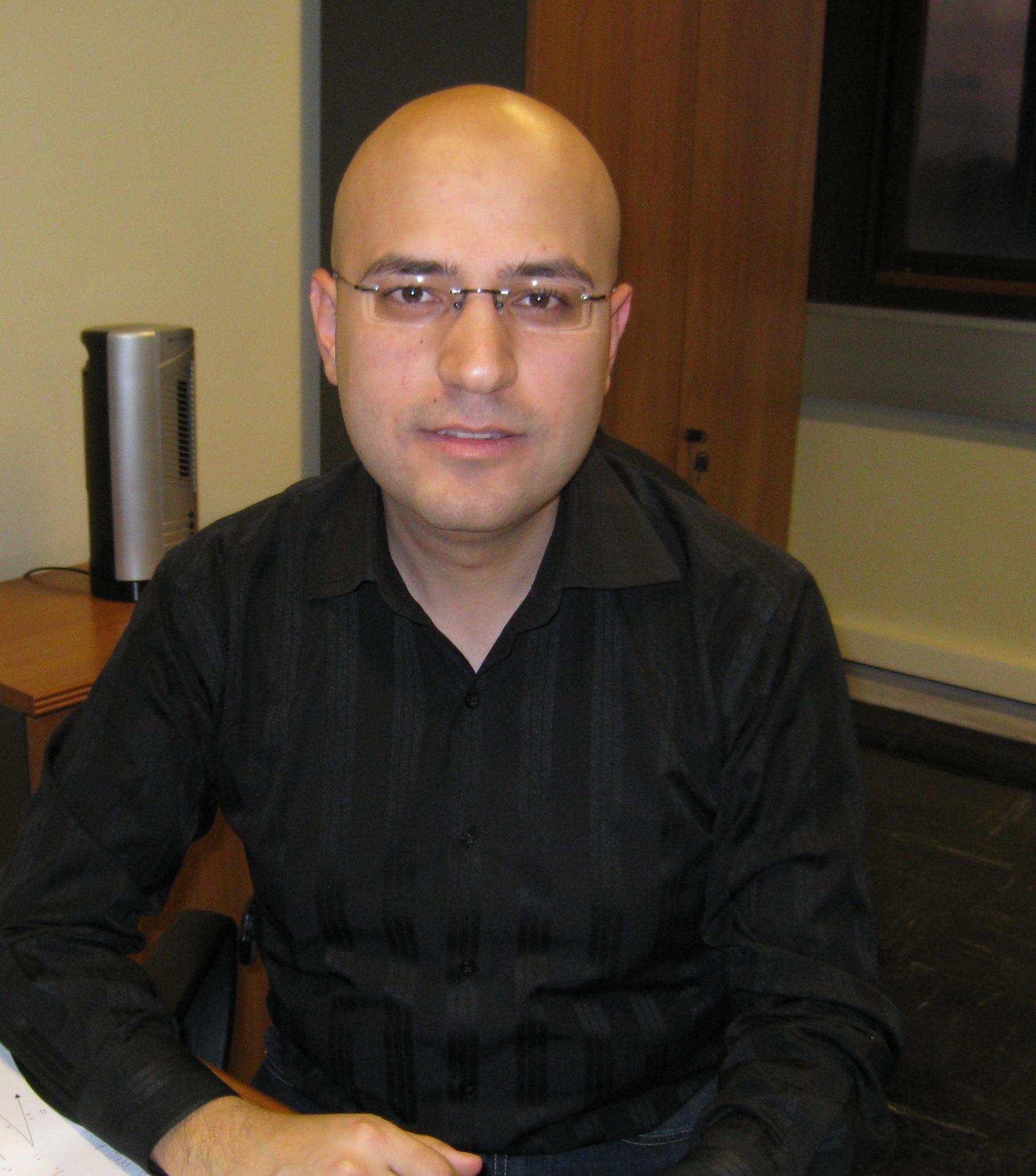}}]
{Mohamad Assaad}
received the BE degree in electrical engineering with high honors from Lebanese University, Beirut, in 2001, and the MSc and PhD degrees (with high honors), both in telecommunications, from Telecom ParisTech, Paris, France, in 2002 and 2006, respectively. Since March 2006, he has been with the Telecommunications Department at CentraleSup\'elec, where he is currently an associate professor. He has co-authored 1 book and more that 75 publications in journals and conference proceedings and serves regularly as TPC member for several top international conferences. His research interests include mathematical models of communication networks, resource optimization and cross-layer design in wireless networks, and stochastic network optimization.

\end{IEEEbiography}

\begin{IEEEbiography}
[{\includegraphics[width=1in,height=1.25in,clip,keepaspectratio]{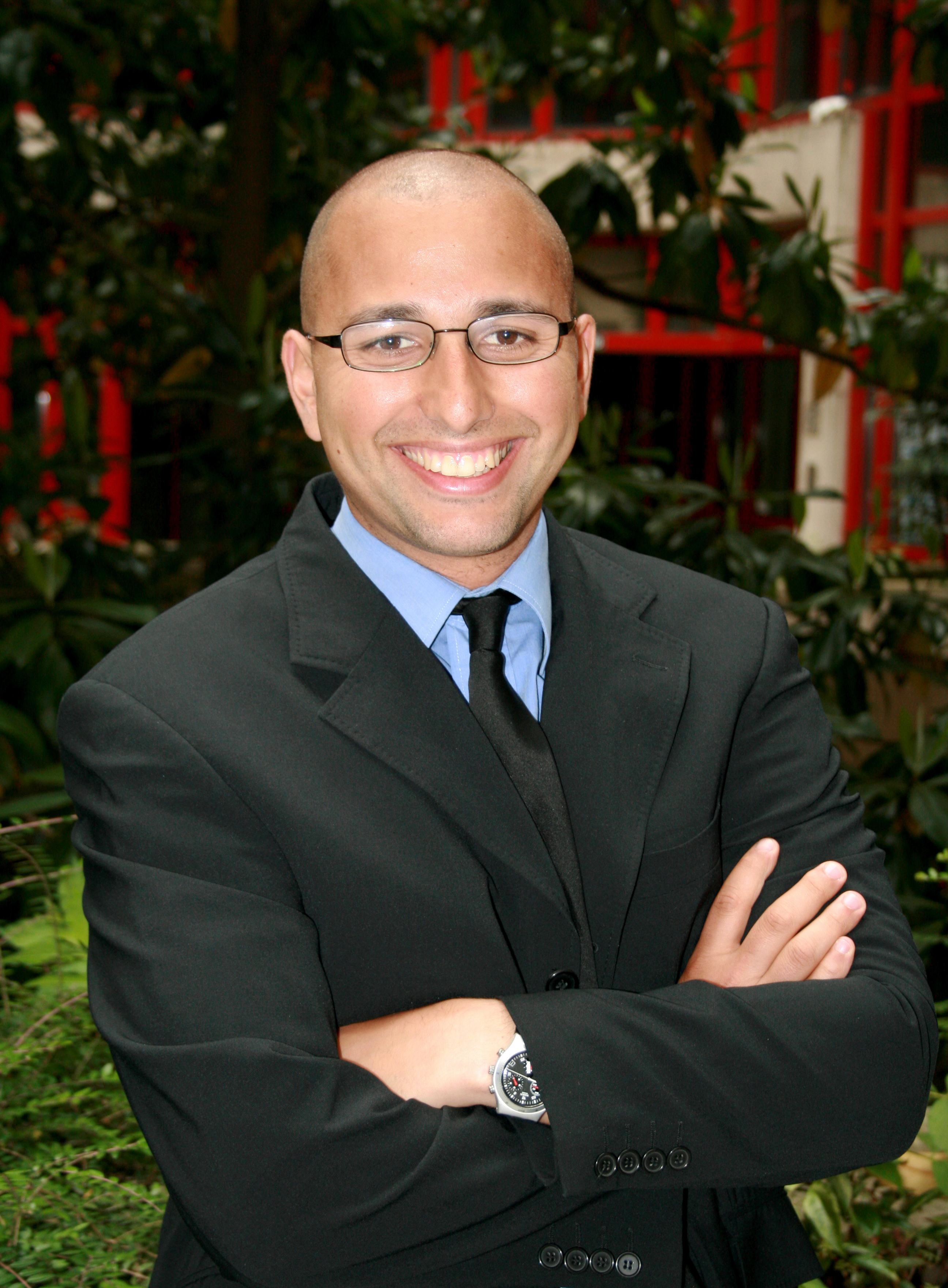}}]
{M\'erouane Debbah}
entered the \'Ecole Normale Sup\'erieure de Cachan (France) in 1996 where he received his M.Sc and Ph.D. degrees respectively. He worked for Motorola Labs (Saclay, France) from 1999-2002 and the Vienna Research Center for Telecommunications (Vienna, Austria) until 2003. From 2003 to 2007, he joined the Mobile Communications department of the Institut Eurecom (Sophia Antipolis, France) as an Assistant Professor. Since 2007, he is a Full Professor at Sup\'elec (Gif-sur-Yvette, France). From 2007 to 2014, he was director of the Alcatel-Lucent Chair on Flexible Radio. Since 2014, he is Vice-President of the Huawei France R\&D center and director of the Mathematical and Algorithmic Sciences Lab. His research interests are in information theory, signal processing and wireless communications. He is an Associate Editor in Chief of the journal Random Matrix: Theory and Applications and was an associate and senior area editor for IEEE Transactions on Signal Processing respectively in 2011-2013 and 2013-2014. M\'erouane Debbah is a recipient of the ERC grant MORE (Advanced Mathematical Tools for Complex Network Engineering). He is a IEEE Fellow, a WWRF Fellow and a member of the academic senate of Paris-Saclay. He is the recipient of the Mario Boella award in 2005, the 2007 IEEE GLOBECOM best paper award, the Wi-Opt 2009 best paper award, the 2010 Newcom++ best paper award, the WUN CogCom Best Paper 2012 and 2013 Award, the 2014 WCNC best paper award as well as the Valuetools 2007, Valuetools 2008, CrownCom2009 , Valuetools 2012 and SAM 2014 best student paper awards. In 2011, he received the IEEE Glavieux Prize Award and in 2012, the Qualcomm Innovation Prize Award.

\end{IEEEbiography}

\end{document}